\documentclass[12pt]{article}
\RequirePackage[colorlinks=true,allcolors=blue,hypertexnames=false]{hyperref}

\usepackage{amsmath}
\usepackage{enumerate}
\usepackage{url} 
\usepackage{amsthm,mathtools,scalerel}
\usepackage{color,soul,latexsym,amsmath,amssymb,amsfonts,amsthm,dsfont,enumitem,xcolor,bbm,threeparttable,graphicx,float}
\usepackage{authblk}
\usepackage[round]{natbib}
\newcommand{\blind}{1}
\usepackage{graphicx, subcaption, mathabx}

\usepackage{setspace}
\usepackage{outlines}
\usepackage{array}
\usepackage{selectp}
\usepackage[ruled,linesnumbered, vlined, noend]{algorithm2e}
\mathtoolsset{showonlyrefs}
\usepackage{titlesec}

\mathtoolsset{showonlyrefs}
\usepackage{titlesec}
\newenvironment{rem}{
   \vskip1mm\indent
   \refstepcounter{myalgctr}
   \textbf{Remark \themyalgctr}
   }{\hfill$\diamond$\par}  

\newcounter{myalgctr}
\numberwithin{myalgctr}{section}
\DeclareMathOperator*{\argmin}{arg\,min}

\newcommand{\E}{\mathbb{E}}
\newcommand{\PP}{\mathbb{P}}

\newtheorem{thm}{Theorem}
\newtheorem{lem}{Lemma}
\newtheorem{cor}{Corollary}
\newtheorem{prop}{Proposition}
\newtheorem{definition}{Definition}

\begin{document}

\def\spacingset#1{\renewcommand{\baselinestretch}%
{#1}\small\normalsize} \spacingset{1}


\if1\blind
{
  \title{\bf Forster--Warmuth Counterfactual Regression: A  Unified Learning Approach}
  \author{Yachong Yang\\
    Department of Statistics and Data Science, University of Pennsylvania\\
    Arun Kumar Kuchibhotla
    \thanks{The authors are grateful for the support by NIH grants R01-AG065276, R01-GM139926, P01-AG041710, R01-CA222147 and NSF-DMS award 2210662.}
    \\
    Department of Statistics and Data Science, Carnegie Mellon University\\
    and\\
    Eric J. Tchetgen Tchetgen\\
    Department of Biostatistics in Biostatistics and Epidemiology,
    Department of Statistics and Data Science, University of Pennsylvania.
    }
  \maketitle
} \fi

\if0\blind
{
  \bigskip
  \bigskip
  \bigskip
  \begin{center}
    {\LARGE\bf Forster--Warmuth Counterfactual Regression: A  Unified Learning Approach}
\end{center}
  \medskip
} \fi



\begin{abstract}
Series regression is a popular non-parametric regression technique obtained by regressing a response on features generated by evaluating basis functions at observed covariate values. The most routinely used series estimator is based on ordinary least squares fitting, which is known to be minimax rate optimal in various settings, albeit under fairly stringent restrictions on the basis functions. Inspired by the recently developed Forster--Warmuth (FW) learner \citep{forster2002relative}, we propose a new series regression estimator that can attain the minimax estimation rate under weaker conditions than existing series estimators in the literature. Moreover, we generalize the FW-learner to a large class of so-called \emph{counterfactual nonparametric regression} problems, in which the response variable of interest may not be directly observed on all sampled units. Although counterfactual regression is not a new area of inquiry, we propose a comprehensive solution to this challenging problem from a unified pseudo-outcome perspective, in the form of a  generic constructive approach for generating a pseudo-outcome which has small bias, namely bias of second order, and attains minimax rate optimality under certain high level conditions.  Several applications in missing data and causal inference are used to illustrate the resulting FW-learner.
\end{abstract}
{\it Keywords:}  Non-parametric regression, Forster--Warmuth estimator, Missing data, Pseudo-outcome 

\spacingset{1.5}
\newpage

\section{Introduction}
\subsection{Nonparametric regression}
Nonparametric estimation plays a central role in many statistical contexts where one wishes to learn conditional distributions by means of say, a conditional mean function $\E[Y|X = x]$ without a priori restriction on the model. Several other functionals of the conditional distribution can likewise be written based on conditional means, which makes the conditional mean an important problem to study. For example, the conditional cumulative distribution function of a univariate response $Y$ given $X = x$ can be written as $\mathbb{E}[\mathbf{1}\{Y \le t\}|X = x].$ This, in turn, leads to conditional quantiles. In general, any conditional function defined via $\theta^\star(x) = \argmin_{\theta\in\mathbb{R}}\mathbb{E}[\rho((X, Y); \theta)|X = x]$ for any loss function $\rho(\cdot; \cdot)$ can be learned using conditional means.

Series, or more broadly, sieve estimation provides a solution by approximating an unknown function based on $k$ basis functions, where $k$ may grow with the sample size $n$, ideally at a rate carefully tuned in order to balance bias and variance. 
The most straightforward approach to construct a series estimator is by the method of least squares, large sample properties of which have been studied extensively both in statistical and econometrics literature in nonparametric settings. To briefly describe the standard least squares series estimator, let $m^\star(x):= \E[Y|X=x]$ denote the true conditional expectation where $m^\star(\cdot)$ is an unrestricted unknown function of $x$. Also consider a vector of approximating basis functions $\widebar{\phi}_k(x) = (\phi_1(x), \ldots, \phi_k(x))^{\top}$, 
which has the property that any square integrable $m^\star(\cdot)$ can be approximated arbitrarily well, with sufficiently large $k$, by some linear combination of $\widebar{\phi}_k(\cdot)$; this idea will be formalized later. Let $(X_i, Y_i), i = 1, \dots, n$ denote an observed sample of data. The least squares series estimator of $m^\star(x)$ is defined as $\widehat{m}(x) = \bar{\phi}_k^\top (x) \widehat{\beta}$, where $\widehat{\beta} = (\Phi_k^\top \Phi_k )^{-1}  \Phi_k^\top \mathbf{Y} $, and $\Phi_k$ is the $n\times k$ matrix $[\widebar{\phi}_k(X_1), \dots, \widebar{\phi}_k(X_n)]^\top$ with $\mathbf{Y} = (Y_1, \dots, Y_n)^\top$. 
Several works in the literature have provided conditions for consistency, corresponding convergence rates, and asymptotic normality of least-squares series estimators, along with settings in which specific basis functions satisfy these conditions, e.g. polynomial series and regression splines; see, for example, \cite{chen2007large}, \cite{newey1997convergence}, \cite{gyorfi2002distribution}. {Under these conditions, the optimal rate of convergence are established for certain basis functions, such as the local polynomial kernel estimator (Chapter 1.6 of \cite{tsybakov2009nonparametric}) and the local polynomial partition series \citep{cattaneo2013optimal}}. 
Recently, \cite{belloni2015some} relaxed certain key assumptions needed for least-squares series estimators to be rate optimal (in probability) for both squared and uniform error norms.  
For instance, they relaxed the requirement in \cite{newey1997convergence} that the number $k$ of approximating functions must satisfy $k^2/n \rightarrow 0$ to $k/n \rightarrow 0$ for bounded (for example Fourier series) or local bases (such as splines, wavelets or local polynomial partition series) to be rate optimal in probability with respect to the squared error norm
, a result which was previously established only for splines \citep{huang2003asymptotics} and local polynomial partitioning estimators \citep{cattaneo2013optimal}.
  An important limitation of the least squares series estimator is that its rate of convergence heavily depends on stringent assumptions imposed on the basis functions. To be specific, a quantity that plays a crucial role in these works, is given by the following quantity which essentially captures the ``size" of basis functions: 
$\xi_k:=\sup _{x \in \mathcal{X}}\|\phi_k(x)\|$, where $\mathcal{X}$ is the support of the covariates $X$ and $ \|\cdot\| $ denote the $l_2$ norm of a vector.
Mainly, they require $\xi_k^2 \log k/n \rightarrow 0$,
so that for basis functions such as Fourier, splines, wavelets, and local polynomial partition series, $\xi_k \leq \sqrt{k}$, yielding $k \log k /n \rightarrow 0$. For other basis functions such as polynomial series, $\xi_k \lesssim k$ corresponds to $k^2 \log k /n \rightarrow 0$, which clearly is more restrictive. 

A main contribution of this paper is to develop a new type of series regression estimator that in principle can attain well-established minimax nonparametric rates of estimation (in mean squared error) in settings where covariates and outcomes are fully observed, under weaker conditions compared to existing literature (e.g. \cite{belloni2015some}) on the distribution of covariates and basis functions. The approach builds on an estimator we refer to as \emph{Forster--Warmuth Learner} (FW-Learner) originating in the online learning literature, which is obtained via a careful modification of the renowned non-linear Vovk--Azoury--Warmuth forecaster~\citep{vovk2001competitive,forster2002relative}.  In particular, our method is optimal in that its mean squared error matches the well-established minimax rate of estimation for a large class of smooth nonparametric regression functions, provided that $\mathbb{E}[Y^2|X]$ is bounded almost surely, regardless of the basis functions used, as long as the approximation error/bias with $k$ bases decays optimally; see Theorem~\ref{thm:forster-full} for more details. This result is more general than the current literature on nonparametric least-squares regression whose rate of convergence depends on the specific ``size" of the basis function being used. For example, \cite{belloni2015some} established that using the polynomials basis can lead to a least-squares series estimator slower convergence rate compared to using a wavelet basis, although both have the same approximation error decay rate for the common H{\"o}lder/Sobolev spaces. Theorem  \ref{thm:forster-full} provides the expected $L_2$-error of our FW-Learner under the full data setting, which is a non-trivial extension of the vanilla Forster--Warmuth estimator and is agnostic to the underlying choice of basis functions. The sharp upper bound on the error rate matches the minimax lower bound of this problem, demonstrating the optimality of the FW-Learner.

\subsection{Counterfactual regression}
Moving beyond the traditional conditional mean estimation problem, an important contribution of this paper is to develop a unified approach to study a more challenging class of problems we name nonparametric {\em counterfactual regression}, where the goal is still to estimate $m^\star(x) = \mathbb{E}[Y|X = x]$ but now the response $Y$ may not be fully/directly observed.


Prominent examples include nonparameric regression of an outcome prone to missingness, a canonical problem in missing data literature, as well as nonparametric estimation of the so-called Conditional Average Treatment Effect (CATE) central to causal inference literature. 
Thus, the key contribution of this work, is to deliver a unified treatment of such counterfactual regression problems with a generic estimation approach which essentially consists of two steps: (i) generate for all units a carefully constructed pseudo-outcome of the counterfactual outcome of interest; 
(ii) apply the FW-Learner directly to the counterfactual pseudo-outcome, in order to obtain an estimator of the counterfactual regression in view. 
The counterfactual pseudo-outcome in step (i) is motivated by modern semiparametric efficiency theory and may be viewed as an element of the orthogonal complement of the nuisance tangent space for the statistical model of the given counterfactual regression problem, see, e.g., \cite{bickel1993efficient}, \cite{van1991differentiable}, \cite{newey1990semiparametric}, \cite{tsiatis2006semiparametric} for some references. As such, the pseudo-outcome endows the FW-Learner with a bias that is at most of second order, which in key settings might be sufficiently small, occasionally it might even be exactly zero, so that it can altogether be ignored without an additional condition. This is in fact the case if the outcome were a priori known to be missing completely at random, such as in some two-stage sampling problems where missingness is by design, e.g. \citep{breslow1988logistic}; or if estimating the CATE in a randomized experiment where the treatment mechanism is known by design. More generally, the pseudo-outcome often requires estimating certain nuisance functions nonparametrically, however, for a large class of such problems considered in this paper, the bias incurred from such estimation is of product form, also known as mixed bias \citep{rotnitzky2021characterization}. In this context, a key advantage of the mixed bias is that one's ability to estimate one of the nuisance functions well, i.e. with relatively fast rates, can potentially make up for slower rates in estimating  another, so that, estimation bias of the pseudo-outcome can potentially be negligible relative to the estimation risk of an oracle with ex ante knowledge of nuisance functions. In such cases, the FW-Learner is said to be {\em oracle optimal} in the sense that its risk matches that of the oracle (up to a multiplicative constant). 

Our main theoretical contribution is a unified analysis of the FW-Learner described above, hereby establishing that it attains the oracle optimality property, under appropriate regularity conditions, in several important counterfactual regression problems, including (1) nonparametric regression under outcome missing at random, (2) nonparametric CATE estimation under unconfoundedness, (3) nonparametric regression under outcome missing not at random leveraging a so-called shadow variable \citep{li2021identification,miao2015identification}, (4) nonparametric CATE estimation in the presence of residual confounding leveraging proxies using the proximal causal inference framework \citep{miao2018identifying,tchetgen2020introduction}. The proposed approach recovers well-established pseudo-outcomes in the literature for certain canonical cases such as the CATE under unconfoundedness; in many other cases, our results appear to be completely new to the literature.   

\subsection{
Organization and notation}
\paragraph{Organization.}The remainder of the paper is organized as follows. Section \ref{sec:notation} introduces the notation that is going to be used throughout the paper. Section \ref{sec:full-pseudo} formally defines our estimation problem and the Forster--Warmuth estimator, where Section \ref{sec:pseudo} builds upon Section~\ref{sec:full} going beyond the full data problem to counterfactual settings where the outcome of interest may not be fully observed. Section \ref{sec:missing} applies the proposed methods to the canonical nonparametric regression problem subject to missing outcome data, where in Section \ref{sec:mar} the outcome is assumed to be Missing At Random (MAR) given fully observed covariates \citep{robins1994estimation}; while in Section \ref{sec:mnar} the outcome may be Missing Not At Random (MNAR) and identification hinges upon having access to a fully observed shadow variable \citep{miao2015identification,li2021identification}. Both of these examples may be viewed as nonparametric counterfactual regression models, whereby one seeks to estimate the nonparametric regression function under a hypothetical intervention that would in principle prevent missing data. Section \ref{sec:cate} presents another application of the proposed methods to a causal inference setting, where the nonparametric counterfactual regression parameter of interest is the Conditional Average Treatment Effect (CATE); Section \ref{sec:unconfoundness} assumes so-called ignorability or unconfoundedness given fully observed covariates, while Section \ref{sec:proximal} accommodates unmeasured confounding under the recently proposed proximal causal inference framework \citep{miao2018identifying,tchetgen2020introduction}. Section \ref{sec:simulations} reports results from a simulation study comparing our proposed FW-Learner to a selective set of existing methods under a range of conditions, while Section \ref{sec:real} illustrates FW-Learner for the CATE in an analysis of the SUPPORT observational study  \citep{conners1996effectiveness} to estimate the causal effect of Right Heart Catheterization (RHC) on 30-day survival, as a function of a continuous baseline covariate which measures a \textit{patient's potential survival probability at hospital admission}, both under standard unconfoundedness conditions assumed in prior causal inference papers, including \cite{tan2006distributional}, \cite{vermeulen2015bias} and \cite{cui2019selective}, and proximal causal inference conditions recently considered in \cite{cui2020semiparametric} in the context of estimating marginal treatment effects.

\subsection{Notation}\label{sec:notation}
We define some notation we use throughout the paper: $a \lesssim b$ means $a \leq C b$ for a universal constant $C$, and $a \sim b$ means $a \lesssim b$ and $b \lesssim a$. 
For any integer $k\geq 2$, let $\|f(\cdot) \|_k$ denote the function $L_k$ norm such that $\|f(O)\|_k := (\E_{O}[|f(O)|^k])^{1/k},$ where $O$ is any data that is the input of $f$.

\section{The Forster--Warmuth Nonparametric Counterfactual Regression Estimator}\label{sec:full-pseudo}
In Section~\ref{sec:full}, we introduce and study the properties of FW-Learner in the standard nonparametric regression setting where data are fully observed, before considering the counterfactual setting of primary interest in Section~\ref{sec:pseudo} where the responses may only be partially observed.

\subsection{Full data nonparametric regression}\label{sec:full}
Suppose that one observes independent and identically distributed random vectors $\left(X_i, Y_i\right), 1 \leq i \leq n$ on $\mathcal{X} \times \mathbb{R}$. Let $\mu$ be the Lebesgue measure on the covariate space $\mathcal{X}$. {\color{black}In fact, it can be any general measure on $\mathcal{X}$.}
The most common nonparametric regression problem aims to infer the conditional mean function $m^\star(x):=\mathbb{E}\left[Y_i \mid X_i=x\right]$ as a function of $x$. 
Let $\Psi := \left\{\phi_1(\cdot) \equiv 1, \phi_2(\cdot), \phi_3(\cdot), \ldots\right\}$ be a fundamental sequence of functions in $L_2(\mu)$ i.e., linear combinations of these functions are dense in $L_2(\mu)$~\citep{lorentz1966metric,yang1999information}. 

For any $f\in L_2(\mu)$ and any $J \ge 1$, let $E_{J}^{\Psi}(f) ~:=~ \min_{a_1, a_2, \ldots, a_{J}}\big\|f - \sum_{k=1}^{J} a_k\phi_k\bigr\|_{L_2(\mu)}$
denote the $J$-th degree approximation error of the function $f$ by the first $J$ functions in $\Psi$. By definition of the fundamental sequence, $E_{J}^{\Psi}(f)\to0$ as $J\to 0$ for any function $f\in L_2(\mu)$.
This fact is the motivation of the traditional series estimators of $m^\star$ which estimate the minimizing coefficients $a_1, \ldots, a_{J}$ using ordinary least squares linear regression. 
Motivated by an estimator in the linear regression setting studied in \cite{forster2002relative}, we define the FW-Learner of $m^\star(\cdot)$, which we denote $\widehat{m}_J(\cdot)$, trained on data $\{ \left(X_i, Y_i\right), 1 \leq i \leq n\},$ using the first $J$ elements of the fundamental sequence $\bar{\phi}_J(x)=\bigl(\phi_1(x), \ldots, \phi_J(x)\bigr)^{\top}$:
\begin{align}\label{def:forster}
\widehat{m}_J(x):= \bigl( 1-h_n(x) \bigr) \bar{\phi}_J^\top (x)  \Bigl(\sum_{i=1}^n \bar{\phi}_J (X_i ) \bar{\phi}_J^{\top}(X_i)+\bar{\phi}_J(x) \bar{\phi}_J^{\top}(x)\Bigr)^{-1} \sum_{i=1}^n \bar{\phi}_J(X_i) Y_i,
\end{align}
where
\begin{equation}\label{eq:leverage-definition}
h_n(x):=\bar{\phi}_J^{\top}(x)\Bigl(\sum_{i=1}^n \bar{\phi}_J (X_i ) \bar\phi^{\top}_J (X_i )+\bar{\phi}_J(x) \bar{\phi}^{\top}_J(x)\Bigr)^{-1} \bar{\phi}_J(x) ~\in~ [0,1] .
\end{equation}
The following result provides a finite-sample result on the estimation error of $\widehat{m}_J$ as a function of $J$.
\begin{thm}\label{thm:forster-full}
Suppose $\E \bigl[Y^2 | X\bigr] \leq \sigma^2$ almost surely $X$ and suppose $X$ has a density with respect to $\mu$ that is upper bounded by $\kappa$. Then the FW-Learner satisfies
$$
\begin{aligned}
\bigl\|\widehat{m}_J - m^\star \bigr\|_2 ^2 =\mathbb{E}\Bigl[\bigl(\widehat{m}_J(X)-m^\star(X)\bigr)^2\Bigr] 
& ~\leq~ \frac{2 \sigma^2 J}{n}+\kappa (E_J^{\Psi}(m^*))^2.
\end{aligned}
$$
Moreover, if $\Gamma = \{\gamma_1, \gamma_2, \ldots\}$ is a non-increasing sequence and if $m^\star \in \mathcal{F}(\Psi, \Gamma) = \{f\in L_2(\mu):\, \E_k^{\Psi}(f) \le \gamma_k, \,\forall\, k\ge1\}$, then for $J_n := \min\{k \ge 1:\, \gamma_k^2 \le \sigma^2k/n\}$, we obtain
$\|\widehat{m}_{J_n} - m^\star\|_2^2 ~\le~ (2 + \kappa)\frac{\sigma^2J_n}{n}.
$
\end{thm}
See Section \ref{sec-app:forster-full} of the supplement for proof of this result.  Note that \citet[Theorem 4.1]{belloni2015some} established a similar result for the least squares series estimator, implying that it yields a related oracle risk under more stringent conditions imposed on the basis functions as discussed in the introduction. Furthermore, their result is in probability and therefore is implied by but does not imply ours. 
The sets of functions $\mathcal{F}(\Psi, \Gamma)$ are called {\em full approximation sets} in~\cite{lorentz1966metric} and~\citet[Section 4]{yang1999information}. If the sequence $\Gamma$ also satisfies the condition $0 < c' \le \gamma_{2k}/\gamma_k \le c \le 1$ for all $k \ge 1$, then Theorem 7 of~\cite{yang1999information} proves that the minimax rate of estimation of functions in $\mathcal{F}(\Psi, \Gamma)$ is given by $k_n/n$, where $k_n$ is chosen so that $\gamma_k^2\asymp k/n$. The upper bound in Theorem~\ref{thm:forster-full} matches this rate under the assumption $c' \le \gamma_{2k}/\gamma_k \le c$. This can be proved as follows: by definition of $J_n$, $\gamma_{J_n - 1}^2 \ge \sigma^2(J_n - 1)/n$. Then using $J_n - 1 \ge J_n/2$ and $\gamma_{J_n - 1} \le \gamma_{J_n/2} \le \gamma_{J_n}/c'$, we get $\gamma_{J_n}^2 \ge (c')^2\sigma^2J_n/(2n)$. Hence, $\gamma_{J_n} \asymp \sigma^2J_n/n$. Therefore, Theorem~\ref{thm:forster-full} proves that the FW-Learner with a properly chosen $J$ is minimax optimal for approximation sets.

A remark is warranted to build intuition about the specific form of the FW-learner. First, note that ignoring the factor $1-h_n(x)$, the FW-learner can be interpreted as a least-squares estimator fitted to the observed sample with the additional artificial observation $(X_{n+1}=x,Y_{n+1}=0)$. Further, note that the term $0\leq h_n(x)\leq 1$ defines the leverage score at a hypothetical observation at $x$ and therefore reflects how far this value lies away from the sample $\{X_i:1\le i\le n\}$. Therefore, for $x$ close to observed values, the factor $1-h_n(x)$ should be close to 1 and the corresponding fitted value should be largely determined by the observed sample. In contrast, for $x$ far from observed values, the least-squares estimator will tend to extrapolate, however  $1-h_n(x)$ controls the degree of extrapolation for such values outside the empirical support of the covariates by shrinking corresponding fitted values towards zero; this is because $1-h_n(x)$ for such $x$ would be close to zero. While this may lead to some bias at such value of $x$, the event is sufficiently rare such that it is negligible over the support of $X$ with respect to mean squared error. Theorem~\ref{thm:forster-full} formalizes this intuition. 

Note that Theorem~\ref{thm:forster-full} does not require the fundamental sequence of functions $\Psi$ to form an orthonormal bases. This is a useful feature when considering sieve-based estimators~\citep[Example 3]{shen1994convergence}, partition-based estimators~\citep{cattaneo2013optimal}, random kitchen sinks~\citep{rahimi2008weighted} or neural networks~\citep{klusowski2018approximation}, just to name a few.

As a special case of Theorem~\ref{thm:forster-full} that is of particular interest for $\alpha_m$-H{\"o}lder or $\alpha_m$-Sobolev spaces {\color{black}(See Section \ref{supp:notation} for definitions)}, suppose $\gamma_J \leq C_m J^{-2 \alpha_m/d}$ for some constant $C_m, \alpha_m>0$, and $d$ is the intrinsic dimension\footnote{We say intrinsic dimension rather than the true dimension of covariates because some bases can take into account of potential manifold structure of the covariates to yield better decay depending on the manifold (or intrinsic) dimension.} of the covariates $X$, then choosing $J=\lceil\left(n \alpha_m \kappa C_m / (d\sigma^2)\right)^{d /(2 \alpha_m+d)}\rceil$ gives
\begin{align}\label{eq:thm1-fourier}
\|\widehat{m}_J - m^\star\|_2 ^2 
&\le   C \left(\frac{\sigma^2}{n}\right)^{2 \alpha_m /(2 \alpha_m+d)}, 
\end{align}
where $C$ is a constant;
see Section~\ref{pf:thm1-fourier} for a proof. The decay condition $\gamma_J \le C_mJ^{-2\alpha_m/d}$ is satisfied by functions in H{\"o}lder class $\Sigma_d(\alpha_m,C_m)$  and Sobolev spaces of order $\alpha_m$ for the classical polynomial, Fourier/trigonometric bases~\citep{devore1993constructive,belloni2015some}
From the discussion above, it is clear that the choice of the number of functions $J$ used is crucial for attaining the minimax rate. In practice, we propose the use of split-sample cross-validation to determine $J$~\citep[Section 7.1]{gyorfi2002distribution}. Our simulations presented in Section~\ref{sec:simulations} shows good performance of such an approach. {\color{black}We refer interested readers to~\citet[Chapter 7]{gyorfi2002distribution} and \cite{vaart2006oracle} for the theoretical properties of the split-sample cross-validation, which requires the estimators to be bounded, see, e.g. Corollary 4.2 of \cite{vaart2006oracle}. Interestingly, the next proposition shows that the FW-Learner automatically satisfies this assumption, 
the proof of which is deferred to Section \ref{app-sec:proof}.

\begin{prop}\label{prop:fw-bound}
The FW-Learner defined in \eqref{def:forster} satisfies that 
\begin{align}\label{eq:fw-bound}
\bigl| \widehat{m}_J(x) \bigr| &\leq h_n(x) \, ( 1-h_n(x) ) \, \Bigl(\frac{1}{n} \sum_{i=1}^n Y_i^2 \Bigr)^{1/2} \leq 1/4\,\Bigl(\frac{1}{n} \sum_{i=1}^n Y_i^2 \Bigr)^{1/2}, \quad\mbox{for all}\quad J\ge1.
\end{align}
Moreover, 
$
\mathbb{E}\bigl[\sup_{J\ge1}\|\widehat{m}_J\|_{\infty}\bigr] \le \frac{1}{4}\bigl(\mathbb{E}[Y^2] \bigr)^{1/2}.
$
\end{prop}
This boundedness property is not satisfied by the ordinary least squares series regression estimator, for arbitrary basis functions. 
}
     

\subsection{Forster--Warmuth Counterfactual Regression: The Pseudo-Outcome Approach}\label{sec:pseudo}
In many practical applications in health and social sciences it is not unusual for an outcome to be missing on some subjects, either by design, say in two-stage sampling studies where the outcome can safely be assumed to be Missing At Random with known non-response mechanism, or by happenstance, in which case the outcome might be missing not at random.

Beyond missing data, counterfactual regression also arises in causal inference where one might be interested in the CATE, the average causal effect experienced by a subset of the population defined in terms of observed covariates. Missing data, in this case, arises as the causal effect defined at the individual level as a difference between two potential outcomes -- one for each treatment value -- can never be observed. 

A major contribution of this paper is to propose a generic construction of a so-called pseudo-outcome which, as its name suggests, replaces the unobserved outcome with a carefully constructed response variable that (i) only depends on the observed data, possibly involving high dimensional nuisance functions (e.g. propensity score) that can nonetheless be identified from the observed data, and therefore can be evaluated for all subjects in the sample and; (ii) has conditional expectation given covariates that matches the counterfactual regression of interest if as for an oracle, nuisance functions were known. The proposed pseudo-outcome approach applies to a large class of counterfactual regression problems including the missing data and causal inference problems described above. The proposed approach recovers in specific cases such as the CATE under unconfoundedness, previously proposed forms of pseudo-outcomes \citep[Section 4.2]{kennedy2023towards},  
while offering new pseudo-outcome constructions in other examples (e.g., Proximal CATE estimation in Section \ref{sec:proximal}). See Section~\ref{subsub:construction-of-pseudo-outcome} for details on constructing pseudo-outcomes.
  
Before describing the construction of the pseudo-outcomes in detail, we first give a key high-level preliminary corollary (assuming that a pseudo-outcome is given) which is the theoretical backbone of our approach.   
Suppose $\widetilde{O}_i, 1\le i\le n$ represent independent and identically distributed random vectors of underlying data of primary interest that include fully observed covariates $X_i, 1\le i\le n$ as subvectors. Let $O_i, 1\le i\le n$ be the observed data which are obtained from $\widetilde{O}_i, 1\le i\le n$ through some coarsening operation. For concrete examples of $\widetilde{O}_i$ and $O_i$ in missing data and causal inference, see Table~\ref{tab:examples-unobserved-observed}; more examples can be found in Sections~\ref{sec:missing} and~\ref{sec:cate}. The quantity of interest is $m^\star(x) = \mathbb{E}[\widetilde{f}(\widetilde{O}_i)|X_i = x]$ for some known function $\widetilde{f}(\cdot)$ operating on (unobserved) $\widetilde{O}_i$. For example, in the context of missing outcome data where $\widetilde{O}_i = (Y_i, X_i, Z_i)$ and $O_i = \widetilde{O}_i$ if $Y_i$ is observed and $O_i = (X_i, Z_i)$ if $Y_i$ is missing, so that in a slight abuse of notation, we write $O_i = (R_i, R_i Y_i, X_i, Z_i)$ where $R_i=1$ if $Y_i$ is observed and $R_i=0$ otherwise.  One could be interested in $\mathbb{E}[Y_i|X_i]$ so that $\widetilde{f}(\widetilde{O}_i) = f(X_i, Z_i, Y_i) = Y_i$. Because $\widetilde{O}_i, 1\le i\le n$ are unobserved, $\widetilde{f}(\widetilde{O}_i)$ may not be fully observed. The pseudo-outcome approach that we propose involves two steps: 
\begin{enumerate}[label=(Step \Alph*), leftmargin=0.8in]
    \item Given a set of identifying conditions such that the quantity of interest $m^\star(x) = \mathbb{E}[\widetilde{f}(\widetilde{O}_i)|X_i = x]$ can be rewritten as $m^\star(x) = \mathbb{E}[f(O_i)|X_i = x]$ for some (estimable) unknown function $f(\cdot)$ applied to the observations $O_i$. There may be several such $f$ under the identifying assumptions and the choice of $f$ plays a crucial role in the rate of convergence of the estimator proposed; see Section~\ref{subsub:construction-of-pseudo-outcome} for more details on finding a ``good'' $f$.\label{step-A-FW-Learner-counterfactual-regression}
    \item Split $\{1, 2, \ldots, n\}$ into two (non-overlapping) parts $\mathcal{I}_1, \mathcal{I}_2$. From $O_i, i\in\mathcal{I}_1$, obtain an estimator $\widehat{f}(\cdot)$ of $f(\cdot)$. Now, with the fundamental sequence of functions $\Psi$, create the data $(\widebar{\phi}_J(X_i), \widehat{f}(O_i)), i\in\mathcal{I}_2$ and obtain the FW-Learner: \label{step-B-FW-Learner-counterfactual}
    \begin{equation}\label{eq:Forster-Warmuth-pseudo-outcome}
\widehat{m}_J(x) := (1 - h_{\mathcal{I}_2}(x))\bar{\phi}_J^{\top}(x)\left(\sum_{i\in\mathcal{I}_2} \bar{\phi}_J(X_i)\bar{\phi}_J^{\top}(X_i) + \bar{\phi}_J(x)\bar{\phi}^{\top}_J (x)\right)^{-1}\sum_{i\in\mathcal{I}_2} \bar{\phi}_J(X_i)\widehat{f}(O_i),
\end{equation}
with
$h_{\mathcal{I}_2}(x) = \widebar{\phi}_J^{\top}(x)\left(\sum_{i\in\mathcal{I}_2} \widebar{\phi}_J(X_i)\widebar{\phi}_J^{\top}(X_i) + \widebar{\phi}_J(x)\widebar{\phi}_J^{\top}(x)\right)^{-1}\widebar{\phi}_J(x),
$
defined, similarly, as in~\eqref{eq:leverage-definition}.
\end{enumerate}
\begin{table}[!h]
\begin{tabular}{lll}\hline
                 & $\widetilde{O}_i$                                                                                                                                                                                                                                                     & $O_i$                                                                                                                                       \\\hline\hline
Missing data     & \begin{tabular}[c]{@{}l@{}}$(X_i, Z_i, Y_i)$\\ $Y_i$ is the response of interest,\\ $Z_i$ is an additional covariate\\ vector of no scientific interest.\end{tabular}                                                                                                 & \begin{tabular}[c]{@{}l@{}}$(X_i, Z_i, R_i, Y_iR_i)$\\ $R_i = 1$ if $Y_i$ is observed,\\ and $R_i = 0$ if $Y_i$ is unobserved.\end{tabular} \\\hline
Causal inference & \begin{tabular}[c]{@{}l@{}}$(X_i, A_i, Y_i^1, Y_i^0)$\\ $A_i$ is the treatment assignment,\\ $Y_i^1$ is the counterfactual response\\ if subject is in treatment group, and\\ $Y_i^{0}$ is the counterfactual response\\ if subject is in control group.\end{tabular} & \begin{tabular}[c]{@{}l@{}}$(X_i, A_i, Y_i)$\\ $Y_i=A_iY_i^1+(1-A_i)Y_i^0$  \\is the observed response \\given the observed treatment $A_i$.\end{tabular} \\ \hline 
\end{tabular}
\caption{Examples of unobserved full data and observed data.}
\label{tab:examples-unobserved-observed}
\end{table}
The following lemma (proved in Section~\ref{sec-app:forster-full}) states the error bound of the FW-Learner $\widehat{m}_J$ that holds for any pseudo-outcome $\widehat{f}$.
\begin{cor}\label{cor:forster-pseudo}
Let $\sigma^2$ be an upper bound on $\mathbb{E}[\widehat{f}^2(O)|X, \widehat{f}]$ almost surely $X$, and suppose $X$ has a density with respect to $\mu$ that is bounded by $\kappa$. Define $H_f(x) = \mathbb{E}[\widehat{f}(O)|X = x, \widehat{f}]$. Then the FW-Learner $\widehat{m}_J$ satisfies
\begin{equation}\label{eq:forster-pseudo}
\left(\mathbb{E}[(\widehat{m}_J(X) - m^\star(X))^2\big|\widehat{f}]\right)^{1/2} 
\le \sqrt{\frac{2\sigma^2 J}{|\mathcal{I}_2|}} + \sqrt{2\kappa} E_J^{\Psi}(m^\star) + \sqrt{6}\left(\mathbb{E}[(H_f(X) - m^\star(X))^2\big|\widehat{f}]\right)^{1/2}.
\end{equation}
\end{cor}

The first two terms of~\eqref{eq:forster-pseudo} are an upper bound on the error of oracle FW-Learner with access to $(X_i, f(O_i)), i\in\mathcal{I}_2$. The last term of \eqref{eq:forster-pseudo}, $H_f - m^\star$, is 
the bias incurred from estimating the oracle pseudo-outcome $f$ with the empirical pseudo-outcome $\widehat{f}$. Here the choice of estimator of the oracle pseudo-outcome is key to rendering this bias term potentially negligible relative to the leading two terms of equation \eqref{eq:forster-pseudo}. We return to this below.

If $|\mathcal{I}_1| = |\mathcal{I}_2| = n/2$, $m^\star\in\mathcal{F}(\Psi, \Gamma)$, the full approximation set discussed in Theorem~\ref{thm:forster-full}, and we set $J = J_n = \min\{k\ge1:\, \gamma_k^2 \le \sigma^2k/n\}$, then Corollary \ref{cor:forster-pseudo} implies that $\|\widehat{m}_J - m^\star\|_2 \le 2(1 + \sqrt{\kappa})\sqrt{{\sigma^2J_n}/{n}} + \sqrt{6}\|H_f - m^\star\|_2.$
Because $\sqrt{J_n/n}$ is the minimax rate in $L_2$-norm for functions in $\mathcal{F}(\Psi, \Gamma)$, we get the FW-Learner with pseudo-outcome $\widehat{f}(O)$ is minimax rate optimal as long as $\|H_f - m^\star\|_2 = O(\sqrt{J_n/n})$. In such a case, we call $\widehat{m}_J$ {\em oracle minimax} in that it matches the minimax rate achieved by the FW-Learner that has access to $f(\cdot)$.
\begin{rem}\label{rem:comparison-with-Kennedy} 
     Section 3 of \cite{kennedy2023towards} provides a result similar to Corollary~\ref{cor:forster-pseudo} but with a more general regression procedure $\widehat{\mathbb{E}}_n(\cdot)$ in the form of a weighted linear estimator, but the assumptions that the weights of the estimator must satisfy require a case by case analysis, which may not be straightforward; whereas our result is tailored to the Forster--Warmuth estimator which applies more broadly under minimal conditions. 
\end{rem}
\begin{rem}
It is worth noting that cross-fitting rather than simple sample splitting can be used to improve efficiency. 
Specifically, by swapping the roles of $\mathcal{I}_1$ and $\mathcal{I}_2$ in~\ref{step-B-FW-Learner-counterfactual}, we can obtain two pseudo-outcomes $\widehat{f}_1(\cdot), \widehat{f}_2(\cdot)$, and also two FW-Learners $\widehat{m}_J^{(1)}(\cdot), \widehat{m}_J^{(2)}(\cdot)$. Instead of using only one of $\widehat{m}_J^{(j)}, j = 1, 2$, one can consider $\widehat{m}_J(x) = 2^{-1}\sum_{j=1}^2 \widehat{m}_J^{(j)}$ and by Jensen's inequality, 
we obtain
\begin{equation}
\begin{aligned}
\|\widehat{m}_J - m^\star \|_2 
&\le \sqrt{\frac{2\sigma^2 J}{n}} + \sqrt{2\kappa} E_J^{\Psi}(m^\star) 
+ \sqrt{\frac{3}{2}} \Bigl( \|H_{f_1} - m^\star\|_2 + \|H_{f_2} - m^\star\|_2 \Bigr),
\label{eq:forster-pseudo-crossfit}
\end{aligned}
\end{equation}
where $H_{f_j}(x) = \mathbb{E}[\widehat{f}_j(O)|X = x, \widehat{f}_j], j = 1, 2$. A similar guarantee also holds for the average estimator obtained by repeating the sample splitting procedure.
\end{rem}

\subsection{Construction of Pseudo-outcome~\ref{step-A-FW-Learner-counterfactual-regression}}\label{subsub:construction-of-pseudo-outcome}
For a given counterfactual regression problem, we construct the counterfactual pseudo-outcome using an influence function (more precisely, the non-centered gradient) of the functional formally defined as the ``marginal" instance of the non-parametric counterfactual regression model in view, under given identifying assumptions, which formally define a semiparametric model $\mathcal{M}$. For instance, in the missing data regression problem, our quantity of interest is $m^\star(x) = \mathbb{E}[Y|X = x]$ and so, the marginal functional is simply $\psi = \mathbb{E}[Y]$, the mean outcome in the underlying target population; both conditional and marginal parameters are identified from the observed data under MAR or the shadow variable model assumptions. Likewise, in the case of the CATE, our quantity of interest is $m^\star(x) = \mathbb{E}[Y^1 - Y^0|X = x]$ and so, the marginal functional is simply $\psi = \mathbb{E}[Y^1 - Y^0]$, the population average treatment effect, both of which are identified under unconfoundedness, or the proximal causal inference assumptions. Importantly, although the nonparametric regression of interest $m^\star(x)$ might not generally be pathwise-differentiable (see the definition in Section \ref{app-sec:proof} of the supplement), and therefore might not admit an influence function, under our identifying conditions and additional regularity conditions, the corresponding marginal functional $\psi$ is a well-defined pathwise-differentiable functional that admits an influence function. Note that a nonparametric regression function that is absolutely continuous with respect to the Lebesgue measure may not be pathwise-differentiable without an additional modeling restriction \citep[Chapter 3] {bickel1993efficient}.   

Influence functions for marginal functionals $\psi$ are in fact well-established in several semiparametric models. Furthermore, unless the model is fully nonparametric, there are infinitely many such influence functions and there is one efficient influence function that has minimum variance. For example, in the setting of missing data with $O = (X, Z, R, YR)$, under the MAR assumption (i.e., $R_i\perp Y_i|(X_i, Z_i)$), the model is well-known to be fully nonparametric in the sense that the assumption does not restrict the observed data tangent space, formally the closed linear span of the observed data scores of the model. The efficient influence function derived by \cite{robins1994estimation} is given by 
\begin{align}\label{eq:eif-robins}
\mathrm{IF}(O; \psi) &:= \frac{R}{\pi^\star(X, Z)}Y - \left(\frac{R}{\pi^\star(X, Z)} - 1\right)\mu^{\star}(X, Z) - \psi,
\end{align}
where $\pi^\star(X, Z) := \mathbb{P}(R = 1|X, Z)$ and $\mu^{\star}(X, Z) := \mathbb{E}[Y|X, Z, R = 1]$.
An estimator of $\psi$ can be obtained by solving the empirical version of the estimating equation $\mathbb{E}[\mathrm{IF}(O; \psi)] = 0$. Interestingly, this influence function also satisfies $m^\star(x) = \mathbb{E}[(\mathrm{IF}(O; \psi) + \psi)|X = x]$. 
Because $\mathrm{IF}(O; \psi) + \psi$ is only a function of $O$, it can be used as $f(O)$ for counterfactual regression. 
In this setting, one can easily construct other pseudo-outcomes. Namely, $f_1(O) := {RY}/{\pi^\star(X, Z)}$ and $f_2(O) := \mu^\star(X, Z),$ both satisfy $\mathbb{E}[f_j(O)|X = x] = m^\star(x)$.
Of these options, the oracle pseudo-outcome $(\mathrm{IF}(O;\psi) + \psi)$ is the only one with mixed bias, i.e. that is doubly robust. This is our general strategy for constructing a pseudo-outcome with second order bias $H_f - m^\star$. An outline of the steps for finding a ``good'' pseudo-outcome for estimating $m^\star(x) = \mathbb{E}[\widetilde{f}(\widetilde{O})|X = x]$ is as follows:
\begin{enumerate}
    \item Derive an influence function $\mathrm{IF}(O; \eta^\star, \psi)$ for  the marginal functional $\psi = \mathbb{E}[\widetilde{f}(\widetilde{O})]$ under a given semiparametric model for which identification of the regression curve is established. Here $\eta^\star$ represents a nuisance component under the model . Note that by definition of an influence function $\mathbb{E}[\mathrm{IF}(O; \eta^\star, \psi)] = 0$.
    \item Because $\mathrm{IF}(O; \eta^{\star}, \psi)+\psi$ is only a function of $O$ and $\eta^\star$, we set $f(O) = \mathrm{IF}(O; \eta^{\star}, \psi) + \psi$. Clearly, $\mathbb{E}[f(O)] = \psi$. By construction, it also follows that $\mathbb{E}[f(O)|X = x] = m^\star(x)$. See Theorem \ref{thm:pseudo-outcome-construction} below.
    \item Construct $\widehat{f}(O) = \widehat{\mathrm{IF}}(O; \widehat{\eta}, \psi) + \psi$, an estimate of the uncentered influence function based on the first split of the data.
\end{enumerate}

The influence functions for both the marginal outcome mean and average treatment effect under MAR and unconfoundedness conditions, respectively, are well-known, the former is given above and studied in Section~\ref{sec:missing}; while the latter is given and studied in Section~\ref{sec:unconfoundness} of the supplement along with their analogs under MNAR with a shadow variable and unmeasured confounding using proxies. See Section \ref{sec:mnar} of the supplement and Section \ref{sec:proximal} respectively for details. A more general result which formalizes the approach for deriving a pseudo-outcome in a given counterfactual regression problem is as follows.
\begin{thm}\label{thm:pseudo-outcome-construction} Suppose that the counterfactual regression function of interest $m^{\ast}(x) =\mathbb{E}[ \widetilde{f}(\widetilde{O})
\left\vert X=x\right.] $ is identified in terms of the observed data $%
O$ (distributed as $F^*\in\mathcal{M}$) by $n^{\ast }\left( x;\eta \right) =\mathbb{E}_{\eta }\left[ r\left(
O;\eta \right) \left\vert X=x\right. \right] $\footnote{To avoid confusion between the counterfactual regression of interest $m^{\ast}$, here we introduce $n^{\ast}$ as the corresponding identifying observed data regression; for instance, for $m^{\ast}$ defined as the CATE, $n^{\ast}$ is a different observed data regression under unconfoundedness vs proximal causal inference identifying conditions, involving a different pair of nuisance functions.} for a known function $r\left(
\cdot ;\eta \right) $ in $L^{2}$ indexed by an unknown, possibly infinite
dimensional, nuisance parameter $\eta\in\mathcal{B}$ (for a normed metric space $\mathcal{B}$ with norm $\|\cdot\|$). Furthermore, suppose that there
exists a function $R(\cdot ;\eta ,n^{\ast }\left( \eta \right) ):O\mapsto
R(O;\eta ,n^{\ast }\left( \eta \right) )$ in $L^{2}$ such that for any
regular parametric submodel $F_{t}\in\mathcal{M}$ with
parameter $t\in \left( -\varepsilon ,\varepsilon \right)$ satisfying $F_{0} =F^{\ast }$ and corresponding likelihood score $S(\cdot)$, the following holds: 
\begin{equation*}
\left. \frac{\partial \mathbb{E}_{\eta}\left[ r\left( O;\eta _{t}\right) \left\vert
X=x\right. \right] }{\partial t}\right\vert _{t=0}=\mathbb{E}_{\eta}\left[ R(O;\eta
,n^{\ast }\left( \eta \right) )S\left( O\right) |X=x\right],\footnote{We also assume that this derivative is continuous in $t$.} {^,} \footnote{ \cite{luedtke2019omnibus} makes an analogous assumption in their condition S.3 in the context of testing the null hypothesis of equality of two distributions.} 
\end{equation*}
with $\E_{\eta}\left[ R(O;\eta ,n^{\ast }\left( \eta \right) )|X\right] =0$, then: 
$\left\Vert \mathbb{E}_{\eta}\left[ R(O;\eta ^{\prime },n^{\ast }\left( \eta
^{\prime }\right) )+r\left( O;\eta ^{\prime }\right) |X\right] -n^{\ast
}\left( X;\eta \right) \right\Vert _{2}=O\left( \left\Vert \eta ^{\prime
}-\eta \right\Vert _{2}^{2}\right)$
for any $\eta'\in \mathbb{B}$,
and 
$
R(O;\eta ,n^{\ast }\left( \eta \right) )+r\left( O;\eta \right) -\psi \left(
\eta \right) 
$
is an influence function of the functional $\psi \left( \eta \right) =%
\mathbb{E}\left[ r\left( O;\eta \right) \right] \,$under $\mathcal{M}$.
\end{thm}
The proof is given in Section \ref{sec:proof-thm2} of the supplement.
Theorem~\ref{thm:pseudo-outcome-construction} formally establishes that a pseudo-outcome for a given counterfactual regression $\mathbb{E}_{\eta }\left[ r\left(
O;\eta \right) \left\vert X=x\right. \right]$, can be obtained by effectively deriving an influence function of the corresponding marginal functional $\psi=\mathbb{E}_X\{\mathbb{E}_{\eta }\left[ r\left(
O;\eta \right) \left\vert X\right. \right]\}$ under a given semiparametric model $\mathcal{M}$. The resulting influence function is given by  $R(O;\eta )+r(O;\eta )-\psi$
and the oracle pseudo-outcome may appropriately be defined as $f(O)=R(O;\eta )+r(O;\eta ).$ 
Theorem~\ref{thm:pseudo-outcome-construction} is quite general as it applies to the most comprehensive class of non-parametric counterfactual regressions studied to date. The result thus provides a unified solution to the problem of counterfactual regression, recovering several existing methods, and more importantly, providing a number of new results. Namely, the theorem provides a formal framework for deriving a pseudo-outcome which by construction is guaranteed to satisfy the so-called ``Neyman Orthogonality" property, i.e., that the bias incurred by estimating nuisance functions is at most of second order \citep{Semenova2018OrthogonalMF}.  In the following sections, we apply Theorem~\ref{thm:pseudo-outcome-construction} to key problems in missing data and causal inference for which we give a precise characterization of the resulting second-order bias. The use-cases we discuss in detail below share a common structure in that the influence function of the corresponding marginal functional is linear in the regression function of interest, and falls within a broad class of so-called mixed-bias functionals introduced by \cite{ghassami2022minimax}. 

{\color{black}We should note that Theorem 1 of \cite{kennedy2017non} established an analogous construction of pseudo-outcome for estimating a counterfactual dose-response curve for a continuous treatment under unconfoundedness, which can be viewed as a special case of the result above; specifically, they note that under unconfoundedness, a functional defined upon marginalizing the dose-response curve with respect to the treatment is pathwise differentiable in a nonparametric model for the observed data, and therefore admits an influence function which in this case yields a doubly robust pseudo-outcome. Relatedly, \cite{rubin2007doubly} proposed a doubly robust transformation in a survival analysis context subject to censoring at random which can likewise be recovered as a special case of our result.  
Going beyond these special cases, our theorem gives a generic high level condition for the existence of, and an explicit construction for a pseudo-outcome with second order bias even in settings where the model may not be locally nonparametric, or a doubly robust estimator may not exist. Key examples of settings where prior approaches to counterfactual prediction problems do not appear to apply directly include proximal causal inference for the counterfactual dose-response curve in presence of hidden confounders for which proxies are available, and the CATE under similar conditions. A treatment of the latter problem is provided in Section 4.1, while the former problem is considered in an application of the theorem to several novel problems of counterfactual prediction in Section S.1.6 of the supplement.}
Indeed, to further demonstrate broader applicability of Theorem~\ref{thm:pseudo-outcome-construction},  we consider settings where the counterfactual regression curve of interest operates on a ``non-linear'' scale in Section \ref{sec-app:examples} of the supplement, in the sense that the influence function for the corresponding marginal functional depends on the counterfactual regression of interest on a nonlinear scale, and as a result, might not strictly belong to the mixed-bias class, nor be doubly robust. Nonetheless, as guaranteed by our theorem, the bias of the resulting pseudo-outcome is indeed of second order albeit not necessarily of mixed-bias form.  These additional applications include the conditional quantile causal effect under unconfoundedness conditions, the CATE for generalized nonparametric regressions incorporating a possibly nonlinear link function such as the log or logit links, to appropriately account for the restricted support of count and binary outcomes  respectively; The CATE for the treated, the compliers, and for the overall population each of which can be identified uniquely in the presence of unmeasured confounding under certain conditions by the so-called conditional Wald estimand, by carefully leveraging a binary instrumental variable \citep{wang2018bounded}; and the nonparametric counterfactual outcome mean for a continuous treatment both under unconfoundedness and proximal causal identification conditions, respectively. 
An important class of influence functions for pathwise differentiable functionals $\psi$ that have second order bias of the mixed-bias form was introduced and studied in \cite{ghassami2022minimax} are of the form:
\begin{align}\label{eq:general-if}
    \mathrm{IF}_{\psi}(O)=q^\star(O_q) h^\star(O_h) g_1(O)+q^\star(O_q) g_2(O)+h^\star(O_h) g_3(O)+g_4(O)-\psi,
\end{align}
where $O_q$ and $O_h$ are (not necessarily disjoint) subsets of the observed data vector $O$ and $g_1, g_2, g_3$, and $g_4$ are known functions and $\eta^\star = (h^\star, q^\star)$ represents nuisance functions that need to be estimated. Then, a natural choice of oracle pseudo-outcome function is given by $f(O) = q^\star(O_q) h^\star(O_h) g_1(O)+q^\star(O_q) g_2(O)+h^\star(O_h) g_3(O)+g_4(O)$, and empirical pseudo-outcome $\widehat{f}(O) = \widehat{q}(O_q)\widehat{h}(O_h) g_1(O)+\widehat{q}(O_q) g_2(O)+\widehat{h}(O_h)g_3(O)+g_4(O)$, where $\widehat{h}, \widehat{q}$ are estimators of corresponding nuisance functions $h^\star$ and $q^{\star}$ using nonparametric method; see Section \ref{app-sec:nonparametric} of the supplement for some nonparametric estimators that can adapt to the low-dimensional structure of $\eta^\star=(h^\star, q^\star)$, when the latter are defined in terms of conditional expectations.  Using an analogous proof to that of Theorem 2 of \cite{ghassami2022minimax}, it can be shown that conditioning on the training sample used to estimate the nuisance functions $h^\star$ and  $q^\star$ with $\widehat{h}$ and $\widehat{q}$, the bias term $H_f - m^\star$ above is equal to 
\begin{align}\label{eq:dr-bias-general}
    \E\bigl\{ g_1(O)(q^\star - \widehat{q})(O_q) (h^\star - \widehat{h})(O_h) | X,\widehat{q}, \widehat{h} \bigr\},
\end{align}
and confirming that the bias term is of second order with product form. A detailed proof largely based on \cite{ghassami2022minimax} is given for completeness in Section \ref{app-sec:proof} of the supplement.
The following sections illustrate these results by considering two specific applications within the mixed-bias class.

\section{FW-Learner for Missing Outcome}\label{sec:missing}
In this section, we suppose that a typical observation is given by $O=(Y R, R, X, Z)$, where $R$ is a nonresponse indicator with $R=1$ if $Y$ is observed, otherwise $R=0$. Here $Z$ are fully observed covariates not directly of scientific interest, but may be helpful to account for selection bias induced by the missingness mechanism. Specifically, Section \ref{sec:mar} considers the MAR setting where the missingness mechanism is assumed to be completely accounted for by conditioning on the observed covariates $(X,Z)$\footnote{In the special case where MAR holds upon conditioning on $X$ only, complete-case estimation of $m^\star$ is known to be minimax rate optimal~\citep{efromovich2011nonparametric,efromovich2014nonparametric,muller2017efficiency}.}, while Section \ref{sec:mnar} in the supplement relaxes this assumption, allowing for outcome data Missing Not At Random (MNAR) leveraging a shadow variable for identification. 

\subsection{FW-Learner under MAR}\label{sec:mar}
Here, we make the MAR assumption that $Y$ and $R$ are conditionally independent given $(X,Z)$, and we aim to estimate the conditional mean  of $Y$ given $X$, which we denote $m^\star(x):= \E [ Y\mid X=x]$.

\begin{enumerate}[label=\bf(MAR),leftmargin=0.66in]
\setcounter{enumi}{0}
\item \label{assump:mar} $O_i=(X_i, Z_i, R_i,Y_iR_i), 1\le i\le n$ are independent and identically distributed random vectors satisfying $R_i\perp Y_i \mid (X_i,Z_i)$.
\end{enumerate}

Under the \ref{assump:mar} assumption, 
the efficient influence function for the marginal function $\psi = \mathbb{E}[Y]$ given by \eqref{eq:eif-robins} was derived by \cite{robins1994estimation}. Following~\ref{step-B-FW-Learner-counterfactual}, we now define empirical pseudo-outcome as follows. Split $\{1, 2, \ldots, n\}$ into two parts: $\mathcal{I}_1$ and $\mathcal{I}_2$. Use the first split to estimate the nuisance functions based on data $\{(Y_iR_i, R_i, X_i, Z_i), i\in\mathcal{I}_1\}$, denoted as $\widehat{\pi}$ and $\widehat{\mu}$; recall $\pi^\star(X, Z) = \mathbb{P}(R = 1|X, Z)$ and $\mu^\star(X, Z) = \mathbb{E}[Y|X, Z]$.
Use the second split and define the empirical pseudo-outcome
\begin{equation}\label{eq:forster-imputation-missing}
\begin{split}
\widehat{f}(O) = \widehat{f}(YR, R, X, Z) &:= \frac{R}{\widehat{\pi}(X, Z)}(YR) - \left(\frac{R}{\widehat{\pi}(X, Z)} - 1\right)\widehat{\mu}(X, Z),\\
&= \frac{R}{\widehat{\pi}(X, Z)}Y - \left(\frac{R}{\widehat{\pi}(X, Z)} - 1\right)\widehat{\mu}(X, Z),
\end{split}
\end{equation}
Note that this corresponds to a member of the DR class of influence function \eqref{eq:general-if} with $h_0(O_h) = 1/\pi^\star(X, Z)$, $q_0(O_q) = \mu^\star(X, Z), g_1 = -R, g_2=1, g_3 = RY$ and $g_4=0$.  

Let $\widehat{m}_J(\cdot)$ represent the FW-Learner computed from the dataset $\{(\widebar{\phi}_J(X_i), \widehat{f}(O_i)), i\in\mathcal{I}_2\}$, as in~\ref{step-B-FW-Learner-counterfactual}
and Corollary \ref{cor:forster-pseudo} guarantees the following result
\begin{equation}\label{eq:oracle-missing}
    \begin{split}
        (\mathbb{E}[(\widehat{m}_J(X) - m^\star(X))^2|\widehat{f}])^{1/2} &\le \sqrt{\frac{2\sigma^2J}{|\mathcal{I}_2|}} + \sqrt{2\kappa}E_J^{\Psi}(m^\star) + \sqrt{6}(\mathbb{E}[(H_f(X) - m^\star(X))^2|\widehat{f}])^{1/2},   
    \end{split}
\end{equation}
where $\sigma^2$ is an upper bound on $\mathbb{E}\bigl[\widehat{f}^2(O)\mid X, \widehat{f}\bigr]$ and $H_{f}(x) := \mathbb{E}[\widehat{f}(O)|X = x, \widehat{f}].$
The following lemma states the mixed bias structure of $H_f - m^\star$. 
\begin{lem}\label{lem:forster-missing}
With~\eqref{eq:forster-imputation-missing} as the empirical pseudo-outcome, under~\ref{assump:mar}, we have
\begin{align}
H_{f}(x) - m^\star(x) = \E \biggl\{R \left(\frac{1}{\widehat{\pi}(X, Z)} - \frac{1}{\pi^\star(X, Z)} \right) \Bigl(\mu^\star(X, Z) - \widehat{\mu}(X, Z) \Bigr) \biggm| X=x, \widehat{\pi}, \widehat{m} \biggr\}. 
\label{eq:I1-bias}
\end{align}
\end{lem}
This result directly follows from the mixed bias form \eqref{eq:dr-bias-general} (also see \cite{rotnitzky2021characterization} and \cite{robins2008higher}); for completeness, we provide a proof in Section \ref{sec-app:forster-missing} of the supplement. Lemma~\ref{lem:forster-missing} combined with~\eqref{eq:oracle-missing} gives the following error bound for the FW-Learner computed with pseudo-outcome~\eqref{eq:forster-imputation-missing}.

\begin{thm}\label{thm:forster-missing}
Let $\sigma^2$ denote an almost sure upper bound on $\mathbb{E}[\widehat{f}^2(O)|X,\widehat{\pi},\widehat{\mu}]$.
Then, under~\ref{assump:mar}, the FW-Learner $\widehat{m}_J(x)$ satisfies
\begin{align}\label{eq:forster-missing}
&(\mathbb{E}[(\widehat{m}_J(X) - m^\star(X))^2|\widehat{f}])^{1/2}\\ &\le \sqrt{\frac{2\sigma^2J}{|\mathcal{I}_2|}} + \sqrt{2\kappa}E_J^{\Psi}(m^\star)
+ \sqrt{6}\mathbb{E}^{1/4}\left[\left(\frac{\pi^\star(X, Z)}{\widehat{\pi}(X, Z)} - 1\right)^4\big|\widehat{\pi}\right]\mathbb{E}^{1/4}[(\mu^\star(X, Z) - \widehat{\mu}(X, Z))^4\big|\widehat{\mu}].%
\end{align}
\end{thm}
The proof of this result is in Section \ref{sec-app:forster-missing} of the supplement. Note that, because $\widehat{f}(O)$ involves $\widehat{\pi}$ in the denominator, the condition that $\sigma^2$ is finite requires $\widehat{\mu}$ and $1/\widehat{\pi}$ to be bounded.

\begin{cor}\label{cor:missing-mar}
Let $d$ denote the intrinsic dimension of $(X,Z)$, if 
\begin{enumerate}
\item The propensity score $\pi^\star(x, z)$ and regression function $\mu^\star (x,z)$ are estimated at $n^{-2 \alpha_\pi /(2 \alpha_\pi+d)}$ and $n^{-2 \alpha_\mu /(2 \alpha_\mu+d)}$ rate respectively in the $L_4$-norm and
\item The conditional mean function $m^\star(\cdot)$ with respect to the fundamental sequence $\Psi$ satisfies $E_J^{\Psi}(m^\star) \le CJ^{-\alpha_m/d}$ for some constant $C$,
\end{enumerate}
then
\begin{align}\label{eq:mar-final}
    \Bigl(\mathbb{E}[(\widehat{m}_J(X) - m^\star(X))^2|\widehat{\pi}, \widehat{\mu}] \Bigr)^{1/2}
    & \lesssim \sqrt{\frac{\sigma^2 J}{n}} + J^{-\alpha_m/d}  + n ^{-\frac{\alpha_\pi}{2\alpha_\pi+d} - \frac{\alpha_\mu}{2\alpha_\mu+d}}.
\end{align}
\end{cor}

When the last term of \eqref{eq:mar-final} is smaller than the oracle rate $n ^{-\frac{\alpha_m}{2\alpha_m+d}}$, the oracle minimax rate can be attained by balancing the first two terms. Therefore, the FW-Learner is oracle efficient if $\alpha_\mu \alpha_\pi \geq {d^2}/{4}-{(\alpha_\pi+\frac{d}{2})(\alpha_\mu+\frac{d}{2})}/{(1+\frac{2 \alpha_m}{d})}$.
In the special case when $\alpha_\mu$ and $\alpha_\pi$ are equal, if we let $s = \alpha_\mu/d = \alpha_\pi/d$ and $\gamma = \alpha_m/d$ denote the effective smoothness, and when $s \geq \frac{\alpha_m/2 }{\alpha_m+d} = \frac{\gamma/2}{\gamma+1}$, the last term in \eqref{eq:forster-missing} is the bias term that comes from pseudo-outcome, which is smaller than that of the oracle minimax rate of estimation of $n^{-\alpha_m/(2 \alpha_m + d)}$ and the FW-Learner is oracle efficient.



\section{FW-Learner of the CATE}\label{sec:cate}
Estimating the Conditional Average Treatment Effect (CATE) plays an important role in health and social sciences where one might be interested in tailoring treatment decisions based on the person's characteristics, a task that requires learning whether and the extent to which the person may benefit from treatment, e.g., personalized treatment in precision medicine \citep{ashley2016towards}. 

Suppose that we have observed i.i.d data $O_i = (X_i, A_i, Y_i), 1\le i\le n$ with $A_i$ representing the binary treatment assignment, $Y_i$ being the observed response, and covariates $X_i$.  
 The CATE is formally defined as $m^\star(x) = \mathbb{E}\left(Y^1-Y^0|X=x\right)$, where $Y^a$ is the potential outcome or counterfactual outcome, had possibly contrary to fact, the person taken treatment $a$. The well-known challenge of causal inference is that one can at most observe the potential outcome for the treatment the person took and therefore, the counterfactual regression defining the CATE is in general not identified outside of a randomized experiment with perfect compliance, without additional assumptions. Section \ref{sec:unconfoundness} of the supplementary material describes the identification and FW-Learner of the CATE under standard unconfoundedness conditions, 
 while the following Section \ref{sec:proximal} presents analogous results for the proximal causal inference setting which does not make the unconfoundedness assumption. Throughout, we assume consistency, that $Y=AY^1+(1-A)Y^0$, and positivity, that $\mathbb{P}(A=a|X,U) >\epsilon > 0  $ almost surely for all $a$, where $U$ denotes unmeasured confounders, and therefore is empty under unconfoundedness. 
 


\subsection{FW-Learner for CATE under proximal causal inference}\label{sec:proximal}
 Proximal causal inference provides an alternative approach for identifying the CATE in the presence of unobserved confounding, provided that valid proxies of the latter are available \citep{miao2018identifying, tchetgen2020introduction}. 
  Throughout, recall that $U$ encodes a (possibly multivariate) unmeasured confounder.  The framework requires that observed proxy variables $Z$ and $W$ satisfy the following conditions.
 
\textbf{Assumption 1:} The proxy variables $W$ and $Z$ satisfy the following conditional independence condition: 
$\left(Y^{a}, W\right) \perp (A, Z) \mid(U, X)$, for $a \in\{0,1\}$.

Assumption 1 
formalizes key conditional independencies the treatment, the potential outcomes and the proxies must satisfy given the unmeasured confounder $U$. It encodes a standard no unmeasured confounding assumption for the causal effect of $A$ on $Y$ had one observed and conditioned on both $X$ and $U$. In addition, the treatment confounding proxy $Z$ should be associated with $Y$ only to the extent that it is associated with $U$ given $(A,X)$, while the outcome confounding proxy $W$ should be associated with $A$ or $Z$ only to the extent that it is associated with $U$ conditional on $X$. Therefore, given $U$, proxies should become irrelevant for confounding adjustment as they should then provide no additional information for that purpose conditional on $U$.  The assumption is illustrated with the causal diagram in Figure \ref{fig:proximal-dag} which describes a possible setting where these assumptions are satisfied (the gray variable $U$ is unobserved). 
\begin{figure}[!h]
    \centering
    \includegraphics[width=0.4\textwidth]{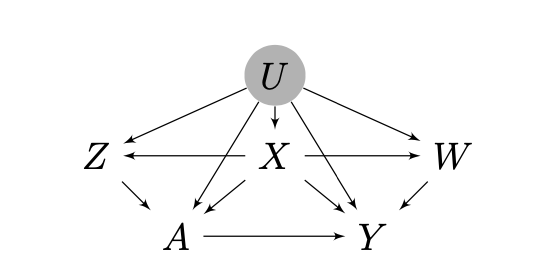}
    \caption{{\color{black}A proximal DAG with exposure $A$, outcome $Y$, observed covariates $X$, unmeasured confounders $U$, treatment and outcome proxies $Z$ and $W$.}}
    \label{fig:proximal-dag}
\end{figure}


A key identification condition of proximal causal inference is there exists an outcome confounding bridge function $h^\star(w, a, x)$ that solves the following integral equation \citep{miao2018identifying, tchetgen2020introduction}
\begin{align}\label{eq:proximal-h}
\mathbb{E}[Y \mid Z, A, X]=\mathbb{E}\left[h^\star (W, A, X) \mid Z, A, X\right], \quad\text{almost surely.}
\end{align}
\cite{miao2015identification} then established sufficient conditions under which such as $h^\star$ exists and the CATE is nonparametrically identified by $\mathbb{E}(h(W,1,X)-h(W,0,X)|X)$. A crucial assumption is that $(W,Z) $ are $U$-relevant, and therefore sufficiently associated with $U$, in the sense that any variation in $U$ induces some variation in $(W,Z)$. We refer the reader to \cite{miao2018identifying, tchetgen2020introduction,cui2020semiparametric} for further details. 
\cite{cui2020semiparametric} considered an alternative identification strategy based on the following condition. 
There exists a treatment confounding bridge function $q^\star(z, a, x)$ that solves the following integral equation
\begin{align}\label{eq:proximal-q}
\mathbb{E}\left[q^\star(Z, a, X) \mid W, A=a, X\right]=\frac{1}{\PP(A=a \mid W, X)}, \quad \text{almost surely.} 
\end{align}
Also, see \cite{deaner2018proxy} for a related condition. \cite{cui2020semiparametric} provide sufficient conditions under which such a function $q^\star$ exists, and the CATE is nonparametrically identified as $\mathbb{E}(Y (-1)^{1-A}q(Z,A,X)|X)$.
Let $O=(X,Z,W,A,Y)$, \cite{cui2020semiparametric} derived the locally semiparametric efficient influence function for the marginal ATE (i.e. $\E[Y^{(1)} - Y^{(0)}]$) in a nonparametric model where one only assumes an outcome confounding bridge function exists, at the submodel $\mathcal{M}$ where both outcome and treatment confounding bridge functions exist and are uniquely identified, but otherwise unrestricted:
\begin{align} \mathrm{IF}_{\psi_0}(O;h^\star,q^\star)=-\mathbbm{1}\{A=a\} q^\star(Z, A, X) h^\star(W, A, X) 
+\mathbbm{1}\{A=a\} Y q^\star(Z, A, X)+h^\star(W, a, X)-\psi_0,
\end{align}
which falls in the mixed-bias class of influence functions \eqref{eq:general-if} with $h_0(O_h) = h^\star(W,A,X), q_0(O_q) = q^\star(Z,A,X)$, $g_1(O) = -\mathbbm{1}\{A=a\}, g_2(O) = \mathbbm{1}\{A=a\}Y, g_3(O)=1, g_4(O)=0$, and motivates the following FW-Learner of the CATE. 

At the submodel $\mathcal{M}$, we have two identifiable representations of the CATE. 
\[
m^\star(X) \equiv \tau^\star(X):=\mathbb{E}(h(W,1,X)-h(W,0,X)|X)=\mathbb{E}(Y (-1)^{1-A}q(Z,A,X)|X).
\]
Proximal CATE FW-Learner estimator:
Split the training data into two parts and train the nuisance functions $\widehat{q},\widehat{h}$ on the first split and define $\widehat{\tau}_J(x)$ to be the Forster--Warmuth estimator computed based on the data $\bigl\{ (\widebar{\phi}_J(X_i), \widehat{I}(O_i)), i\in\mathcal{I}_2\bigr\}$,
where the pseudo-outcome $\widehat{I}$ is
\begin{align}\label{eq:pseudo-outcome-proxy}
    \widehat{I}(O;\widehat{h},\widehat{q})&:=\bigl\{A \widehat{q}(Z,1,X) - (1-A)\widehat{q}(Z,0,X) \bigr\} \{ Y - \widehat{h}(W, A, X)\} 
    +\widehat{h}(W,1,X) - \widehat{h}(W,0,X), 
\end{align}
for any estimators $\widehat{h}, \widehat{q}$ of the nuisance functions $h^\star$ and $q^\star$. 
Write $H_{I}(X) = \E[\widehat{I}(O;\widehat{h}, \widehat{q})|X, \widehat{h}, \widehat{q}]$.
We have the following result. 
\begin{lem}\label{lem:forster-proximal-cate}
The pseudo-outcome  \eqref{eq:pseudo-outcome-proxy} has conditional bias: 
\begin{align}
    H_{I}(x) - \tau^\star(x)
    &= 
    \E \biggl[ A (h^\star - \widehat{h})(W, 1, x)  \bigl( \widehat{q}(Z,1,x)  - q^\star(Z,1,x)  \bigr)  \\
    & \qquad \qquad - (1-A)(h^\star - \widehat{h})(W, 0, x)\bigr(  \widehat{q}(Z,0,x) -q^\star(Z,1,x)  \bigr)   \Bigm| X=x \biggr].\label{eq:proximal-bias}
\end{align}
\end{lem}
This result directly follows from the mixed bias form \eqref{eq:dr-bias-general} in the general class studied by \cite{ghassami2022minimax}; its proof is deferred to Section \ref{sup-sec:proximal-cate} of the supplement. Together with Corollary~\ref{cor:forster-pseudo}, Lemma~\ref{lem:forster-proximal-cate} yields a bound for the error of the FW-Learner $\widehat{\tau}_J$.
\begin{thm}\label{thm:forster-cate-proximal}
Let $\sigma^2$ be an upper bound on $\mathbb{E}[\widehat{I}^2(X, Z,W,A,Y) \mid X]$, the FW-Learner $\widehat{\tau}_J(x)$ satisfies:
\begin{align*}
    \bigl\|\widehat{\tau}_J(X) - \tau^\star (X) \bigr\|_2 \le &\sqrt{\frac{2\sigma^2J}{|\mathcal{I}_2|}} + \sqrt{2}\Bigl\|\sum_{j=J+1}^{\infty} \theta_j^\star\phi_j(X) \Bigr\|_2 + 2(1+\sqrt{2}) \min\Bigl\{ R_h, R_q\Bigr\},
\end{align*}
where
\begin{align*}
R_h &:= \Bigl\| (\widehat{q} - q^\star )(Z,1,X)    \Bigr\|_4\Bigl\| \E\bigl[ ( \widehat{h} - h^\star)(W, 1, X) | Z,X  \bigr] \Bigr\|_4\\ &\quad+  \Bigl\|  ( \widehat{q} - q^\star )(Z,0,X)    \Bigr\|_4  \Bigl\|  \E\bigl[ ( \widehat{h} - h^\star)(W, 0, X) | Z,X \bigr] \Bigr\|_4,\\
    R_q &:= \Bigl\| \E\bigl[ (\widehat{q} - q^\star )(Z,1,X)  \Bigm | W,X  \bigr] \Bigr\|_4\Bigl\| ( \widehat{h} - h^\star)(W, 1, X)  \Bigr\|_4 \\
    &\quad+  \Bigl\| \E\bigl[ (\widehat{q} - q^\star )(Z,0,X)  \Bigm | W,X  \bigr] \Bigr\|_4  \Bigl\| ( \widehat{h} - h^\star)(W, 0, X) \Bigr\|_4.
\end{align*}
\end{thm}
{\color{black}We note that the term $\min\bigl\{ R_h, R_q\bigr\}$ in the bound given above reflects a potential improvement in the convergence rate for the bias term, due to smoothing over $Z$ or $W$, as the projected bias of, say $ \| \E[ (\widehat{q} - q^\star )(Z,1,X)  | W,X  ] \|_4$ will generally be of order no larger than that of $ \| (\widehat{q} - q^\star )(Z,1,X)  \|_4$, and may potentially be of lower order.}
The proof is given in Section \ref{sup-sec:proximal-cate} of the supplement. Note that the condition that $\sigma^2$ is bounded requires that $\widehat{h}_0$, $\widehat{h}_1, \widehat{q}_0$ and $\widehat{q}_1$ are bounded. The rest of this section is concerned with estimation of the bridge functions $h^\star$ and $q^\star$.


\paragraph{Estimation of bridge functions $h^\star$ and  $q^\star$:}
Focusing primarily on $h^\star$, we note that integral equation \eqref{eq:proximal-h}  is a so-called Fredholm integral equation of the first kind, which are well known to be ill-posed \citep{kress1989linear}. 
{\color{black}
Informally, ill-posedness  essentially measures the extent to which the conditional expectation defining the kernel of the integral equation $Q \mapsto \mathbb{E}_Q [h(W_i, A_i, X_i ) \mid Z_i=z, A_i = a, X_i = x]$ smooths out $h$. Let $L_2(X)$ denote the class of functions $\{f: \E_{X} [f^2(X)] \leq \infty\}$, and define the operator $T: L_2(W, A, X) \rightarrow L_2(Z, A, X)$ as the conditional expectation operator given by
$$
[T h](z,a, x)=\E [h(W_i, A_i, X_i ) \mid Z_i=z, A_i = a, X_i = x].
$$
Let $\Psi_J:=\operatorname{clsp}\left\{\psi_{J 1}, \ldots, \psi_{J J}\right\} \subset L_2(W, A, X)$ denote a sieve spanning the space of functions of variables $W, A, X$.  The Sieve $L_2$ measure of ill-posedness coefficient following \cite{blundell2007semi} is defined as $\tau_h:=\sup _{h \in \Psi_J: \eta \neq 0} {\|h\|_{L_2(W, A, X)}}/{\|T h\|_{L_2(Z,A,X)}}.$
}


Thus, minimax estimation of $h^{\star}$ with respect to the sup norm follows from \cite{chen2018optimal} and \cite{chen2021adaptive}, attaining the optimal rate $(n / \log n)^{-\alpha_h /(2(\alpha_h+\varsigma_h)+d_x+d_w)}$  assuming $T$ is mildly ill-posed with exponent $\varsigma_h$ (the definition is given in Definition \ref{def:ill-posedness} of the supplement); a corresponding adaptive minimax estimator that attains this rate is also given by the authors which does not require prior knowledge about $\alpha_h$ and $\varsigma_h$. See details given in Lemma \ref{lem:proximal-h} in the supplement. Analogous results also hold for $q^\star$ which can be estimated at the minimax rate of $(n / \log n)^{-\alpha_q /(2(\alpha_q+\varsigma_q)+d_x+d_z)}$ in the mildly ill-posed case,  as established in Lemma \ref{lem:proximal-q} of the supplement, where $\alpha_q$ and $\varsigma_q$ are similarly defined and $T'$ is the adjoint operator of $T$. 
Without loss of generality, suppose that $R_h = \min\{R_h, R_q\}$.
Further, suppose that  $\mu^\star(X,Z):=\E\bigl[h^\star(W, 0, X) | Z,X \bigr]$ is  $\alpha_\mu$-smooth, and  $ \Bigl\|  \E\bigl[ ( \widehat{h} - h^\star)(W, 0, X) | Z,X \bigr] \Bigr\|_4$ matches the minimax rate of estimation for $\mu^\star(X,Z)$ with respect to the $L_4$-norm given by $n^{-\alpha_\mu /(2\alpha_\mu+d_x+d_z)}$.
Accordingly, Theorem \ref{thm:forster-cate-proximal}, together with Lemmas \ref{lem:proximal-h} and \ref{lem:proximal-q} whose details are deferred to Section \ref{sup-sec:bridge-cate} of the supplement, leads to the following corollary.
\begin{cor}\label{cor:fw-proximal}
Under the above conditions, and assuming that the integral equation with respect to the operator $T'$ is mildly ill-posed with  $\tau_q=O\left(J^{\varsigma_q / (d_x+d_w)}\right)$ for some $\varsigma_h>0$, we have that:   
\begin{align}\label{eq:fw-proxy}
    \bigl\|\widehat{\tau}_J(X) - \tau^\star(X) \bigr\|_2 \lesssim \sqrt{\frac{\sigma^2 J}{n}} + J^{-\alpha_\tau/d_x}  + 
&(n / \log n)^{- \alpha_q/(2(\alpha_q+\varsigma_q)+d_x+d_z)} n^{-\alpha_\mu /(2\alpha_\mu+d_x+d_z) }.
\end{align}
\end{cor}
\begin{rem}\label{rem:proxy}
    A few remarks on Corollary \ref{cor:fw-proximal}:  (1) If the mixed bias term incurred for estimating nuisance functions is negligible relative to the first two terms in \eqref{eq:fw-proxy}, then the order of the error of the FW-Learner matches that of the oracle with ex ante knowledge of $f(O)$; (2) In settings where operator $T'$ is severely ill-posed, i.e. where  $\tau_q=O\left(\exp \left(\frac{1}{2} J^{\varsigma_q / (d_x+d_w)}\right)\right)$ for some $\varsigma_q>0$, Theorem 3.2 of \cite{chen2018optimal} established that the optimal rate of estimating $q^\star$ with respect to the sup norm is of the order $(\log n)^{-\alpha_q/\varsigma_q}$ which would likely dominate the error $\|\widehat{\tau}_J - \tau^\star \bigr\|_2$. In this case, the FW-Learner may not be able to attain the oracle rate. In this case, whether the oracle rate is at all attainable remains an open problem in the literature.
\end{rem}    

 The result thus establishes conditions under which a proximal FW-Learner can estimate the CATE at the same rate as an oracle. It is worth mentioning that recent concurrent work \cite{sverdrup2023proximal} also estimates CATE under the proximal causal inference context with what they call a
P-Learner using a two-stage loss function approach inspired by the R-Learner proposed in \cite{nie2021quasi}, which, in order to be oracle optimal, requires that the nuisance functions are estimated at rates faster than $n^{-1/4}$, a requirement we do not impose. 

\section{Simulations}\label{sec:simulations}

In this section, we study the finite sample performance of the proposed methods focusing primarily on the estimation of the CATE via simulations, assuming no unmeasured confounding, i.e. $(Y^0, Y^1) \perp A | X$. Under this condition, the CATE is nonparametrically identified by  $\tau^\star(x) = \mu^\star_1(x) - \mu^\star_0(x)$, where for $a \in \{0, 1\},
\mu^\star_a(x) := \mathbb{E}[Y|X = x, A = a]$. Let $\pi^\star(x) := \mathbb{P}(A = 1|X = x)$. 
Using Theorem \ref{thm:pseudo-outcome-construction} yields the following pseudo-outcome
\[
\widehat{I}_1(X_i, A_i, Y_i) = \frac{A_i - \widehat{\pi}(X_i)}{\widehat{\pi}(X_i)(1 - \widehat{\pi}(X_i))}(Y_i - \widehat{\mu}_{A_i}(X_i)) + \widehat{\mu}_1(X_i) - \widehat{\mu}_0(X_i);
\]
  we refer interested readers to Section \ref{sec:CATE-glm} for a detailed derivation.  Notably, this pseudo-outcome matches that of the DR-Learner in \cite{kennedy2023towards} which uses a local polynomial smoother to estimate the CATE; in contrast, the corresponding FW-Learner is a new estimator of the CATE which allows the user to specify a basis function of choice to estimate the CATE. Furthermore, while 
\citep{kennedy2022minimax, kennedy2023towards} studied the problem of estimating CATE under ignorability quite extensively obtaining both minimax lower and upper bound rates of estimation for the local polynomial CATE estimator over a broad range of smoothness regimes, our proposed estimator is minimax rate optimal over a more limited range of regimes (see Theorem \ref{thm:forster-cate} and Corollary \ref{cor:missing-ignorability} of the supplement) however, for a broader class of basis functions.

We consider a relatively simple data-generating mechanism considered in Section 2.2 of \cite{kennedy2023towards}. It includes a covariate $X$ uniformly distributed on $[-1,1]$, a Bernoulli distributed treatment with conditional mean equal to $\pi^\star (x) = 0.1 + 0.8 \times \mathrm{sign}(x)$
and $\mu_1(x) = \mu_0(x)$ are equal to the piece-wise polynomial function defined on page 10 of \cite{gyorfi2002distribution}. Therefore we are simulating under the null CATE model. Multiple methods are compared in the simulation study. Specifically, the simulation includes all four methods described in Section 4 of \cite{kennedy2023towards}: (1) a plug-in estimator that estimates the regression functions $\mu_0^\star$ and $\mu_1^\star$ and takes the difference (called the T-Learner by \cite{kunzel2019metalearners}, abbreviated as plugin below), (2) the X-Learner from \cite{kunzel2019metalearners} (xl), (3) the DR-Learner using smoothing splines from \cite{kennedy2023towards} (drl), and (4) an oracle DR Learner that uses the oracle (true) pseudo-outcome in the second-stage regression (oracle.drl), we compare these previous methods to (5) the FW-Learner with basic spline basis (FW\_bs), and (6) the least squares series estimator with basic spline basis (ls\_bs), where cross-validation is used to determine the number of basis functions to use for (5) and (6). Throughout, the nuisance functions $\mu_0^\star$ and $\mu_1^\star$ are estimated using smoothing splines, and the propensity score $\pi^\star$ is estimated using logistic regression and the propensity score estimator is constructed as $\widehat{\pi}=\operatorname{expit}\left\{\operatorname{logit}(\pi)+\epsilon_n\right\}$, where $\epsilon_n \sim N\left(n^{-\alpha}, n^{-2 \alpha}\right)$ with varying convergence rate controlled by the parameter $\alpha$, so that $\operatorname{RMSE}(\widehat{\pi}) \sim n^{-\alpha}$. 

 The top part of Figure \ref{fig:conparison_uniform} gives the mean squared error (MSE) for the six CATE estimators at training sample size $n=2000$, based on 500 simulations with MSE averaged over 500 independent test samples. The bottom part of Figure \ref{fig:conparison_uniform} gives the ratio of MSE of each of the competing estimators compared to the FW-Learner (the baseline method is FW\_bs) across a range of convergence rates for the propensity score estimator $\widehat{\pi}$.  
 The results demonstrate that, at least in the simulated setting, our FW-Learner attains the smallest mean squared error among all methods, approaching that of the oracle as the propensity score estimation error decreases (i.e., as the convergence rate increases). The performance of the FW-Learner and the least squares series estimator is  visually challenging to distinguish in the figure; however closer numerical inspection confirms that the FW-Learner outperforms the least squares estimator. 

To further illustrate the comparison between the proposed FW-Learner and the least squares estimator, we performed an additional simulation study focusing on these two estimators using two different sets of basis functions, in a simulation setting similar to the previous one, with the covariate instead generated from a heavy-tailed distribution that is an equal probability mixture of a uniform distribution on $[-1,1]$ and a standard Gaussian distribution. The results are reported in Figure \ref{fig:forster-ls}, for both FW-Learner (FW) and Least Squares (LS) estimators with basic splines (bs), natural splines (ns) and polynomial basis (poly). We report the ratio of MSE of all estimators against the FW-Learner with basic splines ($\text{FW}\_\text{bs}$). The sample size is $n=2000$ for the left-hand plot, and $n=400$ for the right-hand plot. The FW-Learner consistently dominates the least squares estimator for any given choice of basis function in this more challenging setting. This additional simulation experiment demonstrates the robustness of the FW-Learner against possible heavy-tailed distribution when compared to least-squares Learner. Data and R code to reproduce the experiments in this section are provided at \href{https://github.com/Elsa-Yang98/Forster_Warmuth_counterfactual_regression}{https://github.com/Elsa-Yang98/Forster\_Warmuth\_counterfactual\_regression}.
\begin{figure}
\centering
\begin{subfigure}{.5\textwidth}
  \centering
  \includegraphics[width=.7\linewidth]{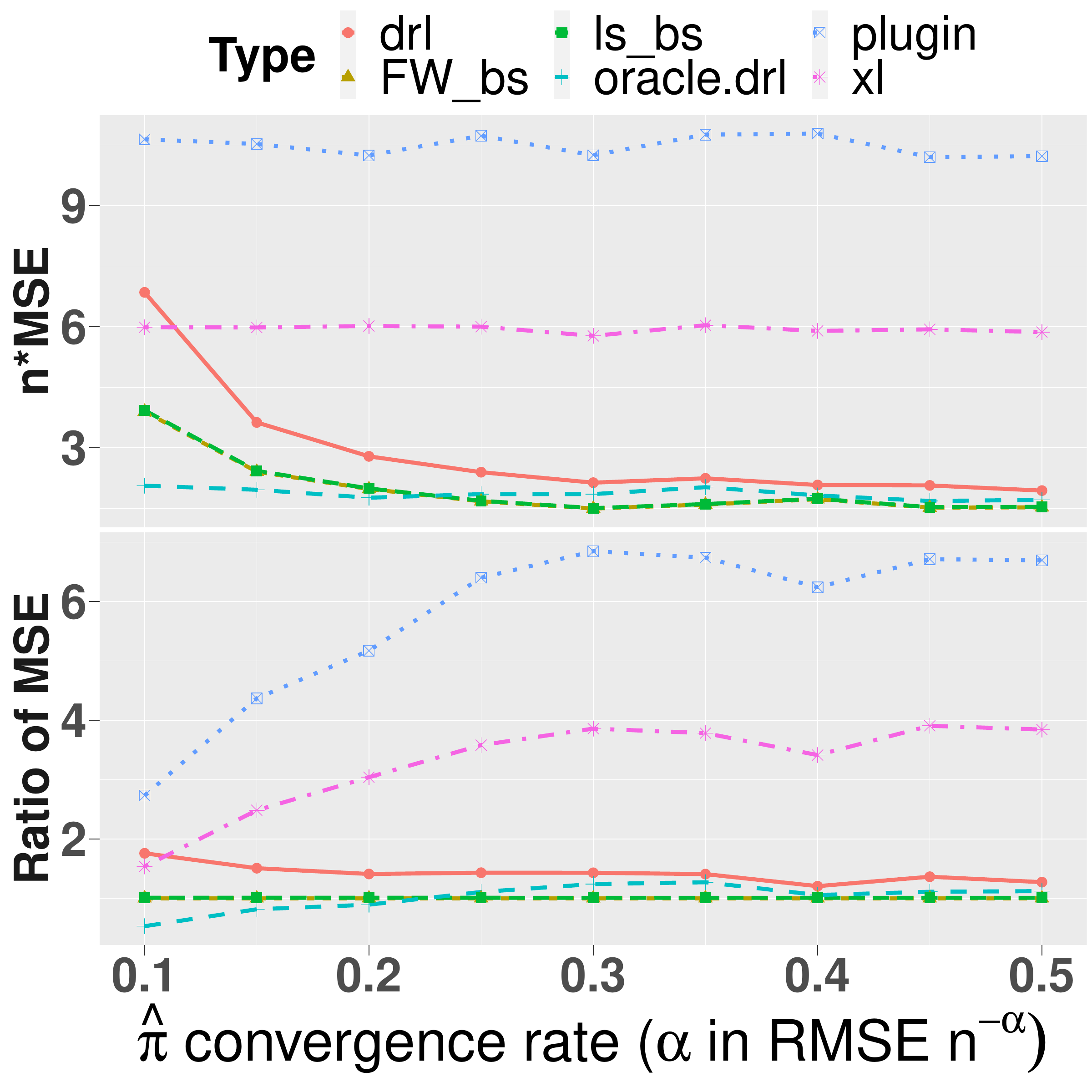}
  \caption{Top: $n\times \text{MSE}$ of each estimator; \\  Bottom: ratio of MSE of different estimators\\  compared to the proposed $\text{FW}\_\text{bs}$ (baseline).}
  \label{fig:conparison_uniform}
\end{subfigure}%
\begin{subfigure}{.5\textwidth}
  \centering
  \includegraphics[width=.9\linewidth]{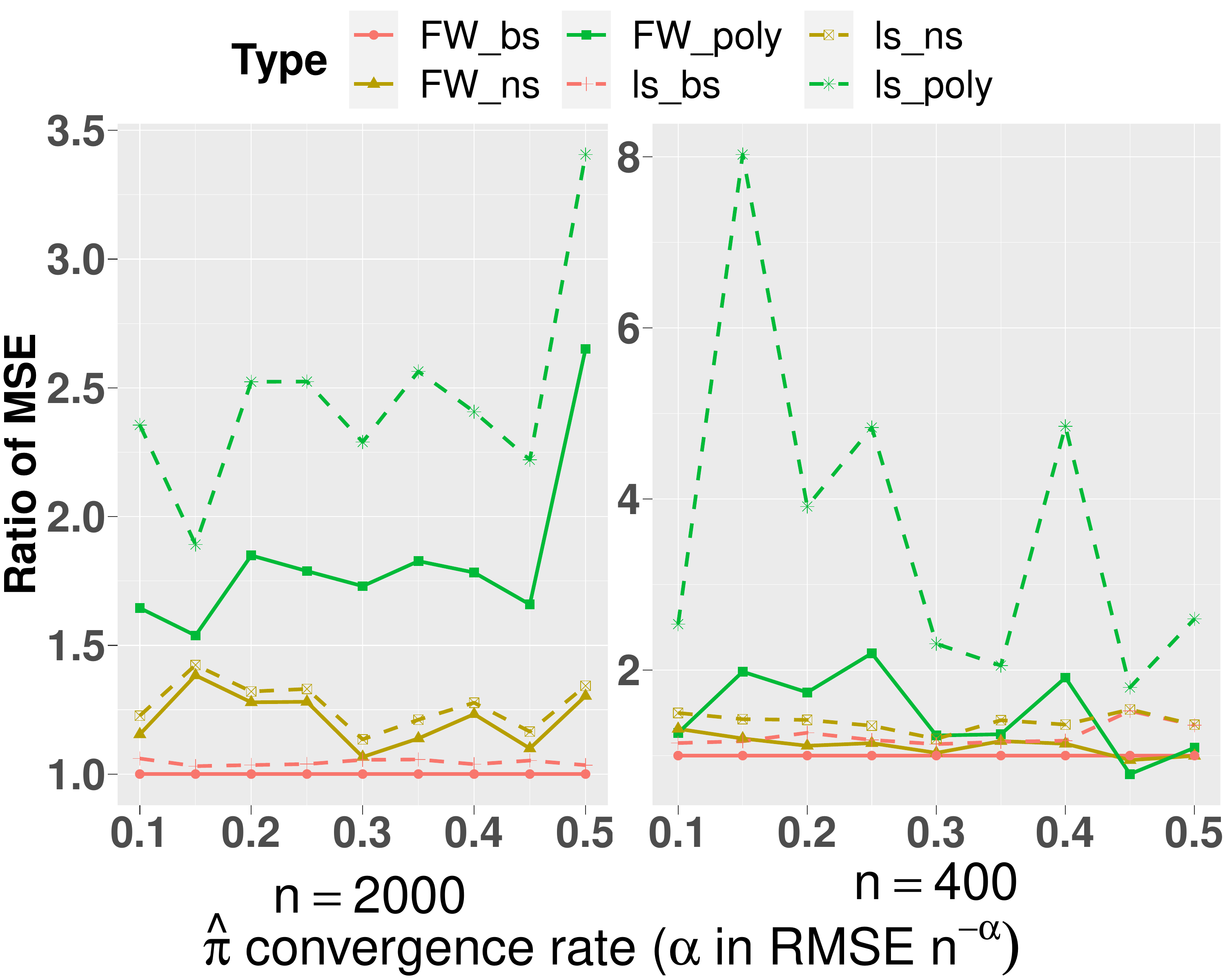}
  \caption{Ratio of MSE of different estimators when $X$ is a heavy-tailed distribution. The baseline method is $\text{FW}\_\text{bs}$. Left: sample size $n=2000$; Right: $n=400$.}
  \label{fig:forster-ls}
\end{subfigure}
\caption{Comparison between different estimators}
\end{figure}


\section{Data Application: CATE of Right Heart Catherization}\label{sec:real}
We illustrate the proposed FW-Learner with an application of CATE estimation with and without assuming unconfoundedness, where in the latter we make use of proximal causal inference. Specifically, we reanalyze the Study to Understand Prognoses and Preferences for Outcomes and Risks of Treatments (SUPPORT) to evaluate the causal effect of Right Heart Catheterization (RHC) during the initial care of critically ill patients in the intensive care unit (ICU) on survival time up to 30 days \citep{connors1996effectiveness}. \cite{tchetgen2020introduction} and \cite{cui2020semiparametric} analyzed this dataset to estimate the marginal average treatment effect of RHC,  using the proximal causal inference framework, with an implementation of a locally efficient doubly robust estimator, using parametric estimators of the bridge functions. Data are available on 5735 individuals, 2184 treated, and 3551 controls.  {\color{black}Before the study, a panel of specialists specified the variables related significantly to the decision to whether or not to use a right heart catheter, which were included in a multivariable logistic regression for RHC in the initial 24 hours to determine the propensity score for each patient in the dataset, see \cite{conners1996effectiveness} for details.} In total, 3817 patients survived and 1918 died within 30 days. The outcome $Y$ is the number of days between admission and death or censoring at day 30. We include all 71 baseline covariates to adjust for potential confounding. To implement the FW-Learner under unconfoundedness,  the nuisance functions $\pi^\star$, $\mu_0^\star$ and $\mu_1^\star$ are estimated using SuperLearner\footnote{SuperLearner is a stacking ensemble machine learning approach that uses cross-validation to estimate the performance of multiple machine learners and then creates an optimal weighted average of those models using test data. This approach has been formally established to be asymptotically as accurate as the best possible prediction algorithm that is tested. For details, please refer to \cite{Polley2010SuperLI}.} that includes both RandomForest and Generalized Linear Model (GLM).

\paragraph{Variance of the FW-Learner:}
In addition to producing an estimate of the CATE, one may wish to quantify uncertainty based on this estimate. We describe a simple approach for computing standard error for the CATE at a fixed value of $x$ and corresponding pointwise confidence intervals. The asymptotic guarantee of the confidence intervals for the least squares estimator is established in \cite{newey1997convergence} and \cite{belloni2015some} under some conditions. 
Because the FW-Learner is asymptotically equivalent to the Least-squares estimator, the same variance estimator as that of the least-squares series estimator may be used to quantify uncertainty about the FW-Learner. Recall that the least-squares estimator is given by
$\widebar\phi(x)^\top  \bigl[ \sum_{i}  \widebar\phi (X_{i}) 
 \widebar\phi(X_{i})^\top  \bigr]^{-1} \bigl\{ \sum_{i}  \widebar\phi (X_{i}) \widehat{I}_i \bigr\}$, the latter has variance $\widebar\phi(x)^\top \bigl[ \sum_{i}  \widebar\phi (X_{i}) 
 \widebar\phi(X_{i})^\top  \bigr]^{-1}  \widebar\phi(x)  \times \sigma^2 (\widehat{I}) $, where  $ \sigma^2 (\widehat{I}) $ is the variance of the pseudo-outcome $\widehat{I}$; where we have implicitly assumed homoscedasticity, i.e. that the conditional variance of $(\widehat{I})$ is independent of $X$. Hence, 
 \begin{align}
     \mathrm{var} (\widehat{\tau}(x)) \approx \widebar\phi(x)^\top \bigl[ \sum_{i} \widebar\phi (X_{i}) 
 \widebar\phi(X_{i})^\top  \bigr]^{-1}  \widebar\phi(x)  \times \sigma^2 (\widehat{I}).
 \end{align}


Similar to \cite{tchetgen2020introduction} and \cite{cui2020semiparametric}, our implementation of the Proximal FW-Learner specified baseline covariates $ (\text{age, sex, cat1 coma}$, cat2 coma, dnr1, surv2md1, aps1)\footnote{Variable description can be found at \hyperlink{https://hbiostat.org/data/repo/rhc}{https://hbiostat.org/data/repo/rhc}.} for confounding adjustment; as well as treatment and outcome confounding proxies $Z = (\text{pafi1}, \text{paco21})$ and $W = (\text{ph1}, \text{hema1})$. Confounding bridge functions were estimated nonparametrically using the adversarial Reproducing Kernel Hilbert Spaces (RKHS) learning approach of  \cite{ghassami2022minimax}. The estimated CATE and corresponding pointwise 95 percent confidence intervals are reported in Figure \ref{fig:forster_cv_poly_proxy_split} as a function of the single variable measuring the 2-month model survival prediction at data 1 (surv2md1), for both approaches, each using both splines and polynomials.   Cross-validation was used throughout to select the number of knots for splines and the degree of the polynomial bases, respectively. The results are somewhat consistent for both basis functions, and suggest at least under unconfoundedness conditions that high risk patients likely benefited most from RHC, while low risk patients may have been adversely impacted by RHC. In contrast, the Proximal FW-Learner produced a more attenuated CATE estimate, which however found  that RHC was likely harmful for low risk patients. Interestingly, these analyses provide important nuances to results reported in the original analysis of \cite{connors1996effectiveness} and the more recent analysis of \cite{tchetgen2020introduction} which concluded that RHC was harmful on average on the basis of the ATE.  
\begin{figure}[H]
    \centering
    \includegraphics[width=0.5\textwidth]{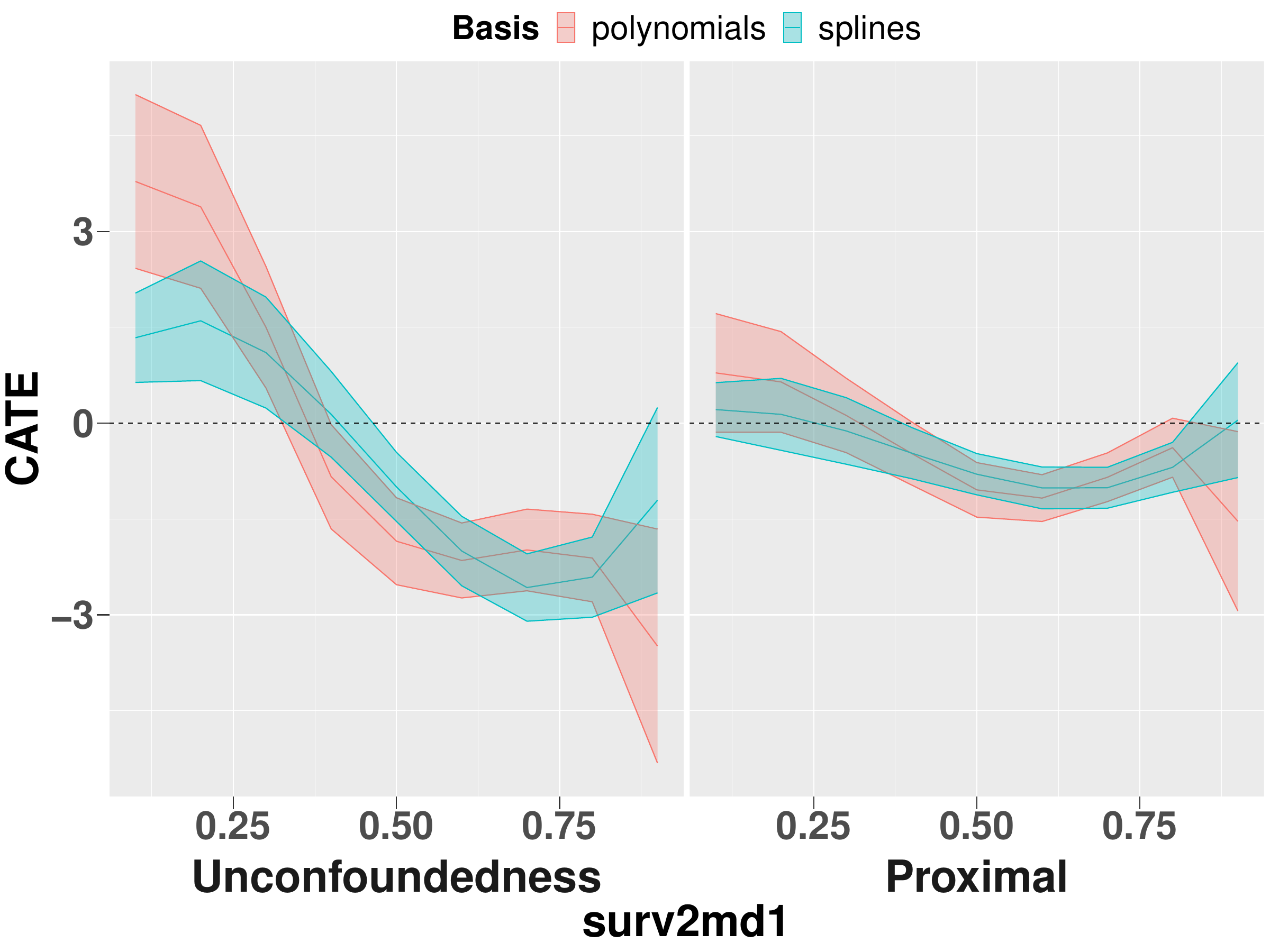}
    \caption{CATE estimation with $95\%$ confidence interval produced by the FW-Learner using polynomial and spline basis. Left: under unconfoundedness; Right: in proximal causal inference setting.}
    \label{fig:forster_cv_poly_proxy_split}
\end{figure}

\section{Discussion}
This paper has proposed a novel nonparametric series estimator of regression functions that requires minimal assumptions on covariates and basis functions. Our method builds on the Forster--Warmuth estimator, which incorporates weights based on the leverage score $h_n(x) = x^\top (\sum_{i=1}^n X_i X_i^\top  + x x^\top)^{-1} x $, to obtain predictions that can be significantly more robust relative to standard least-squares, particularly in small to moderate samples. Importantly, the FW-Learner is shown to satisfy an oracle inequality with its excess risk bound having the same order as $J\sigma^2/n$, requiring only the relatively mild assumption of bounded outcome second moment ($\E[Y^2\mid X]  \leq \sigma^2$). 
Recent works \citep{mourtada2019exact, vavskevivcius2023suboptimality} investigate the potential for the risk of standard least-squares to become unbounded when leverage scores are uneven and correlated with the residual noise of the model. By adjusting the predictions at high-leverage points, which are most likely to lead to an unstable estimator, the Forster--Warmuth estimator mitigates the shortcomings of the least squares estimator and achieves oracle bounds even for unfavorable distributions when least squares estimation fails.
The Forster--Warmuth algorithm leads to the only known exact oracle inequality without imposing any assumptions on the covariates. 

Another major contribution we make is to propose a general method for counterfactual nonparametric regression via series estimation in settings where the outcome may be missing. Specifically, we generalize the FW-Learner using a generic pseudo-outcome that serves as substitution for the missing response and we characterize the extent to which accuracy of the pseudo-outcome can potentially impact the estimator's ability to match the oracle minimax rate of estimation on the MSE scale. We then provide a generic approach for constructing a pseudo-outcome with ``small bias" property for a large class of counterfactual regression problems, based on a doubly robust influence functions of the functional obtained via marginalizing the counterfactual regression in view. This insight provides a constructive solution to the counterfactual regression problem and offers a unified solution to several open nonparametric regression problems in both  missing data and causal inference literatures. The versatility of the approach is demonstrated by considering estimation of nonparametric regression when the outcome may be MAR; or when the outcome may be Missing Not At Random by leveraging a shadow variable. As well as by considering estimation of the CATE under standard unconfoundedness conditions; and when hidden confounding bias cannot be ruled out on the basis of measured covariates, however proxies of unmeasured factors are available that can be leveraged using proximal causal inference framework.  While some of these settings such as CATE under unconfoundedness have been studied extensively, others such as the CATE under proximal causal inference have only recently developed.

Overall, this paper brings together aspects of traditional linear models, nonparametric models, and modern literature on semiparametric theory, with applications in different contexts. This marriage of classical and modern techniques is in similar spirit as recent frameworks such as orthogonal learning \citep{foster2023orthogonal}, however, our assumptions and approach appear to be fundamentally different in that, at least for specific examples considered herein, our assumptions are somewhat weaker yet lead to a form of oracle optimality.  We nevertheless  believe that both frameworks open the door to many future exciting directions to explore.  
A future line of investigation might be to extend the estimator using more accurate pseudo-outcomes of the unobserved response using recent theory on higher order influence functions \citep{robins2008higher,robins2017minimax}, along the  lines of \cite{kennedy2022minimax} who constructs minimax estimators of the CATE under unconfoundness conditions and weaker smoothness conditions on the outcome and propensity score models, however requiring considerable restrictions on the covariate distribution.  
  Another interesting direction is the potential application of our methods to more general missing data settings, such as monotone or nonmonotone coarsening at random settings \citep{robins1994estimation, laan2003unified, tsiatis2006semiparametric}, and corresponding coarsening not at random settings, e.g. \cite{robins2000sensitivity}, \cite{tchetgen2018discrete}, \cite{malinsky2022semiparametric}. 
We hope the current manuscript provides an initial step towards solving this more challenging class of problems and  generates both interest and further developments in these fundamental directions.

\bibliographystyle{plainnat}
\bibliography{ref}

\newpage
\setcounter{section}{0}
\setcounter{equation}{0}
\setcounter{figure}{0}
\renewcommand{\thesection}{S.\arabic{section}}
\renewcommand{\theequation}{E.\arabic{equation}}
\renewcommand{\thefigure}{A.\arabic{figure}}
\setcounter{page}{1}
  \begin{center}
  \Large {\bf Supplement to ``Forster--Warmuth Counterfactual Regression: A Unified Approach''}
  \end{center}
       
\begin{abstract}
This supplement contains the proofs of all the main results in the paper and some supporting lemmas. 
\end{abstract}

\section{Some definitions}\label{supp:notation}
We call a function $\alpha$-smooth if it belongs to the class of H\"older smoothness order $\alpha$.
Formally, let $k=(k_1, k_2, \dots, k_d)$ be a $d$-dimensional index set where each $k_i$ is a non-negative integer and $|k|=\sum_{i=1}^d k_i$. For each $f : \Omega \mapsto \mathbb{R}$ where $x = (x_1, x_2, \dots, x_d) \in \Omega\subseteq \mathbb{R}^d$, differentiable up to the order $k \ge 1$, we define the differential operator $D^k$ as
\begin{equation*}
    D^k f = \frac{\partial^{|k|} f(x)}{\partial^{k_1} x_1\dots \partial^{k_d} x_d}\, \text{ and }\, D^0f = f.
\end{equation*}
 For $\alpha, L > 0$, the H\"{o}lder class $\Sigma_d(\alpha,L)$ on $\Omega$ consists of functions that satisfy the following condition:
\begin{align*}
    &\Sigma_d(\alpha, L) :=
    \left\{f: \Omega \mapsto \mathbb{R}\, \bigg|\, \sum_{0 \le |m| \le \lfloor\alpha\rfloor}\|D^m f\|_\infty +\sum_{|k| = \lfloor\alpha\rfloor}\sup_{x\neq y,\, x,y\in \Omega}\frac{|D^{k} f(y)-D^{k} f(x)|}{\|x-y\|^{\alpha-|k|}}\le L\right\}.
\end{align*}
Following \cite{gine2021mathematical}, for $1 \leq p<\infty$, the $L^p$-Sobolev space of order $m \in \mathbb{N}$ is defined as
$$
W_d(p, m)=\left\{f \in L^p: D^j f \in L^p(\cdot) ~ \forall j=1, \ldots, m: \|f\|_p+\left\|D^m f\right\|_p<\infty\right\}.
$$
\section{More Examples of Pseudo-outcome}\label{sec-app:examples}
\subsection{FW-Learner under MNAR: shadow variables}\label{sec:mnar}
In the Section \ref{sec-app:forster-missing}, we constructed an FW-Learner for a nonparametric mean regression function under MAR. The MAR assumption may be violated in practice, for instance if there are unmeasured factors that are both predictive of the outcome and nonresponse, in which case outcome data are said to be missing not at random and the regression may generally not be identified from the observed data only. 
In this section, we continue to consider the goal of estimating  a nonparametric regression function, however allowing for outcome data to be missing not at random, by leveraging a so-called shadow variable for identification \citep{miao2015identification}. In contrast to the MAR setting, the observed data we consider here is $O_i = (X_i, W_i, R_i, Y_iR_i), 1\le i\le n$, where $W_i$ is the shadow variable allowing identification of the conditional mean. Specifically, a shadow variable is a fully observed variable, that is (i) associated with the outcome given  fully observed covariates and (ii) is independent of the missingness process conditional on fully observed covariates and the possibly unobserved outcome variable. 
Formally, a shadow variable $W$ has to satisfy the following assumption.
\begin{enumerate}[label={\bf (SV)}]
    \item $W \perp R \bigm| (X,Y)$ and $W \not\perp Y \bigm| X$.\label{assump:shadow-variable}
\end{enumerate}
This assumption formalizes the idea that the missingness process may depend on $(X, Y)$, but not on the shadow variable $W$ after conditioning on $(X, Y )$ and therefore, allows for missingness not at random.\footnote{The assumption can be generalized somewhat, by further conditioning on fully observed covariates $Z$ in addition to $X$ and $Y$ in the shadow variable conditional independence statement, as well as in the following identifying assumptions.} 
 Under this condition, it holds (from Bayes' rule) that
\begin{equation}\label{eq:EPS}
\E \Bigr\{ \frac{1}{\PP(R=1|X, Y)} \Bigm| R = 1, X,W \Bigr\} = \frac{1}{\PP(R=1|X, W)}.
\end{equation}
Let $e^\star(X,Y) := \PP [R=1|X,Y]$ denote the \emph{extended} propensity score, which consistent with MNAR, will generally depend on $Y$. Likewise, let $\pi^\star(X,W) := \PP [R=1|X,W]$. Clearly $e^\star(X,Y)$ cannot be estimated via standard regression of $R$ on $X,Y$ given that $Y$ is not directly observed for units with $R=0$. Identification of the extended propensity score follows from the following completeness condition 
\citep{miao2015identification, tchetgen2023single}: define the map $D: L_2 \to L_2$ by $[Dg](x, w) = \E\big\{ g(X,Y)|R=1, X = x, W = w\bigr\}$. 
\begin{enumerate}[label={\bf(CC)}]
    \item $[Dg](X, W) = 0$ almost surely if and only if $g(X,Y)=0$ almost surely.\label{assump:completeness-condition}
\end{enumerate}

 Given a valid shadow variable, suppose also that there exist a so-called outcome confounding bridge function that satisfies the following condition  \citep{li2021identification, tchetgen2023single}.
\begin{enumerate}[label = {\bf (BF)}]
\item There exists a function $\eta^\star(x,w)$ that satisfies the integral equation\label{assump:bridge-function}
 \begin{equation}\label{eq:bridge}
 y = \E \{\eta^\star(X,W)|Y=y, X=x,R=1\}.
 \end{equation}
\end{enumerate}
The assumption may be viewed as a nonparametric measurement error model, whereby the shadow variable $W$ can be viewed as an error-prone proxy or surrogate measurement of $Y$, in the sense that there exists a transformation (possibly nonlinear) of $W$ which is conditionally unbiased for $Y$. In fact, the classical measurement model which posits $W =Y+\epsilon$ where $\epsilon$ is a mean zero independent error clearly satisfies the assumption with $\eta^\star$ given by the identity map.  \cite{li2021identification} formally established that existence of a bridge function satisfying the above condition is a necessary condition for pathwise differentiation of the marginal mean $\E (Y)$ under the shadow variable model, and therefore, a necessary condition for the existence of a root-n estimator for the marginal mean functional in the shadow variable model. From our viewpoint, the assumption is sufficient for existence of a pseudo-outcome with second order bias. 

  Let $\widehat{e}(\cdot)$ denote a consistent estimator of  $e^\star(\cdot)$ that solves an empirical version of its identifying  equation \eqref{eq:EPS}. Similarly, let $\widehat{\eta}(\cdot)$ be an estimator for $\eta^\star(\cdot)$ that solves an empirical version of the integral equation \eqref{eq:bridge}; see e.g. \cite{ghassami2022minimax}, \cite{li2021identification} and  \cite{tchetgen2023single}.  
Following the pseudo-outcome construction of Section \ref{sec:pseudo}, the proposed shadow variable oracle pseudo-outcome follows from the (uncentered) locally efficient influence function of the marginal outcome mean $\E(Y)$  under the shadow variable model, 
given by $f(O) = {R} Y/ {e^\star (X,Y)}- \bigl( {R}/{{e^\star}(X, Y)} - 1 \bigr) {\eta^\star} (X, W);$ see \cite{li2021identification}, \cite{ghassami2022minimax}, and \cite{tchetgen2023single}. It is easily verified that $\mathbb{E}[f(O)|X = x] = m^\star(x)$ under~\ref{assump:shadow-variable},~\ref{assump:completeness-condition}, and~\ref{assump:bridge-function}. 
Note that this pseudo-outcome is a member of the mixed-bias class of influence functions \eqref{eq:general-if} with $h^\star = 1/e^\star$, $q^\star = \eta^\star, g_1 = -R, g_2=1, g_3 = RY$ and $g_4=0$. 
The corresponding empirical pseudo-outcome is given by
\begin{equation}\label{eq:forster-imputation-missing-MNAR}
\begin{split}
\widehat{f}(O) &= \frac{R}{\widehat{e}(X, Y)}Y - \left(\frac{R}{\widehat{e}(X, Y)} - 1\right)\widehat{\eta}(X, W),
\end{split}
\end{equation}
with $\widehat{e}(\cdot, \cdot)$ and $\widehat{\eta}(\cdot, \cdot)$ obtained from the first split of the data.

Following~\ref{step-B-FW-Learner-counterfactual}, we obtain the FW-Learner $\widehat{m}_J(X)$. In practice, similar to Algorithm \ref{alg:split-mar}, cross-validation may be used to tune the truncation parameter $J$.
Set $H_{f}(x) = \E[\widehat{f}(O) |X=x, \widehat{f}]$.  The following lemma gives the form of the mixed-bias for $\widehat{f}(\cdot)$.
\begin{lem}\label{lem:forster-missing-MNAR}
Under~\ref{assump:shadow-variable},~\ref{assump:completeness-condition},~\ref{assump:bridge-function}, the pseudo-outcome \eqref{eq:forster-imputation-missing-MNAR} satisfies
\begin{align}
H_{f}(x) - m^\star(x) ~=~ \E \biggl\{   R\biggl(\frac{1}{\widehat{e}(X, Y)} - \frac{1}{e^\star(X, Y)}\biggr)   ( \eta^\star- \widehat{\eta} )(X, W)  \Bigm| X=x, \widehat{e}, \widehat{\eta} \biggr\}. 
\label{eq:I1-bias-nmar}
\end{align}
\end{lem}
This result directly follows from the mixed bias form \eqref{eq:dr-bias-general} in the general class studied by \cite{ghassami2022minimax} in the shadow variable nonparametric regression setting. The proof is given in Section \ref{sec-app:forster-missing-MNAR} of the supplement. Plugging this into Corollary \ref{cor:forster-pseudo}  leads to the error rate of the FW-Learner $\widehat{m}_J(x)$.
\begin{thm}\label{thm:forster-missing-MNAR} 
Under the same notation as Theorem \ref{thm:forster-missing}, and under~\ref{assump:shadow-variable},~\ref{assump:completeness-condition},~\ref{assump:bridge-function}, the FW-Learner $\widehat{m}_J(x)$ satisfies 
\begin{align}\label{eq:forster-missing-MNAR}
(\mathbb{E}[(\widehat{m}_J(X) - m^\star(X))^2|\widehat{f}])^{1/2} &\le \sqrt{\frac{2\sigma^2J}{|\mathcal{I}_2|}} + \sqrt{2\kappa}E_J^{\Psi}(m^\star)
\\ 
\quad
+ \sqrt{6} \min \Bigl\{ & \Bigl\|\frac{1}{\widehat{e}(X, Y)} - \frac{1}{e^\star(X, Y)}\Bigl\|_4  ~ \Bigl\| \E \bigl[ \bigl( \eta^\star - \widehat{\eta} \bigr) (X, W) \bigm| X,Y \bigr] \Bigr\|_4, \\
& \Bigl\|\E \bigl[ \frac{1}{\widehat{e}(X, Y)} - \frac{1}{e^\star(X, Y)} \bigm| X,W \bigr]\Bigl\|_4  ~ \Bigl\|  \bigl( \eta^\star - \widehat{\eta} \bigr) (X, W) \Bigr\|_4 \Bigr\}
\end{align}
\end{thm}
The proof of this result is in Section \ref{sec-app:forster-missing-MNAR} of the supplement. Note that $\sigma^2$ is finite when $\widehat{\eta}$ and $1/\widehat{e}$ are bounded.
Theorem~\ref{thm:forster-missing-MNAR} demonstrates that the FW-Learner performs nearly as well as the Oracle learner with a slack of the order of the mixed bias of estimated nuisance functions for constructing the pseudo-outcome. Unlike the MAR case, the nuisance functions under the shadow variable assumption are not just regression functions and hence, the rate of estimation of these nuisance components is not obvious. In what follows, we provide a brief discussion of estimating these nuisance components. 
Focusing on the outcome confounding bridge function which solves equation \eqref{eq:bridge}, 
this equation is a so-called Fredholm integral equation of the first kind, which are well known to be ill-posed \citep{kress1989linear}. Informally, ill-posedness  essentially measures the extent to which the conditional expectation defining the kernel of the integral equation $Q \mapsto \mathbb{E}_Q\left[\eta\left(X_i, W_i\right) \mid X_i=x, Y_i = y \right]$ smooths out $\eta$. Let $L_2(X)$ denote the class of functions $\{f: \E_{X} [f^2(X)] \leq \infty\}$, and define the operator $T: L_2(X,W) \rightarrow L_2(X,Y)$ as the conditional expectation operator given by
$$
[T \eta](x,y):=\E\left[\eta\left(X_i, W_i\right) \mid X_i=x, Y_i = y\right] .
$$
Let $\Psi_J:=\operatorname{clsp}\left\{\psi_{J 1}, \ldots, \psi_{J J}\right\} \subset L_2(X, W)$ denote a sieve spanning the space of functions of variables $X,W$. One may then define a corresponding sieve $L_2$ measure of ill-posedness coefficient as in \cite{blundell2007semi} as $\tau_\eta:=\sup _{\eta \in \Psi_J: \eta \neq 0} {\|\eta\|_{L_2(X,W)}}/{\|T \eta\|_{L_2(X,Y)}}.$
\begin{definition}[Measure of ill-posedness]\label{def:ill-posedness}
    Following \cite{blundell2007semi}, the integral equation \eqref{eq:proximal-h}  
    with $(W_i, X_i)$ of dimension $(d_x+d_w)$ is said to be 
    \begin{enumerate}
        \item mildly ill-posed if $\tau_h=O\left(J^{\varsigma_h / (d_x+d_w)}\right)$ for some $\varsigma_h>0$;
\item severely ill-posed if $\tau_h=O\left(\exp \left(\frac{1}{2} J^{\varsigma_h/ (d_w+d_x)}\right)\right)$ for some $\varsigma_h>0$.
    \end{enumerate}

\end{definition}



Under the condition that integral equation \eqref{eq:bridge} is mildly ill-posed and that $\eta^\star$ is $\alpha_\eta$-H{\"o}lder smooth, \cite{chen2018optimal} established that the optimal rate for estimating $\eta^\star$ under the $\sup$ norm is $(n / \log n)^{-\alpha_h /(2(\alpha_\eta+\varsigma_\eta)+d_x+d_w)}$; see Lemma \ref{lem:lb-shadow} in the supplement for details.  Likewise, 
the integral equation \eqref{eq:EPS} is also a Fredholm integral equation of the first kind with its kernel given by the conditional expectation operator $[T' e](x,w):=\E\left[e (X_i, Y_i ) \mid X_i=x, W_i = w\right]$ for any function $u \in L_2(X,Y) $, and $T'$ is the adjoint operator of $T$.
Let $\Psi_J':=\operatorname{clsp}\left\{\psi_{J 1}', \ldots, \psi_{J J}'\right\} \subset L_2(X, Y)$ denote a (different) sieve spanning the space of functions of variables $X,Y$.
Its corresponding sieve $L_2$ measure of ill-posedness may be defined as $\tau_e=\sup _{o \in \Psi_J: o \neq 0} {\|o\|_{L_2(X,Y)}}/{\|T o\|_{L_2(X,W)}}.$
Thus in the mildly ill-posed case $\tau_e=O\left(J^{\varsigma_e / (d_x+1)}\right)$ for some $\varsigma_e>0$, the optimal rate with respect to the sup norm for estimating $e^\star$ is $(n / \log n)^{-\alpha_e /(2(\alpha_e+\varsigma_e)+d_x+1)}$ when $e^\star$ is $\alpha_e$-smooth and bounded.


Together with
\eqref{eq:forster-missing-MNAR}, this leads to the following characterization of the error of the FW-Learner $\widehat{m}_J(X)$ if $E_J^{\Psi}(m^\star) \lesssim J^{-\alpha_m/d_x}$. Without loss of generality, suppose that
\begin{align}\label{eq:FW_SVrate}
  \min \Bigl\{ & \Bigl\|\frac{1}{\widehat{e}(X, Y)} - \frac{1}{e^\star(X, Y)}\Bigl\|_4  ~ \Bigl\| \E \bigl[ \bigl( \eta^\star - \widehat{\eta} \bigr) (X, W) \bigm| X,Y \bigr] \Bigr\|_4, \\
& \Bigl\|\E \bigl[ \frac{1}{\widehat{e}(X, Y)} - \frac{1}{e^\star(X, Y)} \bigm| X,W \bigr]\Bigl\|_4  ~ \Bigl\|  \bigl( \eta^\star - \widehat{\eta} \bigr) (X, W) \Bigr\|_4 \Bigr\}
\\
&=\Bigl\|\E \bigl[ \frac{1}{\widehat{e}(X, Y)} - \frac{1}{e^\star(X, Y)} \bigm| X,W \bigr]\Bigl\|_4  ~ \Bigl\|  \bigl( \eta^\star - \widehat{\eta} \bigr) (X, W) \Bigr\|_4 ,
\end{align}
and suppose that $\pi^\star$ is $\alpha_\pi$-H{\"o}lder smooth, such that 
\begin{align}
&\Bigl\|\E \bigl[ \frac{1}{\widehat{e}(X, Y)} - \frac{1}{e^\star(X, Y)} \bigm| X,W \bigr]\Bigl\|_4 \\
&=\Bigl\|\E \bigl[ \frac{1}{\widehat{e}(X, Y)} \bigm| X,W \bigr] - \frac{1}{\pi^\star(X, W)}\Bigl\|_4
\end{align}
is of the order of $n^{-\alpha_\pi /(2\alpha_\pi +d_x+d_w)}$ the minimax rate of estimation of the regression function $\pi^\star$.
\begin{cor}\label{cor:fw-shadow}
Under the conditions in Lemma \ref{lem:lb-shadow} in the supplement and assuming that the linear operator $T$ is mildly ill-posed with exponent $\varsigma_\eta$; then if $m^\star$ satisfies $E_J^{\Psi}(m^\star) \lesssim J^{-\alpha_m/d_x}$, $\pi^\star$ is $\alpha_\pi$-H{\"o}lder smooth and $\eta^\star$ is $\alpha_\eta$-H{\"o}lder smooth, and equation \eqref{eq:FW_SVrate} holds, then the FW-Learner's estimation error satisfies 
\begin{align}\label{eq:fw-shadow}
    \bigl\|\widehat{m}_J(X) - m^\star(X) \bigr\|_2 \lesssim \sqrt{\frac{\sigma^2 J}{n}} + J^{-\alpha_m/d_x}  + 
(n / \log n)^{-\alpha_\eta /(2(\alpha_\eta +\varsigma_\eta)+d_x+d_w)}n^{-\alpha_\pi /(2\alpha_\pi +d_x+d_w)}.
\end{align}
\end{cor}
A remark analogous to Remark \ref{rem:proxy} equally applies to Corollary \ref{cor:fw-shadow}.
\begin{rem}\label{rem:mnar}
    A few remarks on Corollary \ref{cor:fw-shadow}:  (1) If the mixed bias term incurred for estimating nuisance functions is negligible relative to the first two terms in \eqref{eq:fw-shadow}, then the order of the error of the FW-Learner matches that of the oracle with access to missing data; (2) In settings where operators $T, T'$
    are severely ill-posed, i.e. where  $\tau_\eta=O\left(\exp \left(\frac{1}{2} J^{\varsigma_\eta / (d_x+d_w)}\right)\right)$ for some $\varsigma_\eta>0$, Theorem 3.2 of \cite{chen2018optimal} established that the optimal rate of estimating $\eta^\star$ with respect to the sup norm is of the order $(\log n)^{-\alpha_\eta/\varsigma_\eta}$ which would likely dominate the error $\|\widehat{m}_J - m^\star \bigr\|_2$. In this case, the FW-Learner may not be able to attain the oracle rate. In this case, whether the oracle rate is at all attainable remains an open problem in the literature. 
\end{rem}
\subsection{FW-Learner for CATE under Ignorability}\label{sec:unconfoundness}
In this section, we make the additional assumption of unconfoundedness, so that the treatment mechanism is ignorable. 

\textbf{No unmeasured confounding Assumption:} $(Y^0, Y^1) \perp A | X$. Under this condition, the CATE is nonparametrically identified by  $\tau^\star(x) = \mu_1(x) - \mu_0(x)$, where for $a \in \{0, 1\}$,
\[
\mu^\star_a(x) := \mathbb{E}[Y|X = x, A = a];
\]
Let $\pi^\star(x) := \mathbb{P}(A = 1|X = x)$.  
We will now define the Forster--Warmuth estimator for CATE. Split $\{1, 2, \ldots, n\}$ into two parts $\mathcal{I}_1$ and $\mathcal{I}_2$. Based on $(X_i, A_i, Y_i), i\in\mathcal{I}_1$, estimate $\pi^\star, \mu_0^\star, \mu_1^\star$ with $\widehat{\pi}, 
 \widehat{\mu}_0,  \widehat{\mu}_1$, respectively. For $i\in\mathcal{I}_2$, define the pseudo-outcome 
\[
\widehat{I}_1(X_i, A_i, Y_i) = \frac{A_i - \widehat{\pi}(X_i)}{\widehat{\pi}(X_i)(1 - \widehat{\pi}(X_i))}(Y_i - \widehat{\mu}_{A_i}(X_i)) + \widehat{\mu}_1(X_i) - \widehat{\mu}_0(X_i),
\]
which is an estimator of well-known (uncentered) efficient influence function of the marginal average treatment effect $\E(Y^1-Y^0)$, evaluated at preliminary estimates of nuisance functions, 
and is in our general mixed-bias class of influence functions given by \eqref{eq:general-if} with $h_0(O_h) = \mu_W^\star(X), q_0(O_q) = 1/\pi^\star(X), g_1(O) = -\mathbbm{1}\{A=a\}, g_2(O) = \mathbbm{1}\{A=a\}Y, g_3(O)=1$ and $g_4(O)=0$.
Write
\[
H_{I_1}(x) = \mathbb{E} \Bigl[ \widehat{I}_1(X, A, Y)|X = x \Bigr]. 
\]

We first provide a characterization of the conditional bias of the pseudo-outcome in the following lemma.
\begin{lem}\label{lem:forster-cate}
The conditional bias of the pseudo outcome $\widehat{I}_1(X_i, A_i, Y_i)$
\begin{align} 
H_{I_1}(x) - \tau^\star(x) &= \pi^\star(x)\Bigl(\frac{1}{\widehat{\pi}(x)} - \frac{1}{\pi^\star(x)} \Bigr) \bigl( \widehat{\mu}_1(x) - \mu^\star_1(x) \bigr) \nonumber\\
&\quad - (1 - \pi^\star(x))\Bigl( \frac{1}{1 - \widehat{\pi}(x)} - \frac{1}{1 - \pi^\star(x)} \Bigr) \bigl( \widehat{\mu}_0(x) - \mu^\star_0(x)  \bigr).
\end{align}
\end{lem}
This result directly follows from the mixed bias form \eqref{eq:dr-bias-general} which recovers a well-know result in the literature, originally due to Robins and colleagues; also see \cite{kennedy2023towards}. For convenience, the proof is reproduced in Section \ref{sec-app:forster-cate} of the supplement.
Let $\widehat{\tau}_J(x)$ be the Forster--Warmuth estimator computed from $\{ (\widebar{\phi}_J(X_i), \widehat{I}_1(X_i, A_i, Y_i)), i\in\mathcal{I}_2\}.$

We establish our first oracle result of the FW-Learner of the CATE.

\begin{thm} Under the assumptions given above, including unconfoundedness, suppose that $\sigma^2$ is an upper bound for $\mathbb{E}[\widehat{I}_1^2(X, A, Y) \mid X]$, then FW-Learner $\widehat{\tau}_J(x)$ satisfies the error bound 
\begin{align}
    \bigl\|\widehat{\tau}_J(X) &- \tau^\star(X) \bigr\|_2 \le \sqrt{\frac{2\sigma^2J}{|\mathcal{I}_2|}} + \sqrt{2}\Bigl\|\sum_{j=J+1}^{\infty} \theta_j^\star\phi_j(X) \Bigr\|_2 \nonumber\\ 
    &\quad+ (1+\sqrt{2})\biggl( \Bigl\|\frac{\pi^\star(X)}{\widehat{\pi}(X)} - 1 \Bigr\|_4 \Bigl\|\widehat{\mu}_1(X) - \mu^\star_1(X) \Bigr\|_4 + \Bigl\|\frac{1 - \pi^\star(X)}{1 - \widehat{\pi}(X)} - 1\Bigr\|_4 \Bigl\|\widehat{\mu}_0(X) - \mu^\star_0(X) \Bigr\|_4 \biggr).
\end{align}
\label{thm:forster-cate}
\end{thm}
See Section \ref{sec-app:forster-cate} in the supplement for a formal proof of this result. Note that the condition that $\sigma^2$ is bounded requires $\widehat{\mu}_0$, $\widehat{\mu}_1, 1/\widehat{\pi}$ and $1/(1-\widehat{\pi})$ to be bounded.

\begin{cor}\label{cor:missing-ignorability}
Let $d$ denote the intrinsic dimension of  $X$. If
\begin{enumerate}
    \item The propensity score $\pi^\star(x, z)$ is estimated at an $n^{-2 \alpha_\pi /(2 \alpha_\pi+d)}$ rate in the $L_4$-norm;\item The regression functions $\mu_0^\star$ and $\mu_1^\star$ are estimated at the rate of $n^{-2 \alpha_\mu /(2 \alpha_\mu+d)}$ in the $L_4$-norm.
\item The CATE $\tau^\star$ with respect to the fundamental sequence $\Psi$ satisfies $E_J^{\Psi}(\tau^\star) \le CJ^{-\alpha_\tau/d}$ for some constant $C$,
\end{enumerate}
Then, $\widehat{\tau}_J(x)$ satisfies
\begin{align}\label{eq:final-forster-ignorability}
   \Bigl(\mathbb{E}[(\widehat{\tau}_J(X) - \tau^\star(X))^2|\widehat{\pi}, \widehat{\mu}] \Bigr)^{1/2}
    & \lesssim \sqrt{\frac{\sigma^2 J}{n}} + J^{-\alpha_\tau/d}  + n^{-\frac{\alpha_\pi}{2\alpha_\pi+d} - \frac{\alpha_\mu}{2\alpha_\mu+d}}.
\end{align}
\end{cor}
When the last term of \eqref{eq:final-forster-ignorability} is smaller than the oracle rate $n ^{-\frac{\alpha_\tau}{2\alpha_\tau+d}}$, the oracle minimax rate can be attained by balancing the first two terms. Therefore, the FW-Learner is oracle efficient if $\alpha_\mu \alpha_\pi \geq {d^2}/{4}-{(\alpha_\pi+\frac{d}{2})(\alpha_\mu+\frac{d}{2})}/{(1+\frac{2 \alpha_\tau}{d})}$. 
In the special case when $\alpha_\mu$ and $\alpha_\pi$ are equal, if we let $s = \alpha_\mu/d = \alpha_\pi/d$ and $\gamma = \alpha_\tau/d$ denote the effective smoothness, and when $s \geq \frac{\alpha_\tau/2 }{\alpha_\tau+d} = \frac{\gamma/2}{\gamma+1}$, the last term in \eqref{eq:forster-missing} is the bias term that comes from the pseudo-outcome, which is smaller than that of the oracle minimax rate of estimation of $n^{-\alpha_\tau/(2 \alpha_\tau + d)}$, in which case, the FW-Learner is oracle efficient. 

This method using split data has valid theoretical properties under minimal conditions and is similar to Algorithm \ref{alg:split-mar} for missing outcome described in Section \ref{sec-app:missing} of the supplement, and cross-fitting can be applied as discussed before in Section \ref{sec:pseudo}. We also provide an alternative methodology that builds upon the split data method. It uses the full data for both training and estimation, which is potentially more efficient by avoiding sample splitting. The procedure is similar to what we described in Algorithm \ref{alg:split-mar} and is deferred to Algorithm \ref{alg:cate-full-unconfoundedness} in the supplementary material.

\cite{kennedy2023towards} and \cite{kennedy2022minimax} studied the problem of estimating CATE under ignorability quite extensively--the latter paper derived the minimax rate for CATE estimation where distributional components are H\"{o}lder-smooth, along with a new local polynomial estimator that is minimax optimal under some conditions. In comparison, our procedure is not necessarily minimax optimal in some regimes considered there, with the advantage that it is more general with minimum constraints on the basis functions.

\begin{rem}
    Note that although Theorem \ref{thm:forster-cate} and Corollary \ref{cor:missing-ignorability}
     continue to hold for  modified CATE which marginalizes over some confounders, and therefore conditions on a subset of measured confounders, say  $\mathbb{E}\left(Y^1-Y^0 \mid V=v\right)$ where $V$ is a subset of covariates in $X$, with the error bound of Corollary modified so that the second term of the bound \eqref{eq:final-forster-ignorability} is replaced with $J^{-\alpha_{\tau_v}/d_v}$, where $\alpha_{\tau_v}/d_v$ is the effective smoothness of the modified CATE. The application given in Section \ref{sec:simulations} illustrates our methods for such marginalized CATE function which is particularly well-motivated from a scientific perspective. 
\end{rem}

\subsection{Conditional Quantile Causal Effect} 
Suppose $F(Y|X,A)$ is differentiable on the support of $Y$; we consider
construction of the pseudo-outcome for the conditional quantile causal
effect under unconfoundedness 
\begin{equation*}
\beta \left( X;\eta \right) =F_{Y|A,X}^{-1}\left( q|A=1,X\right)
-F_{Y|A,X}^{-1}\left( q|A=0,X\right) 
\end{equation*}%
In which case 
\begin{equation*}
\psi =\E_{X}\left\{ F_{Y|A,X}^{-1}\left( q|A=1,X\right) -F_{Y|A,X}^{-1}\left(
q|A=0,X\right) \right\} 
\end{equation*}%
and the EIF\ of the latter is given by 
\begin{equation*}
\nabla _{t}\psi _{t}=\nabla _{t} \E_{X,t}\left\{ F_{Y|A,X,t}^{-1}\left(
q|A=1,X\right) -F_{Y|A,X,t}^{-1}\left( q|A=0,X\right) \right\} 
\end{equation*}%
This requires finding $R(O;\eta )$ such that 
\begin{equation*}
\E_{O|X}\left\{ \nabla _{t}F_{Y|A,X,t}^{-1}\left( q|A=1,X\right)
-F_{Y|A,X,t}^{-1}\left( q|A=0,X\right) |X;\eta \right\} =\E\left\{ R(O;\eta
)S(O)|X;\eta \right\} 
\end{equation*}%
note that for $\theta _{t}\left( X\right) =F_{Y|A,X,t}^{-1}\left(
q|A=1,X\right) ,$ 
\begin{eqnarray*}
\nabla _{t}q &=&\nabla _{t}\int_{0}^{\theta _{t}\left( X\right) }f_{t}\left(
y|A=1,X\right)  \\
&=&\nabla _{t}\theta _{t}\left( X\right) f_{t}\left( \theta _{t}\left(
X\right) |A=1,X\right) +\E \left\{ I\left( Y<\theta _{t}\left( X\right)
-q\right) S\left( Y|A=1,X\right) \right\} 
\end{eqnarray*}%
\begin{eqnarray*}
&&\nabla _{t}F_{Y|A,X,t}^{-1}\left( q|A=1,X\right)  \\
&=&\frac{\nabla _{t}F_{Y|A,X,t}\left( q|A=1,X\right) }{f\left(
F_{Y|A,X,t}^{-1}\left( q|A=1,X\right) |A=1,X\right) } \\
&=&-\frac{\E\left\{ \left( I\left( Y\leq F_{Y|A,X,t}^{-1}\left(
q|A=1,X\right) \right) -q\right) S(Y|A=1,X)|X;\eta \right\} }{f\left(
F_{Y|A,X,t}^{-1}\left( q|A=1,X\right) |A=1,X\right) } \\
&=& -\E\left\{ \frac{I(A=1)\left( I\left( Y\leq F_{Y|A,X,t}^{-1}\left(
q|A,X\right) \right) -q\right) }{f(A|X)f\left( F_{Y|A,X,t}^{-1}\left(
q|A,X\right) |A,X\right) }S(Y|A,X)|X;\eta \right\}  \\
&=& -\E\left\{ \frac{I(A=1)\left( I\left( Y\leq F_{Y|A,X,t}^{-1}\left(
q|A,X\right) \right) -q\right) }{f(A|X)f\left( F_{Y|A,X,t}^{-1}\left(
q|A,X\right) |A,X\right) }S(O)|X;\eta \right\} 
\end{eqnarray*}%
Likewise 
\begin{eqnarray*}
&&\nabla _{t}F_{Y|A,X,t}^{-1}\left( q|A=0,X\right)  \\
&=&\frac{\nabla _{t}F_{Y|A,X,t}\left( q|A=0,X\right) }{f\left(
F_{Y|A,X,t}^{-1}\left( q|A=0,X\right) |A=0,X\right) } \\
&=&\frac{\E\left\{ \left( I\left( Y\leq F_{Y|A,X,t}^{-1}\left( q|A=0,X\right)
\right) -q\right) S(Y|A=0,X)|X;\eta \right\} }{f\left(
F_{Y|A,X,t}^{-1}\left( q|A=0,X\right) |A=0,X\right) } \\
&=&\E\left\{ \frac{I(A=0)\left( I\left( Y\leq F_{Y|A,X,t}^{-1}\left(
q|A,X\right) \right) -q\right) }{f(A|X)f\left( F_{Y|A,X,t}^{-1}\left(
q|A,X\right) |A,X\right) }S(Y|A,X)|X;\eta \right\}  \\
&=&\E\left\{ \frac{I(A=0)\left( I\left( Y\leq F_{Y|A,X,t}^{-1}\left(
q|A,X\right) \right) -q\right) }{f(A|X)f\left( F_{Y|A,X,t}^{-1}\left(
q|A,X\right) |A,X\right) }S(O)|X;\eta \right\} 
\end{eqnarray*}%
Therefore 
\begin{equation*}
\nabla _{t}\psi _{t}=\E\left[ \left\{ 
\begin{array}{c}
-\frac{I(A=1)\left( I\left( Y\leq F_{Y|A,X,t}^{-1}\left( q|A,X\right)
\right) -q\right) }{f(A|X)f\left( F_{Y|A,X,t}^{-1}\left( q|A,X\right)
|A,X\right) } \\ 
+\frac{I(A=0)\left( I\left( Y\leq F_{Y|A,X,t}^{-1}\left( q|A,X\right)
\right) -q\right) }{f(A|X)f\left( F_{Y|A,X,t}^{-1}\left( q|A,X\right)
|A,X\right) } 
\end{array}%
\right\} S\left( O\right) \right] 
\end{equation*}%
and the pseudo-outcome is given by 
\begin{eqnarray*}
I &=&R(O;\eta )+r(O;\eta ) \\
&=&\frac{I(A=1)\left( I\left( Y\leq F_{Y|A,X,t}^{-1}\left( q|A,X\right)
\right) -q\right) }{f(A|X)f\left( F_{Y|A,X,t}^{-1}\left( q|A,X\right)
|A,X\right) } \\
&&-\frac{I(A=0)\left( I\left( Y\leq F_{Y|A,X,t}^{-1}\left( q|A,X\right)
\right) -q\right) }{f(A|X)f\left( F_{Y|A,X,t}^{-1}\left( q|A,X\right)
|A,X\right) } \\
&&+F_{Y|A,X}^{-1}\left( q|A=1,X\right) -F_{Y|A,X}^{-1}\left( q|A=0,X\right) 
\end{eqnarray*}

\subsection{CATE in Generalized Linear Model}\label{sec:CATE-glm}
Consider the CATE\ in a GLM\ of the form 
\begin{equation*}
\beta \left( X;\eta \right) =g^{-1}\left\{ \E\left( Y|A=1,X\right) \right\}
-g^{-1}\left\{ \E\left( Y|A=0,X\right) \right\} 
\end{equation*}%
for known link function $g.$ In which case 
\begin{equation*}
\psi =\E_{X}\left\{ g^{-1}\left\{ \E\left( Y|A=1,X\right) \right\}
-g^{-1}\left\{ \E\left( Y|A=0,X\right) \right\} \right\}, 
\end{equation*}%
and the EIF\ of the latter is given by 
\begin{equation*}
\nabla _{t}\psi _{t}=\nabla _{t}\E_{X,t}\left\{ g^{-1}\left\{ \E_{t}\left(
Y|A=1,X\right) \right\} -g^{-1}\left\{ \E_{t}\left( Y|A=0,X\right) \right\}
\right\}. 
\end{equation*}%
This requires finding $R(O;\eta )$ such that 
\begin{equation*}
\E_{O|X}\left\{ \nabla _{t}g^{-1}\left\{ \E_{t}\left( Y|A=1,X\right) \right\}
-g^{-1}\left\{ \E_{t}\left( Y|A=0,X\right) \right\} |X;\eta \right\}
 = \E\left\{ R(O;\eta )S(O|X)|X;\eta \right\}. 
\end{equation*}%
We have that 
\begin{eqnarray*}
&&\E_{O|X}\left\{ \nabla _{t}g^{-1}\left\{ \E_{t}\left( Y|A=1,X\right)
\right\} -g^{-1}\left\{ \E_{t}\left( Y|A=0,X\right) \right\} |X;\eta \right\} 
\\
&=&\E_{O|X}\left\{ \frac{\nabla _{t}\left\{ \E_{t}\left( Y|A=1,X\right)
\right\} }{g^{\prime }\left\{ g^{-1}\left\{ \E_{t}\left( Y|A=1,X\right)
\right\} \right\} }-\frac{\nabla _{t}\left\{ \E_{t}\left( Y|A=0,X\right)
\right\} }{g^{\prime }\left\{ g^{-1}\left\{ \E_{t}\left( Y|A=0,X\right)
\right\} \right\} }|X;\eta \right\}  \\
&=&\E_{O|X}\left\{ \left[ \frac{I\left( A=1\right) \left\{ Y -\E\left(
Y|A,X\right) \right\} }{f\left( A|X\right) g^{\prime }\left\{ g^{-1}\left\{
\E\left( Y|A,X\right) \right\} \right\} }-\frac{I\left( A=0\right) \left\{
Y -\E\left( Y|A,X\right) \right\} }{f\left( A|X\right) g^{\prime }\left\{
g^{-1}\left\{ \E\left( Y|A,X\right) \right\} \right\} }\right] S(O|X)|X;\eta
\right\}. 
\end{eqnarray*}%
Therefore, 
\begin{eqnarray*}
I &=&R(O;\eta )+r(O;\eta ) \\
&=&\frac{\left( -1\right) ^{1-A}\left\{ Y -\E\left( Y|A,X\right) \right\} }{%
f\left( A|X\right) g^{\prime }\left\{ g^{-1}\left\{ \E\left( Y|A,X\right)
\right\} \right\} }+g^{-1}\left\{ \E\left( Y|A=1,X\right) \right\}
-g^{-1}\left\{ \E\left( Y|A=0,X\right) \right\}, 
\end{eqnarray*}%
which in the case of identity link recovers the CATE pseudo-outcome. In the
case of log link $g^{\prime }\left( \cdot \right) =g\left( \cdot \right)
=\exp \left( \cdot \right) $, therefore 
\begin{equation*}
I=\frac{\left( -1\right) ^{1-A}\left\{ Y -\E\left( Y|A,X\right) \right\} }{%
f\left( A|X\right) \E\left( Y|A,X\right) }+\log \left\{ \E\left(
Y|A=1,X\right) \right\} -\log \left\{ \E\left( Y|A=0,X\right) \right\}.
\end{equation*}%
\  Likewise, consider the GLM with the logit link for binary $Y$ 
\begin{equation*}
\psi = \E_{X}\left[\text{logit}\E\left( Y|A=1,X\right) -\text{logit}\E\left(
Y|A=0,X\right) \right], 
\end{equation*}%
and the EIF\ of the latter is given by 
\begin{equation*}
\nabla _{t}\psi _{t}=\nabla _{t}\E_{X,t}\left\{ \underset{\beta \left( X;\eta
\right) }{\underbrace{\text{logit}\left\{ \E_{t}\left( Y|A=1,X\right)
\right\} -\text{logit}\left\{ \E_{t}\left( Y|A=0,X\right) \right\} }}\right\}.
\end{equation*}%
Note that $g(b)=\exp \left( b\right) /\left( 1+\exp (b)\right) $ and $g^{\prime }(b)=\exp \left( b\right) /\left( 1+\exp (b)\right)^{2}$, $g^{-1}\left( p\right) =\log ({p}/{(1-p)}).$ Therefore 
\begin{eqnarray*}
I &=&R(O;\eta )+r(O;\eta ) \\
&=&\frac{\left( -1\right) ^{1-A}\left\{ Y -\E\left( Y|A,X\right) \right\} }{%
f\left( A|X\right) g^{\prime }\left\{ g^{-1}\left\{ \E\left( Y|A,X\right)
\right\} \right\} }+g^{-1}\left\{ \E\left( Y|A=1,X\right) \right\}
-g^{-1}\left\{ \E\left( Y|A=0,X\right) \right\}  \\
&=&\frac{\left( -1\right) ^{1-A}\left\{ Y -\E\left( Y|A,X\right) \right\} }{%
f\left( A|X\right) \PP \left( Y=1|A,X\right) \left( 1-\PP \left(
Y=1|A,X\right) \right) }+\log \left\{ \frac{\PP \left( Y=1|A=1,X\right) }{%
\PP \left( Y=0|A=1,X\right) }\right\}\\ &\quad-&\log \left\{ \frac{\PP \left(
Y=1|A=0,X\right) }{\PP\left( Y=0|A=0,X\right) }\right\}  \\
&=&\frac{\left( -1\right) ^{A+Y}}{f(Y,A|X)}+\log \left\{ \frac{\PP \left(
Y=1|A=1,X\right) }{\PP \left( Y=0|A=1,X\right) }\right\} -\log \left\{ \frac{%
\PP \left( Y=1|A=0,X\right) }{\PP \left( Y=0|A=0,X\right) }\right\}.
\end{eqnarray*}
The leading term above was obtained \cite{tchetgen2010doubly} as an
influence function in a semiparametric odds ratio model.  

\subsection{Dose Response with Continuous Treatment (No unmeasured confounding)}
Consider the case of continuous treatment $A$ where we aim to estimate
the dose response curve $E\left( Y_{a}\right) $ \bigskip under
unconfoundedness of $A$ given $L$%
\begin{equation*}
\beta \left( a\right)  = \E\left( Y_{a}\right) = \E_{L}\left\{ \E\left(
Y|A=a,L\right) \right\}. 
\end{equation*}%
As the outer-expectation can be estimated nonparametrically at rate root-n,
its uncertainty is negligible relative to that of $\E\left( Y|A=a,L\right) $
and therefore we may consider the semiparametric model where $f(L)\ $is
known, in which case under an arbitrary corresponding submodel :  
\begin{eqnarray*}
\nabla _{t}\E_{t}\left\{ r(O;\eta _{t})|A=a\right\}  &=&\nabla _{t}r(O;\eta
_{t})=\nabla _{t}\sum_{l}\E_{t}\left( Y|A=a,l\right) f\left( l\right)  \\
&=&\sum_{l}\nabla _{t}\E_{t}\left( Y|A=a,l\right) f\left( l\right)  \\
&=&\sum_{l}\E\left( \left\{ Y -\E\left( Y|A=a,l\right) \right\} S\left(
Y|A,l\right) |A=a,l\right) f\left( l\right)  \\
&=&\sum_{l}\E\left( \left\{ Y -\E\left( Y|A=a,l\right) \right\} \frac{f\left(
l\right) }{f\left( l|A=a\right) }S\left( Y|A,l\right) |A=a,l\right) f\left(
l|A=a\right)  \\
&=&\E\left[ \left\{ Y -\E\left( Y|A,L\right) \right\} \frac{f\left( A\right) }{%
f\left( A|L\right) }S\left( O\right) |A=a\right]. 
\end{eqnarray*}%
Therefore, 
\begin{eqnarray*}
I &=&R(O;\eta )+r(O;\eta ) \\
&=&\left\{ Y -\E\left( Y|A,L\right) \right\} \frac{f\left( A\right) }{f\left(
A|L\right) }+\sum_{l}\E\left( Y|A,l\right) f\left( l\right), 
\end{eqnarray*}%
recovering the pseudo-outcome of \cite{kennedy2017non}.
\subsection{Dose Response for Continuous Treatment (unmeasured Confounding)}
Consider the case of continuous treatment $A$ where we aim to estimate
the dose response curve $\E\left( Y_{a}\right) $ \bigskip under endogeneity,
given treatment and outcome proxies $Z,W$ and covariates $L,$ using the
proximal causal inference framework of \cite{miao2018identifying} and \cite{tchetgen2020introduction},   
\begin{equation*}
\beta \left( a\right) = \E\left( Y_{a}\right) = \E_{L}\left\{ h\left(
W,a,L\right) \right\}. 
\end{equation*}%
As the outer-expectation can be estimated nonparametrically at rate root-n,
its uncertainty is negligible relative to that of $h\left( W,a,L\right) $,
and therefore we may consider the semiparametric model where $f(w,L)\ $is
known. Thus,  taking 
\begin{equation*}
r(O;\eta )=\sum_{w,l}h\left( w,a,l\right) f\left( w,l\right), 
\end{equation*}
in which case under an arbitrary corresponding submodel: 
\begin{eqnarray*}
\nabla _{t}\E_{t}\left\{ r(O;\eta _{t})|A=a\right\}  &=&\nabla _{t}r(O;\eta
_{t})=\nabla _{t}\sum_{w,l}h\left( w,a,l\right) f\left( w,l\right)  \\
&=&\sum_{l,w}\nabla _{t}h_{t}\left( w,a,l\right) f\left( l,w\right)  \\
&=&\sum_{l,w}\nabla _{t}h_{t}\left( w,a,l\right) \frac{f\left( l,w\right) }{%
f\left( l,w|a\right) }f\left( l,w|a\right)  \\
&=&\sum_{l,w}\nabla _{t}h_{t}\left( w,a,l\right) \frac{f\left( a\right) }{%
f\left( a|l,w\right) }f\left( l,w|a\right)  \\
&=&\sum_{l,w}\nabla _{t}h_{t}\left( w,a,l\right) \E\left[ q\left(
a,Z,l\right) |l,w,a\right] f\left( l,w|a\right)  \\
&=&\sum_{l,w,z}\nabla _{t}h_{t}\left( w,a,l\right) q\left( a,z,l\right)
f\left( z,l,w|a\right)  \\
&=&\sum_{l,z}\E\left[ \nabla _{t}h_{t}\left( W,a,l\right) |a,z,l\right]
q\left( a,z,l\right) f\left( z,l|a\right)  \\
&=&\sum_{l,z}\E\left[ \left\{ Y-h\left( W,a,l\right) \right\} S\left(
Y,W|a,z,l\right) |a,z,l\right] q\left( a,z,l\right) f\left( z,l|a\right)  \\
&=&\E\left[ \left\{ Y-h\left( W,A,L\right) \right\} q\left( A,Z,L\right)
S\left( Y,W|A,Z,L\right) |A=a\right]  \\
&=&\E\left[ \left\{ Y-h\left( W,A,L\right) \right\} q\left( A,Z,L\right)
S\left( O\right) |A=a\right]. 
\end{eqnarray*}%
Therefore, 
\begin{eqnarray*}
I &=&R(O;\eta )+r(O;\eta ) \\
&=&\left\{ Y-h\left( W,A,L\right) \right\} q\left( A,Z,L\right)
+\sum_{w,l}h\left( w,A,l\right) f\left( w,l\right), 
\end{eqnarray*}%
generalizing the pseudo-outcome approach of \cite{kennedy2017non} to the
Proximal inference framework with continuous treatment.

\subsection{CATE under IV Identification}
In this example, we consider the CATE under IV\ identification. In this
vein, let $A$ denote a binary treatment, $Z$ denote a binary instrumental
variable, $L$ measured covariates, $Y$ the outcome variable. Under
identification conditions given in \cite{wang2018bounded}, we have that 
\begin{equation*}
\beta \left( X;\eta \right)  = \E\left\{ Y_{a=1}-Y_{a=0}|L\right\} =\frac{%
\E\left( Y|Z=1,L\right)  -\E\left( Y|Z=0,L\right) }{\E\left( A|Z=1,L\right)
 -\E\left( A|Z=0,L\right) }.
\end{equation*}%
Let 
\begin{eqnarray*}
\delta _{A}\left( L\right)  &\equiv &\E\left( A|Z=1,L\right)  -\E\left(
A|Z=0,L\right);  \\
r(O;\eta ) &=&\frac{\E\left( Y|Z=1,L\right)  -\E\left( Y|Z=0,L\right) }{\E\left(
A|Z=1,L\right)  -\E\left( A|Z=0,L\right) }.
\end{eqnarray*}%
Then following \cite{wang2018bounded}, one has that 
\begin{eqnarray*}
\nabla _{t}\E_{t}\left\{ r(O;\eta _{t})|X=x\right\}  &=&\nabla _{t}r(O;\eta
_{t}) \\
&=&\E\left\{ R(O;\eta ,\beta \left( \eta \right) )S\left( O\right)
|X=x\right\}, 
\end{eqnarray*}%
where 
\begin{equation*}
R(O;\eta ,\beta \left( \eta \right) =\frac{2Z-1}{f\left( Z|X\right) }\frac{%
\left\{ Y-A\beta \left( X;\eta \right)  -\E\left( Y|Z=0,L\right) +\E\left(
A|Z=0,L\right) \beta \left( X;\eta \right) \right\} }{\delta _{A}\left(
L\right) }.
\end{equation*}%
Therefore, 
\begin{eqnarray*}
I &=&R(O;\eta )+r(O;\eta ) \\
&=&\frac{2Z-1}{f\left( Z|X\right) }\frac{\left\{ Y-A\beta \left( X;\eta
\right)  -\E\left( Y|Z=0,L\right) +E\left( A|Z=0,L\right) \beta \left( X;\eta
\right) \right\} }{\delta _{A}\left( L\right) }+\beta \left( X;\eta \right). 
\end{eqnarray*}

\subsection{CATE under IV Identification 2}
We next consider the CATE for the Complier under IV identification. In this
vein, under identification conditions given in \cite{angrist1996identification}, we have that 
\begin{equation*}
\beta \left( X;\eta \right) = \E\left\{ Y_{a=1}-Y_{a=0}|A_{1}>A_{0},L\right\} =%
\frac{\E\left( Y|Z=1,L\right)  -\E\left( Y|Z=0,L\right) }{\E\left(
A|Z=1,L\right)  -\E\left( A|Z=0,L\right) }.
\end{equation*}%
in which case the above results continue to hold.  Likewise, under
identification conditions given by \cite{robins1994estimation}, the CATE for the treated is
given by the same formula
\begin{equation*}
\beta \left( X;\eta \right)  = \E\left\{ Y_{a=1}-Y_{a=0}|A=1,L\right\} =\frac{
\E\left( Y|Z=1,L\right)  -\E\left( Y|Z=0,L\right) }{\E\left( A|Z=1,L\right)
 -\E\left( A|Z=0,L\right) }.
\end{equation*}

\section{Proof of Theorem \ref{thm:forster-full} and Corollary \ref{cor:forster-pseudo}}
\label{sec-app:forster-full}
\begin{proof}[Proof of Theorem \ref{thm:forster-full}]
Theorem $6.2$\footnote{See also Appendix E of \cite{vavskevivcius2023suboptimality} for a proof of that theorem.} of \cite{forster2002relative} implies that 
$$
\begin{aligned}
\mathbb{E}\Bigl[\bigl(Y-\widehat{m}_J(X)\bigr)^2\Bigr] & \leq \inf _{\beta_1, \ldots, \beta_J} \mathbb{E}\Bigl[\bigl(Y-\sum_{j=1}^J \beta_j \phi_j(X)\bigr)^2\Bigr]+\frac{2 \sigma^2 J}{n} \\
& = \mathbb{E}[(Y - m^\star(X))^2] + \inf_{\beta_1, \ldots, \beta_J}\mathbb{E}\Bigl[\bigl(m^\star(X)-\sum_{j=1}^J \beta_j \phi_j(X)\bigr)^2\Bigr]+\frac{2 \sigma^2 J}{n},
\end{aligned}
$$
where  $\sigma^2$ is an upper bound on $\E \bigl[Y^2 | X\bigr]$. To control the second term above, observe that if $f_X(\cdot)$ is the density of $X$ with respect to $\mu$, then for any $(\beta_1, \ldots, \beta_J)$,
\[
\mathbb{E}\Bigl[\bigl(m^\star(X)-\sum_{j=1}^J \beta_j \phi_j(X)\bigr)^2\Bigr] = \int \bigl(m^\star(x)-\sum_{j=1}^J \beta_j \phi_j(x)\bigr)^2f_X(x)d\mu(x) \le \kappa\left\|m^\star - \sum_{j=1}^J \beta_j\phi_j\right\|_{L_2(\mu)}^2.
\]
Hence, the infimum of the left-hand side over all $\beta_1, \ldots, \beta_J$ is bounded by $\kappa (E_J^{\Psi}(m^\star))^2$. 
Finally, note that $m^\star(\cdot)$ being the conditional mean of $Y$ given $X$ implies $\E \bigl[(Y-\widehat{m}_J(X))^2\bigr]-\mathbb{E}\bigl[(Y-m^{\star}(X))^2\bigr]=\mathbb{E}\bigl[(\widehat{m}_J(X)-m^\star(X))^2\bigr]$. Therefore,
$$
\E \Bigl[\bigl(\widehat{m}_J(X)-m^\star(X) \bigr)^2\Bigr] \leq \frac{2 \sigma^2 J}{n}+\kappa(E_J^{\Psi}(m^\star))^2.
$$

To prove the second part of Theorem~\ref{thm:forster-full}, note that $m^\star\in\mathcal{F}(\Psi, \Gamma)$ implies that $E_J^{\Psi}(m^\star) \le \gamma_J$ and by definition of $J_n$, $\gamma_{J_n}^2 \le \sigma^2J_n/n$. These inequalities imply that
\[
\|\widehat{m}_{J_n} - m^\star\|_2^2 \le \frac{2\sigma^2J_n}{n} + \kappa\gamma_{J_n} \le \frac{2\sigma^2J_n}{n} + \kappa\frac{\sigma^2J_n}{n} = (2 + \kappa)\frac{\sigma^2J_n}{n}.
\]
\end{proof}

\begin{proof}[Proof of Corollary \ref{cor:forster-pseudo}]
Applying the result for Forster--Warmuth estimator (from Theorem $6.2$ \cite{forster2002relative}), this gives
\begin{equation}\label{eq:Them-6.2-application-Pseudo-outcome}
\mathbb{E}\Bigl[\bigl(\widehat{f}(O) - \widehat{m}_J(X)\bigr)^2\big|\widehat{f}\Bigr] \le \inf_{\theta\in\mathbb{R}^J}\mathbb{E}\Bigl[\bigl(\widehat{f}(O) - \theta^{\top}\bar{\phi}_J(X)\Bigr)^2\big|\widehat{f}\Bigr] + \frac{2\sigma^2J}{|\mathcal{I}_2|},
\end{equation}
where $\sigma^2 = \sup_{x}\mathbb{E}[\widehat{f}^2(O)|X = x, \widehat{f}]$. 
Now write $m^\star(x) = \sum_{j=1}^{\infty}\theta_j^\star\phi_j(x),$ and taking $(\theta_1^\star, \ldots, \theta_J^\star)$ for the infimum, we conclude
\begin{align*}
\inf_{\theta\in\mathbb{R}^J}\mathbb{E}\Bigl[\bigl(\widehat{f}(O) - \theta^{\top}\bar{\phi}_J(X)\bigr)^2\big|\widehat{f}\Bigr]
&\le \mathbb{E}\Bigl[\bigl(\widehat{f}(O) - \sum_{j=1}^J \theta_j^\star\phi_j(X)\bigr)^2\big|\widehat{f}\Bigr]\\
&= \mathbb{E}\Bigl[\bigl(\widehat{f}(O) - H_f(X)\bigr)^2\big|\widehat{f}\Bigr] + \mathbb{E}\Bigl[\bigl(H_f(X) - \sum_{j=1}^J \theta_j^\star\phi_j(X)\bigr)^2\big|\widehat{f}\Bigr]\\
&\le \mathbb{E}\Bigl[\bigl(\widehat{f}(O) - H_f(X)\bigr)^2\big|\widehat{f}\Bigr] + 2\mathbb{E}\Bigl[\bigl(H_f(X) - m^\star(X)\bigr)^2\big|\widehat{f}\Bigr]\\
&\quad+ 2\mathbb{E}\Bigl[\bigl(m^\star(X) - \sum_{j=1}^J \theta_j^\star\phi_j(X)\bigr)^2\Bigr].
\end{align*}
Substituting this inequality in~\eqref{eq:Them-6.2-application-Pseudo-outcome} yields
\begin{align*}
    \mathbb{E}\Bigl[\bigl(\widehat{f}(O) - \widehat{m}_J(X)\bigr)^2\big|\widehat{f}\Bigr] &\le \mathbb{E}\Bigl[\bigl(\widehat{f}(O) - H_f(X)\bigr)^2\big|\widehat{f}\Bigr] + 2\mathbb{E}\Bigl[\bigl(H_f(X) - m^\star(X)\bigr)^2\big|\widehat{f}\Bigr]\\
    &\qquad+ 2\mathbb{E}\Bigl[\bigl(m^\star(X) - \sum_{j=1}^J \theta_j^\star\phi_j(X)\bigr)^2\Bigr] + \frac{2\sigma^2J}{|\mathcal{I}_2|}. 
\end{align*}
Because $H_f(x) = \mathbb{E}[\widehat{f}(O)|X = x, \widehat{f}]$ and the density of $X$ with respect to $\mu$ is bounded by $\kappa$, this yields
\begin{align*}
\mathbb{E}\Bigl[\bigl(\widehat{m}_J(X) - H_f(X)\bigr)^2\big|\widehat{f}\Bigr] &\le 2\mathbb{E}\Bigl[\bigl(H_f(X) - m^\star(X)\bigr)^2\big|\widehat{f}\Bigr]
    + 2\kappa (E_J^{\Psi}(m^\star))^2 + \frac{2\sigma^2J}{|\mathcal{I}_2|}.
\end{align*}
Therefore,
\begin{align*}
\|\widehat{m}_J - m^\star\|_{2|\widehat{f}} &\le \|\widehat{m}_J - H_f\|_{2|\widehat{f}} + \|H_f - m^\star\|_{2|\widehat{f}}\\
&\le \|H_f - m^\star\|_{2|\widehat{f}} + \sqrt{2}\|H_f - m^\star\|_{2|\widehat{f}} + \sqrt{2\kappa}E_J^{\Psi}(m^\star) + \sqrt{\frac{2\sigma^2J}{|\mathcal{I}_2|}}\\
&= \sqrt{\frac{2\sigma^2J}{|\mathcal{I}_2|}} + \sqrt{2\kappa}E_J^{\Psi}(m^\star) + (1 + \sqrt{2})\|H_f - m^\star\|_{2|\widehat{f}}.
\end{align*}
Here, for any function $h$, we use the notation $\|h\|_{2|\widehat{f}} = (\mathbb{E}[h^2(X)|\widehat{f}])^{1/2}.$ Because $1 + \sqrt{2} \le \sqrt{6},$ the result is proved.
\end{proof}

\begin{proof}[Proof of \eqref{eq:thm1-fourier}]\label{pf:thm1-fourier}
    For notational convenience, set $f(x) = ax + bx^{-c}$. For our case, $a = 2\sigma^2/n, b = C_m\kappa,$ and $c = 2\alpha_m/d$. ($x$ is a proxy for $J$, but note $J$ is an integer.) The minimizer of $f$ is $(cb/a)^{1/(c + 1)}$, which may or may not be an integer and we choose $J = \lceil x^*\rceil$, where $x^* = (cb/a)^{1/(c + 1)}$. Clearly, $x^*/2 \le J \le 2x^*$. Therefore, 
\begin{align*}
f(J) &\le 2ax^* + b(x^*/2)^{-c}\\
&= 2a(cb/a)^{1/(c + 1)} + \frac{b}{2^c(cb/a)^{c/(c + 1)}}\\
&= 2a^{c/(c + 1)}b^{1/(c + 1)}c^{1/(c + 1)}\left[1 + \frac{1}{2^{c + 1}c}\right] \\
&\le 3a^{c/(c + 1)}b^{1/(c + 1)}\left[1 + \frac{1}{2c}\right],
\end{align*}
because $x^{1/(x + 1)} \le 1.5$ for all $x > 0$. Now substituting $a, b, c$ and simplifying the bound gives us
\begin{align}
   \Bigl\|\widehat{m}_J(X) - m^\star(X) \Bigr\|_2 ^2 &= f(J) \leq 3 (\frac{2\sigma^2}{n})^{\frac{2\alpha_m}{2\alpha_m+d}} (C_m \kappa )^{\frac{d}{2\alpha_m+d}}[1+\frac{d}{4\alpha_m}]\\
   &\leq C \left(\frac{\sigma^2}{n}\right)^{2 \alpha_m /(2 \alpha_m+d)}, 
\end{align}
where
$C = 6(C_m\kappa/2)^{d/(2\alpha_m + d)}(1 + d/(4\alpha_m)).$

\end{proof}

\section{Proof of Theorem \ref{thm:pseudo-outcome-construction}}\label{sec:proof-thm2}
\begin{proof}
To prove the first result, note that for all submodels $\eta _{t}$ in $%
\mathcal{M}$, 
\begin{equation*}
\mathbb{E}_{\eta _{t}}\left\{ R(O;\eta _{t},n^{\ast }\left( \eta _{t}\right)
)+r\left( O;\eta _{t}\right) -n^{\ast }\left( x;\eta _{t}\right)
|X=x\right\} =0
\end{equation*}%
for all $\eta $ and $x$. Therefore%
\begin{equation*}
\frac{\partial }{\partial t}\mathbb{E}_{\eta _{t}}\left\{ R(O;\eta
_{t},n^{\ast }\left( \eta _{t}\right) )+r\left( O;\eta _{t}\right) -n^{\ast
}\left( x;\eta _{t}\right) |X=x\right\} =0,
\end{equation*}%
which implies that 
\begin{equation}\label{eq:first-derivative-expansion}
\begin{split}
\mathbb{E}\left\{ R(O;\eta ,n^{\ast })S\left( O|X\right) |X=x\right\}+\mathbb{E}\left\{ r(O;\eta
)S\left( O|X\right) |X=x\right\}  \\ 
+\frac{\partial }{\partial t}\mathbb{E}\left\{R(O;\eta _{t},n^{\ast }\left( \eta
_{t}\right) )+r\left( O;\eta _{t}\right) |X=x%
\right\} -\frac{\partial n^{\ast }\left( x;\eta _{t}\right) }{\partial t}=0.
\end{split}
\end{equation}%
Further note that by assumption: 
\begin{eqnarray*}
\frac{\partial n^{\ast }\left( x;\eta _{t}\right) }{\partial t} &=&\mathbb{E}%
\left\{ r(O;\eta )S\left( O|X\right) |X=x\right\} +\mathbb{E}\left[ \frac{\partial
r\left( O;\eta _{t}\right) }{\partial t}|X=x\right]  \\
&=&\mathbb{E}\left\{ r(O;\eta )S\left( O|X\right) |X=x\right\} +\mathbb{E}%
\left[ R(O;\eta ,n^{\ast }\left( \eta \right) )S\left( O|X\right) |X=x\right].
\end{eqnarray*}%
This combined with~\eqref{eq:first-derivative-expansion}, we get%
\begin{equation*}
\frac{\partial }{\partial t}\mathbb{E}\left\{ R(O;\eta _{t},n^{\ast }\left(
\eta _{t}\right) )+r\left( O;\eta _{t}\right) |X=x\right\} =0,
\end{equation*}%
from which we may conclude via a Taylor expansion at $\eta $, that%
\begin{eqnarray*}
&&\left\Vert \mathbb{E}\left[ R(O;\eta ^{\prime },n^{\ast }\left( \eta
^{\prime }\right) )+r\left( O;\eta ^{\prime }\right) |X\right] -n^{\ast
}\left( X;\eta \right) \right\Vert _{2}
=O\left( \left\Vert \eta ^{\prime }-\eta \right\Vert ^{2}\right),
\end{eqnarray*}%
as $\mathbb{E}\left[ R(O;\eta ^{\prime },n^{\ast }\left( \eta ^{\prime
}\right) )+r\left( O;\eta ^{\prime }\right) |X\right] -n^{\ast }\left(
X;\eta \right) =0$ at $\eta ^{\prime }=\eta $. This
proving the first result. To prove the second result, consider the
functional 
\begin{equation*}
\psi  = \E\left\{ r(O;\eta )\right\} =E_{X}\left[ \E_{O|X}\left\{ r(O;\eta
)|X;\eta \right\} \right] =\E_{X}\left[ \beta \left( X;\eta \right) \right],
\end{equation*}
under a semiparametric model $\mathcal{M}$, where $\eta $ is an infinite
dimensional parameter indexing the law of $O$ conditional on $X$. Then if $%
\psi $ is pathwise differentiable on $\mathcal{M}$, an influence function of 
$\psi $ can be obtained by pathwise differentiation as follows
\begin{eqnarray*}
\frac{\partial \psi \left( \eta _{t}\right) }{\partial t} &=&\frac{\partial
E_{t}\left[ r(O;\eta _{t})\right] }{\partial t} \\
&=&\E\left[ r(O;\eta )S\left( O\right) \right] +\E\left\{ \frac{\partial r(O;\eta _{t})}{\partial t}\right\}  \\
&=&\E\left[ r(O;\eta )S\left( O\right) \right]
+\E\left\{ \E\left[ \left. \frac{\partial r(O;\eta _{t})}{\partial t}%
\right\vert X\right] \right\}  \\
&=&\E\left[ r(O;\eta )S\left( O\right) \right]  +\E\left\{ \E\left[ R(O;\eta ,n^{\ast }\left( \eta \right) )S\left(
O|X\right) |X\right] \right\}  \\
&=&\E\left[ \left[ r(O;\eta )-\psi \left( \eta \right) \right] S\left(
O\right) \right] +\E\left\{ \E\left[ R(O;\eta ,n^{\ast }\left( \eta \right)
)S\left( O\right) \right] \right\}.
\end{eqnarray*}
This completes the proof of the second result.
\end{proof}
\section{Examples of Nonparametric Estimators}\label{app-sec:nonparametric}
A common approach in nonparametric regression literature is to suppose that the regression function is $\beta$-smooth. The
minimax estimation rate on the mean-squared error scale is then as indicated above, of the order of $n^{-2\beta/(2\beta+d)}$ \citep{stone1982optimal}.  
which may be excessively large due to the curse of dimensionality in practical settings where $d$ is itself large. This can have a detrimental impact both on one's ability to estimate the nuisance functions sufficiently well for the oracle rate to apply.  
To address this concern alternative smoothness function classes may also be considered particularly in settings where $d$ is large. For instance,  \cite{schmidt2020nonparametric} considers functions that can be parametrized using large neural networks with a number of potential network parameters exceeding the sample size and shows that estimators based on fine-tuned sparsely connected deep neural network achieve the minimax rates of convergence under a general composition framework on the regression function. The multilayer neural networks can adapt to specific structures in the signal and achieves faster rates under a hierarchical composition assumption including (generalized) additive models. Specifically, let $f_0$ denote the regression function of interest and assume that it is a composition of several (denoted as $q$) functions, that is
$f_0=g_q \circ g_{q-1} \circ \ldots \circ g_1 \circ g_0,$
where $g_i:\mathbb{R}^{d_i}\rightarrow \mathbb{R}^{d_{i+1}}$ with $d_0=d$ and $d_{q+1}=1$.  Note here that non-identifiability of the single components $g_0, \dots, g_q$ is not necessarily a problem because out of all possible representations, one would in practice select a representation that leads to the fastest possible estimation rate for $f_0$. Assuming that each of the functions $g_{ij}$ has H\"{o}lder smoothness $\beta_i$, the convergence rate of the network estimator $\widehat{f}_n$ is 
$R(\widehat{f}_n, f_0):= \E [(\widehat{f}_n - f_0 )^2] \asymp \phi_n \log^3 n$ under certain conditions for the composite regression function class, where $\phi_n:= \max_{i=0, \dots, q} n^{- 2 \beta_i^*/(2 \beta_i^* + t_i)}$, the effective smoothness index $\beta_i^*:=\beta_i \prod_{\ell=i+1}^q (\beta_{\ell} \wedge 1 )$ and $t_i$ is the maximal number of variables on which each component of $g_i$ depends on, which, under specific constraints such as additive models, will be much smaller than $d_i$. 
Alternatively, \cite{haris2019nonparametric} tackles high dimensional non-parametric regression by using a penalized estimation framework that is well-suited for high-dimensional sparse additive models. Specifically, they proposed a penalized estimation method motivated by the projection estimator that may be used to fit additive models of specific form $f_0 = \sum_{j=1}^d f_j (x_j)$. It attains the minimax optimal rates $O(n^{-\frac{2\alpha_m}{2\alpha_m+ d}})$ under standard smoothness assumptions in the univariate case and in the sparse additive case where $s$ is the sparsity (the number of non-zero $f_j$), it attains the rate  $O\bigl\{\max \bigl(s n^{-\frac{2 \alpha_m}{2 \alpha_m+d}}, \frac{s \log d}{n}\bigr)\bigr\}$ under a suitable compatibility condition. Even without the compatibility condition, it may still be consistent with convergence rate $O\bigl\{\max \bigl(s n^{-\frac{\alpha_m}{2 \alpha_m+d}}, s \sqrt{\frac{\log d}{n}}\bigr)\bigr\}$. \cite{kohler2021rate}  provides analogous results on the approximation of smooth functions and models with hierarchical composition structures by fully connected deep neural networks.
\section{Other Proofs and Results}\label{app-sec:proof}
\begin{proof}[Proof of Proposition \ref{prop:fw-bound}]
By the definition of $\widehat{m}_J(x)$ in \eqref{def:forster}, it holds that
    \begin{align}
    \bigl| \widehat{m}_J(x) \bigr| &= \bigl( 1-h_n(x) \bigr) \biggl| \bar{\phi}_J^\top (x)  \Bigl(\sum_{i=1}^n \bar{\phi}_J (X_i ) \bar{\phi}_J^{\top}(X_i)+\bar{\phi}_J(x) \bar{\phi}_J^{\top}(x)\Bigr)^{-1/2} \cdot \\
    &\qquad \qquad \qquad \Bigl(\sum_{i=1}^n \bar{\phi}_J (X_i ) \bar{\phi}_J^{\top}(X_i)+\bar{\phi}_J(x) \bar{\phi}_J^{\top}(x)\Bigr)^{-1/2} \sum_{i=1}^n \bar{\phi}_J(X_i) Y_i \biggr|\\
    &\leq \bigl( 1-h_n(x) \bigr) \biggl( \bar{\phi}_J^\top (x)  \Bigl(\sum_{i=1}^n \bar{\phi}_J (X_i ) \bar{\phi}_J^{\top}(X_i)+\bar{\phi}_J(x) \bar{\phi}_J^{\top}(x)\Bigr)^{-1} \bar{\phi}_J (x)  \biggr)^{1/2} \cdot\\
    &\qquad \qquad \qquad \biggl( \bigl( \sum_{i=1}^n \bar{\phi}_J(X_i) Y_i \bigr)^\top \Bigl(\sum_{i=1}^n \bar{\phi}_J (X_i ) \bar{\phi}_J^{\top}(X_i)+\bar{\phi}_J(x) \bar{\phi}_J^{\top}(x)\Bigr)^{-1} \sum_{i=1}^n \bar{\phi}_J(X_i) Y_i \biggr)^{1/2 } \\
    &= \bigl( 1-h_n(x) \bigr)\, h_n(x) \,\biggl( \bigl( \sum_{i=1}^n \bar{\phi}_J(X_i) Y_i \bigr)^\top \Bigl(\sum_{i=1}^n \bar{\phi}_J (X_i ) \bar{\phi}_J^{\top}(X_i)+\bar{\phi}_J(x) \bar{\phi}_J^{\top}(x)\Bigr)^{-1} \sum_{i=1}^n \bar{\phi}_J(X_i) Y_i \biggr)^{1/2 }\\
    &\leq \bigl( 1-h_n(x) \bigr)\, h_n(x) \,\biggl( \bigl( \sum_{i=1}^n \bar{\phi}_J(X_i) Y_i \bigr)^\top \Bigl(\sum_{i=1}^n \bar{\phi}_J (X_i ) \bar{\phi}_J^{\top}(X_i) \Bigr)^{-1} \sum_{i=1}^n \bar{\phi}_J(X_i) Y_i \biggr)^{1/2 }\\
    & =  \bigl( 1-h_n(x) \bigr)\, h_n(x) \, \bigl\| \widehat{\Sigma}_J^{1/2} \widehat{\beta}\bigr\|_2, \label{eq:fw_bound-middle} 
\end{align}
where $\widehat{\Sigma}_J$ and $\widehat{\beta}$ denote $\sum_{i=1}^n \bar{\phi}_J (X_i ) \bar{\phi}_J^{\top}(X_i)/n$ and $ \widehat{\Sigma}_J^{-1} \sum_{i=1}^n \bar{\phi}_J(X_i) Y_i/n$ respectively.
Lastly, because
\begin{align}
    \bigl\| \widehat{\Sigma}_J^{1/2} \widehat{\beta}\bigr\|_2 &= \bigl\| \frac{1}{n}\sum_{i=1}^n \widehat{\Sigma}_J^{-1/2}   \bar{\phi}_J(X_i) Y_i  \bigr\|_2\\
    &=\sup_{a \in S^{d-1}} \bigl| \frac{1}{n} \sum_{i=1}^n a^\top \widehat{\Sigma}_J^{-1/2} \bar{\phi}_J(X_i) Y_i  \bigr|\\
    &\leq \sup_{a \in S^{d-1}} \biggr( \frac{1}{n} \sum_{i=1}^n \bigl( a^\top \widehat{\Sigma}_J^{-1/2}  \bar{\phi}_J(X_i) \bigr)^2 \biggr)^{1/2} \, \Bigl(\frac{1}{n} \sum_{i=1}^n Y_i^2 \Bigr)^{1/2}  \\  
    &= \sup_{a \in S^{d-1}} \biggr( \frac{1}{n} \sum_{i=1}^n \bigl( a^\top \widehat{\Sigma}_J^{-1/2}  \bar{\phi}_J(X_i) \bar{\phi}_J^\top (X_i) \widehat{\Sigma}_J^{-1/2} a \biggr)^{1/2} \, \Bigl(\frac{1}{n} \sum_{i=1}^n Y_i^2 \Bigr)^{1/2}  \\  
    &=  \Bigl(\frac{1}{n} \sum_{i=1}^n Y_i^2 \Bigr)^{1/2},\label{eq:fw_bound2} 
\end{align}
combining \eqref{eq:fw_bound-middle} and \eqref{eq:fw_bound2} yields the desired result in \eqref{eq:fw-bound}.
\end{proof}

\begin{proof}[Proof of \eqref{eq:dr-bias-general}]
Proposition 1 of \cite{ghassami2022minimax} gives that $\E [ \mathrm{IF}_{\psi}(O; q^\star, h^\star)] = \E [ \mathrm{IF}_{\psi}(O; q^\star, \widehat{h}]$. Therefore, this and the construction of the pseudo-outcome gives
 \begin{align}\label{eq:proof-dr-general-1}
     &\quad \,\,  H_f(X) - m^\star(x) \\
     &=\E \Bigl[ \mathrm{IF}_{\psi}(O; \widehat{q}, \widehat{h}) - \mathrm{IF}_{\psi}(O; q^\star, h^\star) |X, \widehat{q}, \widehat{h} \Bigr]\\
     &=\E \Bigl[ \mathrm{IF}_{\psi}(O; \widehat{q}, \widehat{h}) - \mathrm{IF}_{\psi}(O; q^\star, \widehat{h}) |X ,\widehat{q}, \widehat{h}\Bigr]\\
     &=\E \Bigl[ \bigl( \widehat{h}(O_h) g_1(O)+ g_2(O) \big) (\widehat{q} - q^\star ) (O_q) | X, \widehat{q}, \widehat{h}\Bigr]\\
     &= \E \Bigl[ \E \bigl[ \bigl( \widehat{h}(O_h) g_1(O)+ g_2(O) \big) | O_q \bigr] (\widehat{q} - q^\star) (O_q) | X,\widehat{q}, \widehat{h}\Bigr].
\end{align}  
It can be shown in the same way as in the proof of Proposition 1 of \cite{ghassami2022minimax} that
$\mathbb{E}[h^\star(O_h) g_1(O)+g_2(O) \mid O_q]=0$; see the discussion following Eq. (13) there. Therefore, for all functions $h$, it holds that
$$
\mathbb{E} \bigl[h(O_h) g_1(O)+g_2(O) | O_q \bigr]=\mathbb{E}\bigl[g_1(O)(h - h^\star)(O_h) | O_q\bigr].
$$
Applying this equality to \eqref{eq:proof-dr-general-1} yields
\begin{align}
    &\quad \,\, \E \Bigl[ \mathrm{IF}_{\psi}(O; q^\star, h^\star) - \mathrm{IF}_{\psi}(O; \widehat{q}, \widehat{h})|X,\widehat{q}, \widehat{h} \Bigr]\\
    &= \E \Bigl[ \E \bigl[ \bigl( g_1(O)(h^\star -\widehat{h}) (O_h) \big) | O_q \bigr] (\widehat{q} - q^\star ) (O_q) | X,\widehat{q}, \widehat{h}\Bigr]\\
    &=\E \Bigl[   g_1(O)(h^\star -\widehat{h}) (O_h) (\widehat{q} - q^\star) (O_q) | X,\widehat{q}, \widehat{h}\Bigr].
\end{align}
\end{proof}

\begin{definition}\label{def:if} Given a semiparametric model $\mathcal{F}$, a law ${F}^{*}$ in $\mathcal{F}$, and a class $\mathcal{A}$ of regular parametric submodels of $\mathcal{F}$, a real valued functional
$$
\theta: \mathcal{F} \rightarrow \mathbb{R}
$$
is said to be a {\bf pathwise differentiable} or regular parameter at F* wrt $\mathcal{A}$ in model $\mathcal{F}$ iff there exists $\psi_{{F}^ *}({x})$ in $\mathcal{L}_{2}(\dot{F}^{*})$ such that for each submodel in $\mathcal{A}$, say indexed by $t$ and with ${F}^{*}={F}_{t^{*}}$, and score, say $S_{t}\left(t^{*}\right)=s_{t}\left(X ; t^{*}\right)$ at $t^{*}$, it holds that
$$
\left.\frac{\partial}{\partial t} \theta\left(F_{t}\right)\right|_{t=t^{*}}=\E_{F^{*}}\left[\psi_{F^{*}}(X) S_{t}\left(t^{*}\right)\right]
$$
$\psi_{{F}^*}(.)$ is called a {\bf{gradient}} of $\theta$ at ${F}^{*}$ (wrt $\left.\mathcal{A}\right)$.
If, in addition, $\psi_{{F}^*}( {X})$ has mean zero under $\mathrm{F}^{*}, \psi_{{F}^*}( {X})$ is most commonly referred to as an {\bf{influence function}} of the functional $\theta$ at $F^{*}$.
\end{definition}
\section{Some Results for Missing Outcome}\label{sec-app:missing}
\subsection{FW-Learner Algorithm for Missing Outcome}\label{sec:algorithm}
The procedure for the FW-Learner in the case of  missing outcome under MAR using split data and cross-validation is described in \ref{alg:split-mar}.
\begin{algorithm}[h]
    \SetAlgoLined
    \SetEndCharOfAlgoLine{}
   \KwIn{Training data $\mathcal{D}^{\text{tr}}=(X_i, Z_i, R_i,Y_iR_i), i = 1, \dots, N$; basis function $\phi(\cdot)$, estimators $\widehat{\pi}, \widehat{\mu}$ and the point for estimation $x$, a grid of tuning parameters for number of basis  $J_{\text{grid}}$ and a hyper-parameter $K$ (for Cross validation purpose).}
    \KwOut{An estimator for $m^\star(x) = \E[Y|X=x]$, denoted as $\widehat{m}(x)$.}
    Split training data $\mathcal{D}^{\text{tr}}$ randomly into $\mathcal{D}_1$ and $\mathcal{D}_2$, where $\mathcal{D}_1 = \{Z_i \in \mathcal{D}^{\text{tr}}, i \in \mathcal{I}_1\}$ and $\mathcal{D}_2 = \{Z_i \in \mathcal{D}^{\text{tr}}, i \in \mathcal{I}_2\}$.\;
    Fit estimators $\widehat{\pi},\widehat{\mu}$ on $\mathcal{D}_1$ and for each $i\in\mathcal{I}_2$, define pseudo-outcomes $\widehat{I}_i = \frac{Y_iR_i}{\widehat{\pi}(X_i, Z_i)} - \bigl(\frac{R_i}{\widehat{\pi}(X_i, Z_i)} - 1\bigr)\widehat{\mu}(X_i, Z_i)$.\;
    For each $k = 1, \dots, K$, further split $\mathcal{I}_2$ into two parts, $\mathcal{I}_{21}$ and $\mathcal{I}_{22}$, and for each $J \in J_{\text{grid}}$, fit the Forster--Warmuth estimator according to \eqref{def:forster} on $\bigl\{\bigl(\bar{\phi}_J(X_i), \widehat{I}_i\bigr),  i \in  \mathcal{I}_{21} \bigr\}$, where  $\bar{\phi}_J(x)=\bigl(\phi_1(x), \ldots, \phi_J(x)\bigr)^{\top}$. This is denoted by $\widehat{m}_{{J}}^{(k)}$. Use the rest of the data to choose a parameter $\widehat{J}_k$ to minimize the test error such that
    $$ \widehat{J}_k := \operatorname{argmin}\displaylimits_{J\in J_\text{grid}} \,\frac{1}{|\mathcal{I}_{22}|} \sum_{i \in \mathcal{I}_{22}}  \bigl| \widehat{m}_J(X_i) - Y_i \bigr|^2.$$\;
    Repeat the above step $K$ times and obtain $\widehat{m}(x) := \frac{1}{K}\sum_{k=1}^K \widehat{m}^{(k)}_{\widehat{J}_k }(x)$.\;
    \Return the estimation result $\widehat{m} (x)$.
    \caption{The FW-Learner for missing outcome under MAR with CV using split data}
    \label{alg:split-mar}
\end{algorithm}

\subsection{Proof of Lemma \ref{lem:forster-missing} and Theorem \ref{thm:forster-missing}}
\label{sec-app:forster-missing}
\begin{proof}[Proof of Lemma \ref{lem:forster-missing}]
Because
\begin{equation}\label{eq:forster-missing-impute}
\begin{split}
\widehat{f}(O) &= \frac{R}{\widehat{\pi}(X, Z)}(YR) - \left(\frac{R}{\widehat{\pi}(X, Z)} - 1\right)\widehat{\mu}(X, Z),\\
&= \frac{R}{\widehat{\pi}(X, Z)}Y - \left(\frac{R}{\widehat{\pi}(X, Z)} - 1\right)\widehat{\mu}(X, Z),
\end{split}
\end{equation}
taking the expectation on both sides of \eqref{eq:forster-missing-impute} conditional on $X,Z$ yields that
\begin{align}
\E \bigl\{ \mathbb{E}[\widehat{f}(O)|X, Z] \bigm| X \bigr\} &= \frac{\pi^\star(X, Z)}{\widehat{\pi}(X, Z)}\mu^\star(X, Z) - \left(\frac{\pi^\star(X, Z)}{\widehat{\pi}(X, Z)} - 1\right)\widehat{\mu}(X, Z).
\label{eq:forster-missing-bias}
\end{align}
Furthermore, because $\E[\mu^\star(X, Z)|X] = m^\star (X)$, this gives
\begin{align}
\mathbb{E}[\widehat{f}(O)|X=x] - m^\star(x)&=
\E \Bigl\{ \mathbb{E}[\widehat{I}_1(YR, R, X, Z)|X, Z] - \mu^\star(X, Z) \bigm| X=x\Bigr\}\\
&= \E \Bigl\{ \left(\frac{\pi^\star(X, Z)}{\widehat{\pi}(X, Z)} - 1\right) \bigl(\mu^\star(X, Z) - \widehat{\mu}(X, Z) \bigr)  \bigm| X=x\Bigr\}. 
\label{eq:forster-missing-bias2}
\end{align}
\end{proof}
\begin{proof}[Proof of Theorem \ref{thm:forster-missing}]
Taking the square on both sides of \eqref{eq:forster-missing-bias2} gives
\begin{align*}
\Bigl[ H_{I_1}(X) - m^\star(X) \Bigr]^2 &= \Bigl[ \E \Bigl\{ \Bigl(\frac{\pi^\star(X, Z)}{\widehat{\pi}(X, Z)} - 1 \Bigr) \bigl(\mu^\star(X, Z) - \widehat{\mu}(X, Z) \bigr) \bigm| X \Bigr\} \Bigr]^2\\
& \leq \E \Bigl\{ \Bigl(\frac{\pi^\star(X, Z)}{\widehat{\pi}(X, Z)} - 1 \Bigr)^2 \bigm| X \Bigr\} \E \Bigl\{ \bigl(\mu^\star(X, Z) - \widehat{\mu}(X, Z) \bigr)^2 \bigm| X \Bigr\}.
\end{align*}
Taking the expectation on both sides and applying the Cauchy--Schwarz inequality yields, 
\begin{equation}\label{eq:mu-H_1-bound}
\Bigl\|H_{I_1}(X) - m^\star (X) \Bigr\|_2 \le \Bigl\|\frac{\pi^\star(X, Z)}{\widehat{\pi}(X, Z)} - 1\Bigr\|_4 ~ \Bigl\|\mu^\star(X, Z) - \widehat{\mu}(X, Z) \Bigr\|_4.
\end{equation}
Substituting this into the last term of \eqref{eq:oracle-missing} gives that
\begin{align}
\Bigl\|\widehat{m}_J(X) - m^\star(X)\Bigr\|_2 &\le \sqrt{\frac{2\sigma^2J}{|\mathcal{I}_2|}} + \sqrt{2}\Bigl\|\sum_{j=J+1}^{\infty} \theta^\star_j\phi_j(X)\Bigr\|_2\\ &\quad+ (1+\sqrt{2})\Bigl\|\frac{\pi^\star(X, Z)}{\widehat{\pi}(X, Z)} - 1\Bigl\|_4  ~ \Bigl\|\mu^\star(X, Z) - \widehat{\mu}(X, Z)\Bigr\|_4\\
&\lesssim \sqrt{\frac{\sigma^2 J}{n}} + J^{-\alpha_m}  + n ^{-\frac{\alpha_\pi}{2\alpha_\pi+d} - \frac{\alpha_\mu}{2\alpha_\mu+d}}.
\end{align}
And this concludes our proof.
\end{proof}

\subsection{Proof of Lemma \ref{lem:forster-missing-MNAR} and Theorem \ref{thm:forster-missing-MNAR}}
\label{sec-app:forster-missing-MNAR}
\begin{proof}[Proof of Lemma \ref{lem:forster-missing-MNAR}]
Because
\begin{equation}\label{eq:forster-missing-bias-nmar}
\begin{split}
\widehat{f} (O) &= \frac{R}{\widehat{e}(X, Y)}Y - \left(\frac{R}{\widehat{e}(X, Y)} - 1\right)\widehat{\eta}(X, W).
\end{split}
\end{equation}
Because $R \perp W |(X,Y)$, taking the expectation of the above display conditional on $X,Y$ yields that
\begin{equation}
\begin{split}
\mathbb{E}[\widehat{f}(O)|X, Y] &= \frac{e^\star(X, Y)}{\widehat{e}(X, Y)}Y - \left(\frac{e^\star(X, Y)}{\widehat{e}(X, Y)} - 1\right)\E \bigl[ \widehat{\eta}(X, W) |X,Y \bigr].\\
\end{split}
\end{equation}
Therefore, 
\begin{align}
\mathbb{E}[\widehat{f} (O)|X, Y] - Y&=\biggl(\frac{e^\star(X, Y)}{\widehat{e}(X, Y)} - 1\biggr) \Bigl(Y - \E \bigl[ \widehat{\eta}(X, W) |X,Y \bigr] \Bigr)\\
&= \biggl(\frac{e^\star(X, Y)}{\widehat{e}(X, Y)} - 1\biggr)  \E \Bigl[ (\eta^\star - \widehat{\eta})(X, W) |X,Y \Bigr]. \label{eq:I1-bias-nmar-middle}
\end{align}
Taking the expectation on both sides of \eqref{eq:I1-bias-nmar-middle} conditional on $X$ gives the desired result.
\end{proof}

\begin{proof}[Proof of Theorem \ref{thm:forster-missing-MNAR}]
Because \eqref{eq:I1-bias-nmar-middle} gives
\begin{align}
    \E \Bigl[ H_{f}(X) - m^\star(X) \Bigr] &=\E \biggl\{  \biggl(\frac{e^\star(X, Y)}{\widehat{e}(X, Y)} - 1\biggr)  \E \Bigl[ (\eta^\star - \widehat{\eta})(X, W) |X,Y \Bigr] \Bigm| X \biggr\}\\
    &=\E \biggl\{ (\eta^\star - \widehat{\eta})(X, W)    \E \Bigl[ \biggl(\frac{e^\star(X, Y)}{\widehat{e}(X, Y)} - 1\biggr) |X,W \Bigr] \Bigm| X \biggr\}.
    \label{eq:two-forms}
\end{align}
Therefore, 
\begin{align*}
 \E \Bigl[ H_{f}(X) - m^\star(X) \Bigr]^2 & =  \E_{X} \Bigl[ \E\bigl( \{  \widehat{f}(O)-Y \} | X  \bigr)  \Bigr]^2\\
 & \leq \E_{X,Y} \Bigl[ \E\bigl( \{  \widehat{f}(O)-Y \} | X,Y  \bigr)  \Bigr]^2\\
 & = \E \Bigl[ \E \bigl\{  \widehat{f}(O) | X,Y \bigr\} - Y  \Bigr]^2\\
 & = \E \Biggl\{ \biggl(\frac{e^\star(X, Y)}{\widehat{e}(X, Y)} - 1\biggr)^2 \E \Bigl[ (\eta^\star - \widehat{\eta})(X, W) |X,Y \Bigr]^2 \Biggr\} \text{ from } \eqref{eq:I1-bias-nmar-middle} \\
 &\leq \biggl\{ \biggl\|\frac{e^\star(X, Y)}{\widehat{e}(X, Y)} - 1 \biggr\|_4  ~ \biggl\| \E \Bigl[ (\eta^\star - \widehat{\eta})(X, W) |X,Y \Bigr] \biggr\|_4  \biggr\}^2,
\end{align*}
where the last inequality is from the Cauchy--Schwarz inequality. Similarly, the second inequality can be replaced so that the outer expectation is taken w.r.t $(X,W)$, i.e.
\begin{align}
 \E \Bigl[ H_{f}(X) - m^\star(X) \Bigr]^2 & =  \E_{X} \Bigl[ \E\bigl( \{  \widehat{f}(O)-Y \} | X  \bigr)  \Bigr]^2\\
 & \leq \E_{X,Y} \Bigl[ \E\bigl( \{  \widehat{f}(O)-Y \} | X,W \bigr)  \Bigr]^2.   
\end{align}
Using \eqref{eq:two-forms} and then plug the minimum of these two outcomes into \eqref{eq:oracle-missing} gives the desired result.

\end{proof}

\section{Some Results for Estimating the CATE}
\subsection{FW-Learner Algorithm for estimating the CATE without data splitting}
\begin{algorithm}[h]
    \SetAlgoLined
    \SetEndCharOfAlgoLine{}
   \KwIn{Training data $\mathcal{D}^{\text{tr}} = (X_i, A_i, Y_i), i = 1, \dots, N$; a basis function $\phi(\cdot)$ and number of basis to use $J$, estimators $\widehat{\pi}, \widehat{\mu}_0, \widehat{\mu}_1$ and the point for estimation $x$.}
    \KwOut{An estimator for the CATE $\tau^\star$.}
    Fit estimators $\widehat{\pi},\widehat{m}$ on $\mathcal{D}^{\text{tr}}$ and for each $i = 1, \dots, N$, define pseudo-outcomes $\widehat{I}_i = \frac{A_i - \widehat{\pi}(X_i)}{\widehat{\pi}(X_i)(1 - \widehat{\pi}(X_i))}(Y_i - \widehat{\mu}_{A_i}(X_i)) + \widehat{\mu}_1(X_i) - \widehat{\mu}_0(X_i)$.\;
    Fit the Forster--Warmuth regression at estimation point $x$ according to \eqref{def:forster} on $\left(\bar{\phi}_J\left(X_i\right), \widehat{I}_i\right),  i = 1, \dots, N$ where $\bar{\phi}_J(x)=\bigl(\phi_1(x), \ldots, \phi_J(x)\bigr)^{\top}$.\;
    \Return the estimation result $\widehat{\tau}_J(x)$.
    \caption{Full FW-Learner for the CATE under strong ignorability}
    \label{alg:cate-full-unconfoundedness}
\end{algorithm}

\subsection{Proof under ignorability--Lemma \ref{lem:forster-cate} and Theorem \ref{thm:forster-cate}}
\label{sec-app:forster-cate}
\begin{proof}[Proof of Lemma \ref{lem:forster-cate}]
\begin{align}
    H_{I_1}(X)-\tau^\star(X) &= \E \biggl\{ \frac{A - \widehat{\pi}(X)}{\widehat{\pi}(X)(1 - \widehat{\pi}(X))}\bigl( Y - \widehat{\mu}_{A}(X) \bigr) \bigg| X \biggr\} + \widehat{\tau}(X)- \tau^\star(X) \nonumber\\
    &= \E \biggl\{ \frac{A - \widehat{\pi}(X)}{\widehat{\pi}(X)(1 - \widehat{\pi}(X))}\bigl( \E( Y|A,X) - \widehat{\mu}_{0}(X) - A\widehat{\tau}(X) \bigr) \bigg| X \biggr\} + \widehat{\tau}(X)- \tau^\star(X) \nonumber\\
    &= \E \biggl\{ \frac{A - \widehat{\pi}(X)}{\widehat{\pi}(X)(1 - \widehat{\pi}(X))}\bigl( \E( Y|A=0,X)+A\tau^\star - \widehat{\mu}_{0}(X) - A\widehat{\tau}(X) \bigr)\bigg| X \biggr\} + \widehat{\tau}(X)- \tau^\star(X) \nonumber\\
    &= \frac{\pi^\star(X) - \widehat{\pi}(X)}{\widehat{\pi}(X)(1-\widehat{\pi}(X))} \bigl(\mu^\star_0(X) - \widehat{\mu}_0(X) \bigr) + \biggl(1-\frac{\pi^\star(X)}{\widehat{\pi}(X)} \biggr) \bigl( \widehat{\tau}(X) -\tau^\star(X) \bigr). \label{eq:cate-bias2}
\end{align}
Re-parameterizing with $\tau^\star = \mu^\star_1 - \mu^\star_0$ yields
\begin{equation}\label{eq:cate-bias-sutva}
H_{I_1}(X) - \tau^\star (X) = \Bigl( \frac{\pi^\star(X)}{\widehat{\pi}(X)} - 1 \Bigr) \Bigl( \mu^\star_1(X) - \widehat{\mu}_1(X) \Bigr) - \Bigl(\frac{1 - \pi^\star(X)}{1 - \widehat{\pi}(X)} - 1\Bigr) \bigl( \mu^\star_0(X) - \widehat{\mu}_0(X) \bigr).
\end{equation}
\end{proof}
\begin{proof}[Proof of Theorem \ref{thm:forster-cate}]
In the following, all expectations and conditional expectations are conditional on the first split of the data.
Corollary \ref{cor:forster-pseudo} implies that
\begin{align}\label{eq:oracle-cate}
\Bigl\|\widehat{\tau}_J(X) - \tau^\star (X)\|_2 \le \sqrt{\frac{2\sigma^2J}{|\mathcal{I}_2|}} + \sqrt{2} \Bigl\|\sum_{j=J+1}^{\infty} \theta_j^\star\phi_j(X)\Bigr\|_2 +  (1+\sqrt{2}) \Bigl\|H_{I_1}(X) - \tau^\star (X) \Bigr\|_2.
\end{align}

Lemma \ref{lem:forster-cate} and the Cauchy--Schwarz inequality gives us
\[
\bigl\| H_{I_1}(X) - \tau^\star (X) \bigr\|_2 \le \Bigl\|\frac{\pi^\star(X)}{\widehat{\pi}(X)} - 1 \Bigr\|_4 \bigl\|\widehat{\mu}_1(X) - \mu^\star_1(X) \bigr\|_4 + \Bigl\|\frac{1 - \pi^\star(X)}{1 - \widehat{\pi}(X)} - 1\Bigr\|_4 \bigl\|\widehat{\mu}_0(X) - \mu^\star_0(X) \bigr\|_4.
\]
Plugging this into the oracle inequality \eqref{eq:oracle-cate} yields
\begin{align*}
    \bigl\|\widehat{\tau}_J(X) - \tau^\star (X) \bigr\|_2 &\le \sqrt{\frac{2\sigma^2J}{|\mathcal{I}_2|}} + \sqrt{2}\Bigl\|\sum_{j=J+1}^{\infty} \theta_j^\star\phi_j(X) \Bigr\|_2\\ 
    &\quad+ 2\Bigl\|\frac{\pi^\star(X)}{\widehat{\pi}(X)} - 1 \Bigr\|_4 \Bigl\|\widehat{\mu}_1(X) - \mu^\star_1(X) \Bigr\|_4 + \Bigl\|\frac{1 - \pi^\star(X)}{1 - \widehat{\pi}(X)} - 1\Bigr\|_4 \bigl\|\widehat{\mu}_0(X) - \mu^\star_0(X) \bigr\|_4.
\end{align*}
\end{proof}

\subsection{Proof of Lemma \ref{lem:forster-proximal-cate} and Theorem \ref{thm:forster-cate-proximal}}\label{sup-sec:proximal-cate}
\begin{proof}[Proof of Lemma \ref{lem:forster-proximal-cate}]
\begin{align}
    H_{I}(X) - \tau^\star (X)
    &= \E \Bigl[ \bigl\{A \widehat{q}_1 - (1-A)\widehat{q}_0 \bigr\} \{ Y - \widehat{h}(W, A, X)\} +\widehat{h}_1 - \widehat{h}_0 \bigm|X \Bigr] - \tau^\star \\
    &=\E \Bigl[ \bigl\{A \widehat{q}_1 - (1-A)\widehat{q}_0 \bigr\} \{ Y - \widehat{h}(W, A, X)\} \bigm|X \Bigr] + \E \Bigl[\widehat{h}_1 - \widehat{h}_0 - \tau^\star  \bigm|X \Bigr]. \label{eq:proximal-middle-1}
\end{align}
The first term of \eqref{eq:proximal-middle-1} amounts to
\begin{align}
     \quad ~&\E \biggl[ \bigl\{A \widehat{q}_1- (1-A)\widehat{q}_0 \bigr\} \{ Y - \widehat{h}(W, A, X)\} \biggm|X \biggr] \\
=&\E \biggl[ \E\Bigl\{ \bigl\{A \widehat{q}_1- (1-A)\widehat{q}_0 \bigr\} \{ Y - \widehat{h}(W, A, X)\} \Bigm| Z,A,X \Bigr\} \biggm|X \biggr]\\  \stackrel{\eqref{eq:proximal-h}}{=} &\E \biggl[ \E\Bigl\{ \bigl\{A \widehat{q}_1- (1-A)\widehat{q}_0 \bigr\} \{ (h^\star - \widehat{h})(W, A, X)\} \Bigm| Z,A,X \Bigr\} \biggm|X \biggr]\\
    =&\E \biggl[ \E\Bigl\{ \E \Bigl[ \bigl\{A \widehat{q}_1- (1-A)\widehat{q}_0 \bigr\} \{ (h^\star - \widehat{h})(W, A, X)\} \Bigm| W,X \Bigr] \Bigm| Z,A,X \Bigr\} \biggm|X \biggr]\\
    =&\E \biggl[ \E\Bigl\{ \E \Bigl[ \bigl\{A \widehat{q}_1- (1-A)\widehat{q}_0 \bigr\} \{ (h^\star - \widehat{h})(W, A, X)\} \Bigm| W,X \Bigr] \Bigm| X \Bigr\} \biggm|X \biggr]\\
    =&\E \biggl[  \E \Bigl[ \bigl\{A \widehat{q}_1- (1-A)\widehat{q}_0 \bigr\} \{ (h - \widehat{h})(W, A, X)\} \Bigm| W,X \Bigr]  \biggm|X \biggr]\\
    =&\E \biggl[ \E \Bigl[ A \widehat{q}(Z,1,X) \{ (h^\star - \widehat{h})(W, A, X)\} \Bigm| W,X \Bigr] \biggm|X \biggr]\\
    &\qquad - \E \biggl[  \E \Bigl[ \bigl\{ (1-A)\widehat{q}(Z,0,X) \bigr\} \{ (h^\star - \widehat{h})(W, A, X)\} \Bigm| W,X \Bigr]  \biggm|X \biggr]. \label{eq:proximal-middle-2}
\end{align}
Therefore, 
\begin{align}
    \quad ~H_{I}(X) - \tau^\star (X) &=\E \biggl[ \bigl[ A (h^\star - \widehat{h})(W, 1, X)  \E [ \widehat{q}(Z,1,X)   \bigm| W, 1,X ] \\
    &\qquad  - (1-A)(h^\star - \widehat{h})(W, 0, X)  \E [ \widehat{q}(Z,0,X)   | W,0,X ] + \widehat{h}_1 - \widehat{h}_0  \bigr] \Bigm| X \biggr]- \tau^\star \\
    &= \E \biggl[ \bigl[ A(h^\star - \widehat{h})(W, 1, X)\bigr\{ \E [ \widehat{q}(Z,1,X)   | W, 1,X ] - \E [ {q^\star}(Z,1,X)   | W, 1,X ] \bigr\} \\
    & \qquad  - (1-A)(h^\star - \widehat{h})(W, 0, X)\bigr\{ \E [ \widehat{q}(Z,0,X)   | W, 0,X ] - \E [ {q^\star}(Z,0,X)   | W, 0,X ] \bigr\} \Bigm| X \biggr]\\
    & = \E \biggl[ \Bigl\{ A  (h^\star - \widehat{h})(W, 1, x)\bigr( \widehat{q}(Z,1,x)  
 -  q^\star (Z,1,x) \bigr) \\
    & \qquad  - (1-A) (h^\star - \widehat{h})(W, 0, x)\bigr(  \widehat{q}(Z,0,x) - {q^\star}(Z,0,x) \bigr) \Bigr\} \Bigm| X \biggr] ,\label{eq:proximal-middle-3}
\end{align}
where in the second equality we used $\E[Y^{(a)}|X] = \E[h(W,a,X)|X]$, so that $\tau^\star  (x) = \E[h^\star(W,1,X) - h^\star(W,0,X)|X=x]$.
\end{proof}

\begin{proof}[Proof of Theorem \ref{thm:forster-cate-proximal}]
In this proof, all expectations and conditional expectations are conditional on the first split of the data.
Corollary \ref{cor:forster-pseudo} implies that
\begin{align}\label{eq:oracle-cate-proximal}
\Bigl\|\widehat{\tau}_J(X) - \tau^\star (X)\|_2 \le \sqrt{\frac{2\sigma^2J}{|\mathcal{I}_2|}} + \sqrt{2} \Bigl\|\sum_{j=J+1}^{\infty} \theta_j^\star\phi_j(X)\Bigr\|_2 +  (1+\sqrt{2}) \Bigl\|H_{I}(X) - \tau(X) \Bigr\|_2.
\end{align}
Because \eqref{eq:proximal-middle-3} gives that
\begin{align}
    \quad ~H_{I}(X) - \tau^\star (X) 
    & = \E \biggl[ \Bigl[ A (h^\star - \widehat{h})(W, 1, x) \E \Bigl\{ 
     \widehat{q}(Z,1,x)   -  {q^\star}(Z,1,x) | W,X  \Bigr\} \\
    & \qquad  - (1-A)(h^\star - \widehat{h})(W, 0, x) \E \Bigl\{ \widehat{q}(Z,0,x)    q^\star(Z,0,x)    | W,X  \Bigr\}   \Bigr] \Bigm| X \biggr]\\
    & = \E \biggl[ \Bigl[ \bigl\{ \frac{h^\star }{\widehat{h}}(W, 1, x) -1 \bigr\} \E \Bigl\{ 
     \widehat{q}(Z,1,x)   -  {q^\star}(Z,1,x) | W,X  \Bigr\} \\
    & \qquad  - \bigl\{ \frac{h^\star}{\widehat{h}}(W, 0, x) -1 \bigr\} \E \Bigl\{ \widehat{q}(Z,0,x)    q^\star(Z,0,x)    | W,X  \Bigr\}   \Bigr] \Bigm| X \biggr].
    \label{eq:proximal-middle-41}
\end{align}
Similarly,
\begin{align}
    & = \E \biggl[ \Bigl[ A  \bigl\{ \widehat{q}(Z,1,x)  - {q^\star}(Z,1,x)  \bigr\} \E \Bigl\{ (h^\star - \widehat{h})(W, 1, x)
   | Z,X  \Bigr\} \\
    & -   (1-A) \bigl\{ \widehat{q}(Z,0,x)   - {q^\star}(Z,0,x) \bigr\}   \E \Bigl\{ (h^\star - \widehat{h})(W, 0, x) | Z,X  \Bigr\}   \Bigr] \Bigm| X \biggr] \\
    & = \E \biggl[ \Bigl[  \bigl\{ \frac{\widehat{q}(Z,1,x)} {q^\star(Z,1,x)  }  - 1 \bigr\} \E \Bigl\{ (h^\star - \widehat{h})(W, 1, x)
   | Z,X  \Bigr\} \\
    & -    \bigl\{ \frac{\widehat{q}(Z,0,x) } {q^\star(Z,0,x) }   -1 \bigr\}  \E \Bigl\{ (h^\star - \widehat{h})(W, 0, x) | Z,X  \Bigr\}   \Bigr] \Bigm| X \biggr].\label{eq:proximal-middle-42}
\end{align}

Lemma \ref{lem:forster-proximal-cate} gives us
\begin{align}
    \| H_{I}(X) - \tau^\star (X)\|^2_2 &= \E_{X} \Biggl\{ \E^2 \biggl[ \bigl[ (h^\star - \widehat{h})(W, 1, x)\bigr(\frac{ \widehat{q}(Z,1,x)  }{ {q^\star}(Z,1,x)  }-1\bigr) \\
    & \qquad  - (h^\star - \widehat{h})(W, 0, x)\bigr( \frac{ \widehat{q}(Z,0,x)   }{ {q^\star}(Z,0,x)   }-1\bigr) \Bigm| X \biggr]   \Biggr\}\\
    &\leq 2\E_{X} \Biggl\{ \E^2 \biggl[ \bigl[ (h^\star - \widehat{h})(W, 1, X)\bigr(\frac{ \widehat{q}(Z,1,X)  }{ {q^\star}(Z,1,X)  }-1\bigr) \Bigm| X \biggr] \Biggr\} \\
    & \qquad  + 2 \E_{X} \Biggl\{ \E^2 (h^\star - \widehat{h})(W, 0, X)\bigr( \frac{ \widehat{q}(Z,0,X)  }{ {q^\star}(Z,0,x)  }-1\bigr) \Bigm| X \biggr] \Biggr\}\\
    &\leq 2\E_{W,X} \Biggl\{ \E^2 \biggl[ \bigl[ (h^\star - \widehat{h})(W, 1, X)\bigr(\frac{ \widehat{q}(Z,1,X)  }{ {q^\star}(Z,1,X)  }-1\bigr) \Bigm| W, X \biggr] \Biggr\} \\
    & \qquad  + 2 \E_{W,X} \Biggl\{ \E^2 (h^\star - \widehat{h})(W, 0, X)\bigr( \frac{ \widehat{q}(Z,0,X)  }{ {q^\star}(Z,0,x)  }-1\bigr) \Bigm| W, X \biggr] \Biggr\}\\
    &\leq 2 \biggl\{ \Bigl\| \E \bigl[ \frac{ \widehat{q}(Z,1,X)  }{ {q^\star}(Z,1,X)  }-1  \Bigm | W,X  \bigr] \Bigr\|_4 \Bigl\| (h^\star - \widehat{h})(W, 1, X) \Bigr\|_4 \\
    &\qquad +  2\Bigl\| \E \bigl[ \frac{ \widehat{q}(Z,0,X)  }{ {q^\star}(Z,0,X)  }  -1 \Bigm | W,X  \bigr] \Bigr\|_4  \Bigl\| (h^\star - \widehat{h})(W, 0, X) \Bigr\|_4 \biggr\}^2,
\end{align}
where the last inequality is from the Cauchy--Schwarz inequality. Similarly, the third inequality can be written so that the outer layer of expectation is taken w.r.t $(Z,W)$.
Leveraging \eqref{eq:proximal-middle-41} and \eqref{eq:proximal-middle-42} and plugging the minimum of the two outcomes into the oracle inequality \eqref{eq:oracle-cate-proximal} yields
\begin{align*}
    \bigl\|\widehat{\tau}_J(X) - \tau^\star (X) \bigr\|_2 \le &\sqrt{\frac{2\sigma^2J}{|\mathcal{I}_2|}} + \sqrt{2}\Bigl\|\sum_{j=J+1}^{\infty} \theta_j^\star\phi_j(X) \Bigr\|_2\\ 
    + \min\Bigl\{ & 2(1+\sqrt{2})\Bigl\| \frac{ \widehat{q}(Z,1,X)  }{ {q^\star}(Z,1,X)  } -1    \Bigr\|_4\Bigl\| \E\bigl[ ( \widehat{h} - h^\star)(W, 1, X) | Z,X  \bigr] \Bigr\|_4 \\
    &+ 2(1+\sqrt{2}) \Bigl\|  \frac{ \widehat{q}(Z,0,X)  }{ {q^\star}(Z,0,X)  } -1   \Bigr\|_4  \Bigl\|  \E\bigl[ ( \widehat{h} - h^\star)(W, 0, X) | Z,X \bigr] \Bigr\|_4, \\
    & 2(1+\sqrt{2})\Bigl\| \E\bigl[ \frac{ \widehat{q}(Z,1,X)  }{ {q^\star}(Z,1,X)  } -1 \Bigm | W,X  \bigr] \Bigr\|_4\Bigl\| ( \widehat{h} - h^\star)(W, 1, X)  \Bigr\|_4 \\
    &+ 2(1+\sqrt{2}) \Bigl\| \E\bigl[ \frac{ \widehat{q}(Z,0,X)  }{ {q^\star}(Z,0,X)  } -1  \Bigm | W,X  \bigr] \Bigr\|_4  \Bigl\| ( \widehat{h} - h^\star)(W, 0, X) \Bigr\|_4 \Bigr\}.
\end{align*}
\end{proof}

\section{Additional results for estimating bridge functions}\label{sup-sec:bridge}
\subsection{Missing data under MNAR}\label{sup-sec:bridge-mnar}
We state the following minimax lower bound convergence result for the estimation of $\eta^\star$, which is a direct application of Theorem 3.2 of \cite{chen2018optimal}.

The following are some working conditions of \cite{chen2018optimal}, where they gave the minimax lower bound for this estimation problem along with a method that has a matching upper bound.

\textbf{Assumptions for bridge function estimation (bridge)}: (i) Variables $X_i, W_i$ have compact rectangular support $\mathcal{X}, \mathcal{W} \subset \mathbb{R}^{d_x}, \mathbb{R}^{d_w}$ with nonempty interiors and the densities of $X_i, W_i$ are uniformly bounded away from 0 and $\infty$ on $\mathcal{X}, \mathcal{W}$; (ii) $Y_i$ has compact rectangular support $\mathcal{Y} \subset \mathbb{R}^{1}$ and the density of $Y_i$ is uniformly bounded away from 0 and $\infty$ on $\mathcal{Y}$; (iii) $T: L^2(X,W) \rightarrow L^2(X,Y)$ is injective; (iv)There is a positive decreasing function $\nu$ such that $\|T \eta\|_{L^2(X,Y)}^2 \lesssim$ $\sum_{j, G, k}\left[\eta \left(2^j\right)\right]^2\langle \mu, \tilde{\psi}_{j, k, G}\rangle_{X,W}^2$ holds for all $\eta \in B_{\infty}(\alpha_\eta, L)$.

\begin{lem}\label{lem:lb-shadow}
Assume the 4 conditions above hold for the kernel $T$ of the integral equation \eqref{eq:bridge} with a random sample $\left\{\left(X_i, Y_i, W_i\right)\right\}_{i=1}^n$, the following result holds for the optimal rate for estimating
$$
\liminf _{n \rightarrow \infty} \inf _{\widehat{\eta_n}} \sup _{\eta \in B_{\infty}(\alpha_\eta, L)} \mathbb{P}_\eta \left(\left\|\widehat{\eta}_n - \eta\right\|_{\infty} \geq c r_n\right) \geq c^{\prime}>0,
$$
where
$$
r_n=\left[\begin{array}{ll}
(n / \log n)^{-\alpha_\eta /(2(\alpha_\eta+\varsigma_\eta)+d_x + d_w)} & \text { in the mildly ill-posed case, } \\
(\log n)^{-\alpha_\eta / \varsigma_\eta} & \text { in the severely ill-posed case}.
\end{array}\right.
$$
$\inf _{\widehat{\eta_n}}$ denotes the infimum over all estimators of $\eta$ (based on the sample of size $n$), $\sup _{\eta \in B_{\infty}(\alpha_\eta, L)} \mathbb{P}_\eta$ denotes the sup over $\eta \in B_{\infty}(\alpha_\eta, L)$, and distributions of $\left(X_i, Y_i, W_i, u_i\right)$ that satisfy Condition $\mathrm{LB}$ with fixed $\nu$, and the finite positive constants $c$ and $c^{\prime}$ do not depend on $n$.
\end{lem}

Note that we only focus on the mildly ill-posed case.
\cite{chen2018optimal} established that under some conditions this lower bound is tight under the supremum norm where they also provided methods that would attain these rates. In addition, a new paper \cite{chen2021adaptive} proposed a method that would attain this rate while being adaptive to the unknown parameters of the function.

\subsection{CATE under proximal causal inference}\label{sup-sec:bridge-cate}
The following  two lemmas on bridge function estimation for $h^\star$ and $q^\star$ are direct applications of Theorem 3.2 of \cite{chen2018optimal}.
\begin{lem}\label{lem:proximal-h}
Assuming the conditions similar to Lemma \ref{lem:lb-shadow} hold for the kernel $T$ to the integral equation \eqref{eq:proximal-h} with a random sample $\left\{\left(X_i, Y_i, W_i,Z_i\right)\right\}_{i=1}^n$ and that $T$ is mildly ill-posed with  $\tau_h=O\left(J^{\varsigma_h / (d_x+d_w)}\right)$ for some $\varsigma_h>0$. Then  
$$
\liminf _{n \rightarrow \infty} \inf _{\widehat{h}_n} \sup _{h \in B_{\infty}(\alpha_h, L)} \mathbb{P}_h\left(\bigl\|\widehat{h}_n - h^\star\bigr\|_{\infty} \geq c (n / \log n)^{-\alpha_h /(2(\alpha_h+\varsigma_h)+d_x+d_w+1)} \right) \geq c^{\prime}>0,
$$
where
$\inf _{\widehat{h}_n}$ denotes the infimum over all estimators of $h^\star$ based on the sample of size $n$, $\sup_{h \in B_{\infty}(\alpha_h, L)} \mathbb{P}_h$ denotes the sup over $h \in B_{\infty}(\alpha_h, L)$.
\end{lem}
The next result gives the convergence rate for the estimation of $q^\star$.
\begin{lem}\label{lem:proximal-q}
Assume similar conditions for Lemma \ref{lem:lb-shadow} hold for the kernel $T'$ of the integral equation \eqref{eq:proximal-q} with a random sample $\left\{\left(X_i, Y_i, W_i,Z_i\right)\right\}_{i=1}^n$, and that $T'$ is mildly ill-posed with  $\tau_q=O\left(J^{\varsigma_q / (d_x+d_w)}\right)$ for some $\varsigma_q>0$. Then 
$$
\liminf _{n \rightarrow \infty} \inf _{\widehat{q}_n} \sup _{q \in B_{\infty}(\alpha_q, L)} \mathbb{P}_q\Bigl(\bigl\|\widehat{q}_n - q^\star\bigr\|_{\infty} \geq c (n / \log n)^{-\alpha_q /(2(\alpha_q+\varsigma_q)+d_x+d_z)} \Bigr) \geq c^{\prime}>0,
$$
where
$\inf _{\widehat{q}_n}$ denotes the infimum over all estimators of $q^\star$ based on the sample of size $n$, $\sup_{q \in B_{\infty}(\alpha_q, L)} \mathbb{P}_q$ denotes the sup over $q \in B_{\infty}(\alpha_q, L)$.
\end{lem}

\end{document}